\documentclass{article}

\usepackage[final]{neurips_2025}

\usepackage[utf8]{inputenc} 
\usepackage[T1]{fontenc}    
\usepackage{hyperref}       
\usepackage{url}            
\usepackage{booktabs}       
\usepackage{amsfonts}       
\usepackage{nicefrac}       
\usepackage{microtype}      
\usepackage{xcolor}         
\usepackage{url}            
\usepackage{booktabs}       
\usepackage{amsfonts}       
\usepackage{nicefrac}       
\usepackage{microtype}      
\usepackage{xcolor}         
\usepackage{graphicx,psfrag,amsmath,amsfonts,verbatim}
\usepackage[small,bf]{caption}
\usepackage{hyperref}       
\usepackage{url}            
\usepackage{booktabs}       
\usepackage{amsfonts}       
\usepackage{nicefrac}       
\usepackage{microtype}      
\usepackage{xcolor}         
\usepackage{amsmath}
\usepackage{amsthm}
\usepackage{thmtools}
\usepackage{thm-restate}
\newtheorem{theorem}{Theorem}
\newtheorem{lemma}[theorem]{Lemma}

\newtheorem{remark}{Remark}

\theoremstyle{definition}
\newtheorem{definition}{Definition}

\newtheorem{assumption}{Assumption}

\usepackage{subfigure}
\usepackage{caption}
\usepackage{algorithm}
\usepackage{algorithmic}
\usepackage{graphicx}
\usepackage{thm-restate}
\usepackage{caption}
\usepackage{booktabs}
\usepackage{wrapfig}
\usepackage{multirow}

\title{Finite Sample Analyses for Continuous-time Linear Systems: System Identification and Online Control}

\author{%
  Hongyi Zhou\thanks{Equal Contribution} \\
  IIIS, Tsinghua University\\
  Shanghai Qizhi Institute \\
\texttt{zhouhong24@mails.tsinghua.edu.cn} \\
  \And
 Jingwei Li$^{*}$ \\
  IIIS, Tsinghua University\\
  Shanghai Qizhi Institute \\
\texttt{ljw22@mails.tsinghua.edu.cn} \\
  \And
  Jingzhao Zhang\thanks{Corresponding Author} \\
  IIIS, Tsinghua University \\
  Shanghai Qi zhi Institute \\
  \texttt{jingzhaoz@mail.tsinghua.edu.cn} \\
}

\begin{document}

\maketitle

\begin{abstract}
 Real world evolves in continuous time but computations are done from finite samples. Therefore, we study algorithms using finite observations in continuous-time linear dynamical systems. 
We first study the system identification problem, and propose a first non-asymptotic error analysis with finite observations. Our algorithm identifies system parameters without needing integrated observations over certain time intervals, making it more practical for real-world applications. Further we propose a lower bound result that shows our estimator is provably optimal up to constant factors.
Moreover, we apply the above algorithm to online control regret analysis for continuous-time linear system. Our system identification method allows us to explore more efficiently, enabling the swift detection of ineffective policies. We achieve a regret of $\mathcal{O}(\sqrt{T})$ over a single $T$-time horizon in a controllable system, requiring only $\mathcal{O}(T)$ observations of the system. 
\end{abstract}

\section{Introduction}
Finding optimal control policies requires accurately modelling the system~\citep{ref25}. However, real-world environments often involve unknown system parameters. In such cases,  estimating unknown parameters from exploration becomes essential to identify the unseen dynamics. This process is recognized as system identification, a fundamental tool employed in various research fields, including time-series analysis~\citep{ref12}, control theory~\citep{ref13}, robotics~\citep{ref14}, and reinforcement learning~\citep{ref15}. 

The identification of linear systems has long been studied because linear systems, as one of the most fundamental systems in both theoretical frameworks and practical applications,  has wide applications ranging from natural physical processes to robotics. Most classical results provide only \emph{asymptotic} convergence guarantees for parameter estimation~\citep{ref8, ref20, ref26}. 

On the other hand, with the rapid increase in data scale, there is a growing concern for statistical efficiency. Consequently, the non-asymptotic convergence of~\emph{discrete-time} linear system identification has emerged as another pivotal topic in this field. Investigations into this matter delve into understanding how estimation confidence is influenced by the sample complexity of trajectories~\citep{ref5}, or the running time on a single trajectory~\citep{ref4,ref23}. Furthermore, many of these studies operate under the common assumption of stochastic noise, there has been a parallel exploration into the identification of discrete-time linear dynamical systems with diverse setups. This includes scenarios where perturbations are adversarial~\citep{hazan2020nonstochastic} or when only black-box access is available~\citep{ref28}.

In contrast to studies in discrete time system, there have been relatively fewer non-asymptotic results addressing parameter identification for \emph{continuous-time systems}. Two problems exist for continuous time analysis. First, nonasymptotic analysis in continuous system without noise can be degenerate, as a short time interval can contain infinite pieces of information. Second, if we consider the non-degenerate case when finite noisy observations are available, then the analyses require concentration results that become known only as in~\citep{ref4,ref5,ref23}. Recently~\cite{ref3} provides novel analyses for estimating  system parameters, which relies on continuous data collection and
interaction with the environment. 
Motivated by progress in these works, our first goal is to answer the question below: 

\begin{center}
\emph{Can we design a continuous-time stochastic system identification algorithm that provides nonasymptotic error bounds with only a finite number of samples?}
\end{center}

We will introduce our system identification algorithms tailored to meet the above requirements. As expected, we discretize time into small intervals, thereby reducing the problem to a discrete system. The interesting part involves ensuring that the discretization remains bijective and that the inversion is unbiased. Our algorithm identifies the continuous system using only a finite number of samples from the discrete system. We further propose a information theoretic lower bound that shows our algorithm is optimal. 

As an application of our system identification methods, we study an online continuous-time linear control problem as introduced in~\citep{ref1}. In this context, exploration is essential for estimating unknown parameters, with the goal of identifying a more optimal control policy that narrows the performance gap. The primary challenge involves finding the right balance between exploration and exploitation. Leveraging our identification method for more efficient parameter estimation allows us to effectively manage exploration and exploitation, achieving an expected regret of $\mathcal{O}(\sqrt{T})$ over a single trajectory with only $\mathcal{O}(T)$ samples in time horizon $T$. This surpasses the previously best known result of $\mathcal{O}(\sqrt{T}\log(T))$, which needs continuous data collection from the system. 

We summarize our contributions below.
\begin{enumerate}
    \item When the system can be stabilized by a known controller, we establish an algorithm with $\mathcal{O}(T)$ samples that achieves estimation error $\mathcal{O}(\sqrt{1/T})$ on a single trajectory with running time $T$, which is shown in Theorem~\ref{thm:main-stable}. We also provide Theorem~\ref{thm:lower bound} which shows that the estimation error of our system identification method is optimal up to constant factors.
    \item When a stable controller is not available, we can use $N$ independent short trajectories to obtain estimators with error $\mathcal{O}(\sqrt{1/N})$, as is shown in Theorem~\ref{thm:main-unstable} .
    \item We apply our system identification method to an online continuous linear control algorithm, which only requires $\mathcal{O}(T)$ samples and achieves $\mathcal{O}(\sqrt{T})$ regret on a single trajectory with lasting time $T$ (Theorem~\ref{thm:total regret}), improving upon the best known result $\mathcal{O}(\sqrt{T}\log(T))$ in~\citep{ref1}.
\end{enumerate}

\section{Related Works}
Control of both discrete and continuous linear dynamical systems have been extensively studied in various settings, such as linear quadratic optimal control~\citep{ref16}, $H_{2}$ stochastic control~\citep{ref17}, $H_{\infty}$ robust control~\citep{ref18,ref19} and system identification~\citep{ref13,ref20}. Below we introduce some of the important results on both system identification and optimal control for linear dynamical systems.

\paragraph{System Identification} 
Earlier literature focused primarily on the asymptotic convergence of system identification~\citep{campi1998adaptive, ljung1998system}. Recently, there has been a resurgence of interest in non-asymptotic system identification for \emph{discrete-time} systems. \cite{ref5} studied the sample complexity of multiple trajectories, with $\mathcal{O}(\sqrt{1/N})$ estimation error on $N$ independent trajectories. For systems with dynamics $x_{t+1}=Ax_{t}+w_{t}$(without  controllers),
\cite{ref4} established an analysis for $\mathcal{O}(\sqrt{1/T})$ estimation error on a single stable trajectory with running time $T$, while~\cite{ref29} and~\cite{ref23} extended to more general discrete-time systems. 

Non-asymptotic analyses for continuous-time linear system are less studied. Recently, \cite{ref3} examined continuous-time
linear quadratic control systems with standard brown noise and unknown system dynamics. Our algorithm is specifically designed for finite observations, achieving an error rate that cannot be attained through the direct discretization of integrals as done in~\citep{ref3}. 

\paragraph{Regret Analysis of Online Control}
In  online control, if the system's parameters are known, achieving the optimal control policy in this setup can be straightforward~\citep{ref18,ref24}. However, when the system parameters are unknown, identifying the system incurs regret. \cite{abbasi2011regret} achieved an $\mathcal{O}(\sqrt{T})$ regret for discrete-time online linear control, which has been proven optimal in $T$ under that setting in~\cite{ref7}. Subsequent works have extended this setup, focusing on worst-case analysis with adversarial noise and cost, including~\citep{mania2019certainty,cohen2019learning,lale2020explore,ref7, ref30}. These analyses are limited to discrete systems. For continuous-time systems, 
works of \cite{shirani2022thompson,ref1,li2024online} established algorithms for online continuous control that achieves $\mathcal{O}\left(\sqrt{T}\log(T)\right)$ regret.

\section{Problem Setups and Notations}\label{section: problem setting}

In this section, we introduce the background and notation for linear dynamical systems and online control. 

\subsection{Linear Dynamical Systems}\label{subsection: linear dynamical system}

We first introduce discrete-time linear dynamical systems as follows: Let $x_{k} \in \mathbb{R}^{d}$ represent the state of the system at time $k$, and let $u_{k} \in \mathbb{R}^{p}$ denote the action at time $k$. Then, for some linear time-invariant dynamics characterized by $A \in \mathbb{R}^{d \times d}$ and $B \in \mathbb{R}^{d \times p}$, the transition of the system to the next state can be represented as:
\begin{align}
x_{k+1} = Ax_{k} + Bu_{k} + w_{k},
\label{equ:discrete}
\end{align}
where $w_{k} \in \mathbb{R}^{d}$ are i.i.d. Gaussian random vectors with zero means and certain covariance.

Similarly, a continuous-time linear dynamical system with stochastic disturbance at time $t$ is defined by a differential equation, instead of a recurrence relation:
\begin{align}
dX_{t} = AX_{t}dt + BU_{t}dt + dW_{t}.
\label{equ:differential}
\end{align}
In this context, we use $X_t$ and $U_t$ to represent the state and action in the continuous-time linear system, distinguishing them from $x_t$ and $u_t$ in discrete-time systems. $W_{t}$ denotes the stochastic noise, which is modeled by standard Brownian motion.

For a continuous control problem, an important question of a linear dynamical system is whether such system can be stably controlled. Below we define the concepts of stable dynamics and stabilizers.

\begin{definition}\label{def:alpha}
For any square matrix $A$, define $\alpha(A)=\max_{i}\{\Re(\lambda_{i})\vert \lambda_{i}\in \lambda(A)  \}$, where $\Re(\lambda)$ represents the real part of complex number $\lambda$, $\lambda(A)$ is the set of all eigenvalues of $A$.    
\end{definition}

\begin{definition}\label{def:stable}
A matrix $A\in \mathbb{R}^{d\times d}$ is stable if $\alpha(A)<0$. A control matrix $K\in \mathbb{R}^{p\times d}$ is said to be a stabilizer for system $(A,B)$ if $A+BK$ is stable. 
\end{definition}
Under the above definition, a stable dynamic guarantees that the state can automatically go to the origin when no external forces are added, while applying a stabilizer as the dynamic for controller will also ensure that the state does not diverge.

\subsection{Continuous-time LQR Problems and Optimal Control}\label{sec:online}

For continuous-time linear systems disturbed by stochastic noise, as introduced in~\ref{subsection: linear dynamical system}, we denote the strategy of applying control to such systems through a specific causal policy, $f: X \to U$. This policy maps states $X$ to control inputs $U$, where the policy at time $t$ can only depend on the states and actions prior to $t$.

The optimal controls in linear systems are often linear \citep{ref18,ref24}, which takes the following form $$U_{t}=K_{t}X_{t},$$ 
where $K_{t}\in \mathbb{R}^{p\times d}$ represents the linear parameterization at time $t$ under some policy $f(X) = KX$. Additionally, we define the cost function of applying the action $U_{t}=K_{t}X_{t}$ with linear quadratic regulator (LQR) control. Given predefined symmetric positive definite matrices $Q\in\mathbb{R}^{d\times d}$ and $R\in\mathbb{R}^{p\times p}$, along with the initial state $X_{0}$, the cost during $t\in[0,T]$ is denoted by $J_{T}$, as represented in the following equation:
\begin{align}\label{eq:cost}
J_{T} = \mathbb{E}\left[\int_{t=0}^{T}\left(X^{\mathrm{T}}_{t}QX_{t}+U_t^{\mathrm{T}}RU_t\right)dt\right]\,.
\end{align}
Here the expectation is taken over the randomness of $X_{t}$.

Among all the polices there exists an optimal mapping $f_{*}$ which minimizes $J_{T}$. When the system is dominated by dynamics $(A,B)$, with the state transits according to~\eqref{equ:differential}, such optimal $K_{t}$ can be computed via the Lyapunov matrix $P_{t}$ that solves the Ricatti differential equation~\citep{ref24}:
\begin{align}\label{eq:ricatti}
\frac{d}{dt}{P}_{t}={P}_{t}^{\mathrm{T}}{B}R^{-1}{B}^{\mathrm{T}}{P}_{t}-{A}^{\mathrm{T}}{P}_{t}-{P}_{t}^{\mathrm{T}}{A}-Q, \quad {P}_{T} = 0.
\end{align}
Then, under $f_{*}$ the action dynamic is set to be  $K_{t}=-R^{-1}B^{\mathrm{T}}P_{t}$.

When $T\to +\infty$, the starting dynamic $P_{0}$ converges to some special dynamic $P_{*}$ satisfying
\begin{align}\label{eq:stationary}
P_{*}^{\mathrm{T}}BR^{-1}B^{\mathrm{T}}P_{*}-A^{\mathrm{T}}P_{*}-P_{*}^{\mathrm{T}}A-Q=0\,,
\end{align}
and the optimal control policy for infinite time horizon is by setting $K_{*}=-R^{-1}B^{\mathrm{T}}P_{*}$ and apply the action by $U_{t}=K_{*}X_{t}$.

\paragraph{Online Control Problems.} Online learning aims to find a strategy to output a sequence of controls $\{U_t\}$  that minimizes the cost $J_{T}$ without knowing the system parameters $A, B$. In this scenario, the algorithms explore to obtain valuable information, such as estimators $(\hat{A},\hat{B})$ for $(A,B)$, while simultaneously exploit gathered information to avoid large instantaneous cost.

To quantify the progress in an online learning problem with horizon $T$, one quantity of interest is the regret $R_{T}$, which quantifies the performance gap between the control taken $U_t = f(X_t)$ and a baseline optimal policy which takes $U_{t} = K_{*}X_{t} =-R^{-1}B^{\mathrm{T}}P_{*}X_{t}$, where $K_*$ is defined in~\eqref{eq:stationary}. Formally, by denoting $J_{T}$ be the expected cost under $f$, and $J^{*}_{T}$ be the expected cost under the baseline optimal policy, the regret $R_{T}$ is represented as:
\begin{align}\label{expression:regret}
R_{T}=J_{T}-J^{*}_{T}\,.
\end{align}

\paragraph{Other Notations} 
Denote the d-dimensional unit sphere $\mathcal{S}^{d-1}=\{v\in \mathbb{R}^{d},\|v\|_{2}=1\}$, where $\|\cdot\|_{2}$ is the $L_{2}$ norm.
For any matrix $A\in \mathbb{R}^{m\times n}$, denote $\|A\|$ be the spectral norm of $A$, or equivalently, 
\begin{align*}
\|A\|=\sup_{v\in \mathcal{S}^{n-1}}\|Av\|_{2}=\sup_{u\in \mathcal{S}^{m-1},v\in \mathcal{S}^{n-1}}u^{\mathrm{T}}Av.
\end{align*}

\section{The Proposed System Identification Method}
In this section we propose our system identification method. Before presenting our method, we first introduce the formal definition of system identification and the finite observation setting.

\subsection{System Identification and Finite Observation}

We start with the definition of system identification.

\begin{definition}[System Identification]
The system identification task aims to recover the true system dynamics matrices $A$ and $B$ by observing the system's response over time. Specifically, one selects a time horizon $T$ and a sequence of actions $U$, observes the resulting states $X$, and computes estimates $\hat{A}$ and $\hat{B}$ of the true dynamics. The goal is to design an effective algorithm that achieves the following non-asymptotic estimation bound:
\begin{align*}
\|\hat{A} - A\|,\ \|\hat{B} - B\| &\leq f(T)\,,
\end{align*}
for some function $f$ depending on $T$. In particular, as $T \to \infty$, we expect the estimation error $f(T)$ to converge to zero.
\end{definition}

Next, we introduce the \emph{finite observation assumption}. Under this setting, the number of observed states $N$ grows at most linearly with the trajectory running time $T$. In other words, for any trajectory of length $T$, we can only access a finite set of states $\{X_{1}, X_{2}, \dots, X_{N}\}$ to identify the system, where $N = O(T)$ and does not exhibit superlinear growth.

To analyze the continuous-time system, we need to discretize it. Prior works~\cite{ref3, shirani2022thompson, ref1} commonly approximate the dynamics using
\[
X_{t+h} \approx (I + hA)X_t + hBU_t + (W_{t+h} - W_t).
\]
However, this approximation introduces a discretization error between the approximated and true dynamics. The error term, characterized by $(e^{hA} - I)/h - A$, is of order $O(h)$. Consequently, the sampling interval $h$ must be chosen as $O(1/\sqrt{T})$ to ensure that discretization error does not dominate. This leads to a super-linear sampling complexity of $m = T/h = \Omega(T^{3/2})$, which violates the finite observation assumption and significantly increases computational demands.

In contrast, our method overcomes this limitation by directly estimating the matrix exponential $e^{Ah}$ in Lemma~\ref{lem:tran}, and subsequently recovering $(A, B)$ from this estimate. As a result, our approach avoids discretization error entirely, allowing the sampling interval to depend solely on system parameters rather than the total sampling time $T$. This innovation reduces the sampling complexity to grow linearly with $T$, offering significant computational advantages.

\subsection{Algorithm 
Design}\label{section:alg for system id}
Then we introduce our algorithm. We choose a small sampling time interval $h$ across a single trajectory of time length $T$. We then divide the time into small intervals and consider the state evolution within each interval. We get the following Lemma:

\begin{restatable}{lemma}{tran}
\label{lem:tran}
   In the time interval $[t, t+h]$, the following transition function holds:
\begin{align*}
X_{t+h} = e^{Ah}X_{t}+\int_{s=0}^{h}e^{A(h-s)}BU_{t+s}ds+ w_t \,,
\end{align*}
\end{restatable}

Here, $w_{t}$ is Gaussian noise $\mathcal{N}(0, \Sigma)$ with covariance $\Sigma = \int_{s=0}^{h}e^{As}e^{A^{\mathrm{T}}s}ds$. The formal proof of this Lemma is deferred to the Appendix~\ref{supp:lem1}.

This transition equation connects continuous-time and discrete-time systems. In our method, the whole trajectory is partitioned into intervals with proper determined length $h$. During time $t\in[kh,(k+1)h]$, we observe a state $x_{k}$ at time $t=kh$, and fix the action $U_{t}\equiv u_{k}$ in this interval. Denoting $A^{'}=e^{Ah}$ and $B^{'}=\left[\int_{s=0}^{h}e^{A(h-s)}ds\right] B$, then the set of observations $\{x_{k} \vert k=0,1,2,...  \}$ and actions $\{u_{k} \vert k=0,1,2,...  \}$ follow the standard discrete-time linear dynamical system:
$$ x_{k+1}=A^{'}x_{k}+B^{'}u_{k}+w_{k}.$$

Then we can apply the system identification method of discrete-time system~\cite{ref4, ref5}. However, different from classical discrete-time systems, continuous-time systems present new challenges. The crucial one is that knowing $e^{Ah}$ is not sufficient to determine $A$, because the matrix exponential function $f(X)=e^{X}$ is not one-to-one. This means we might obtain an incorrect estimator $\hat{A}$ by solving $e^{\hat{A}h}=M$, where $M$ is the estimate of $e^{Ah}$. From the above analysis, we introduce our assumptions of the algorithm.
\begin{assumption}[Assumptions for Algorithm~\ref{alg-single} and Theorem~\ref{thm:main-stable}]\label{assumption:stable} We assume
    \begin{enumerate}
        \item The linear dynamic $A$ is stable, with $\alpha(A)<0$ (see Definition~\ref{def:alpha}). This is equivalent to assuming the existence of a stable controller $K$ and then set $A \leftarrow A + BK$.
        \item $\|A\|\leq\kappa_{A}$, $\|B\|\leq \kappa_{B}$ for some known $\kappa_{A},\kappa_{B}$ ($\kappa_{A},\kappa_{B}$ need not be closed to $\|A\|,\|B\|$).
        \item The sample interval $h$ is chosen to be $h=\frac{1}{15\kappa_{A}}$.
    \end{enumerate}
\end{assumption}

With the above assumptions, we design our algorithm as described in Algorithm~\ref{alg-single}. In the $k$-th interval of length $h$, the state $x_k$ is observed at the beginning, and a randomly selected action $u_k$ is applied uniformly throughout the interval. The state-action pair ${x_k, u_k}$ is then used to estimate the discretized dynamics via:
\begin{align}\label{eq:alg1}
(\widetilde{A})^{\mathrm{T}}=\left[\sum_{k=0}^{T_{0}-1}x_{k}x^{\mathrm{T}}_{k}\right]^{\dagger}\sum_{k=0}^{T_{0}-1}x_{k}x^{\mathrm{T}}_{k+1}\,,(\widetilde{B})^{\mathrm{T}}=\left[\sum_{k=0}^{T_{0}-1}u_{k}u^{\mathrm{T}}_{k}\right]^{\dagger}\sum_{k=0}^{T_{0}-1}u_{k}\left(x_{k+1}-\widetilde{A}x_{k}\right)^{\mathrm{T}}\,.
\end{align}	
The continuous-time dynamics $(A,B)$ are then recovered from $(\widetilde{A},\widetilde{B})$. Under the condition $\|A\|h \ll 1$, we employ Taylor expansion to compute $\hat{A}h = \log(\widetilde{A})$, approximating $Ah$. The estimators $(\hat{A},\hat{B})$ are given by:
\begin{align}\label{eq:recover}
\hat{A}=\frac{1}{h}\sum_{k\geq1}\frac{(-1)^{k-1}}{k}(\widetilde{A}-I)^{k},\hat{B}=\left[\int_{t=0}^{h}e^{\hat{A}t}dt\right]^{-1}\widetilde{B}\,.
\end{align}

\begin{algorithm}[t]
	\caption{System identification algorithm for stable system} \label{alg-single}
	\begin{algorithmic}
		\STATE \textbf{Input:} Running time $T$, sample interval $h$ satisfying the condition in Assumption~\ref{assumption:stable}. 
		\STATE Define the number of samples $T_0 = \lceil T/h \rceil$.
		\FOR{$k = 0, \ldots, T_0-1$}
		\STATE Sample the action $u_{k}\stackrel{\text { i.i.d. }}{\sim}\mathcal{N}\left(0, I_p\right)$.
		\STATE Use the action $U_t \equiv u_{k}$  during the time period $t \in [kh, (k+1)h]$.
		\STATE Observe the new state $x_{k+1}$ at time $(k+1)h$.
		\ENDFOR 
		\STATE Compute system estimates $(\hat{A},\hat{B})$ via ~\eqref{eq:recover}.

	\end{algorithmic}
\end{algorithm}


We now establish the efficiency of our algorithm and derive the main theorem as follows.
\begin{theorem}[Upper bound]\label{thm:main-stable}
In Algorithm~\ref{alg-single}, there exists a constant $C\in poly\left(\vert \alpha(A)\vert^{-1},\kappa_{A},\kappa_{B}\right)$ such that, $\forall\,0<\delta<\frac{1}{2}$, when $T\geq C\left(\|X_{0}\|^{2}_{2}+\log^{2}{1/\delta}\right)$, with probability at least $1-\delta$, we have:
	\begin{align}
	\|\hat{A}-A\|,\|\hat{B}-B\|\leq C\sqrt{\frac{\log(1/\delta)}{T}}\,.
	\end{align}
\end{theorem}

We defer the proof of the theorem to Appendix~\ref{sec:upper bound} and highlight the key idea below.
The key idea of the proof is to analyze the error transformation from the discrete system to the original system. We prove Lemma~\ref{lem:error-transform}, which shows that the errors in the discrete and original systems differ only by a constant factor. This allows us to focus solely on the discrete system identification problem.

\begin{lemma}\label{lem:error-transform}
In Algorithms~\ref{alg-single}, suppose we obtain the relative error $\|\widetilde{A}-A^{'}\|,\|\widetilde{B}-B^{'}\|\leq \epsilon$ for some $\epsilon\leq \frac{1}{15}$ and $\|Ah\|\leq \frac{1}{15}$. Then, the relative error in the original system satisfies:
	\begin{align}
		\|\hat{A}-A\|,\|\hat{B}-B\|\leq \frac{1}{h}\left(2+\frac{\kappa_{B}}{\kappa_{A}}\right)\epsilon\,.
	\end{align}
\end{lemma}

From this lemma, it follows that if we develop a system identification algorithm for the discrete system that produces dynamics estimates $\widetilde{A}$ and $\widetilde{B}$ with minimal error, we can obtain accurate estimates for the original system. The remaining task is to analyze the discrete system with the transition function $x_{k+1}=Ax_{k}+Bu_{k}+w_{k}$, which has been discussed in previous works such as~\cite{ref4}.

\subsection{Lower Bound}
In this section, we discuss the lower bound of the problem. We prove Theorem~\ref{thm:lower bound} and establish that this method has already attained the optimal convergence rate for parameter estimation. The theorem primarily asserts that, given a single trajectory lasting for time $T$, any algorithm that estimates system parameters solely based on \emph{an arbitrarily large number of finite observed states} cannot guarantee an estimation error of $o(\sqrt{1/T})$.

\begin{theorem}[Lower bound]\label{thm:lower bound}
Suppose $T\geq 1$ be the running time of a single trajectory of continuous-time linear differential system, represented as in~\eqref{equ:differential}. Then there exist constants $c_{1},c_{2}$ independent of $d$ such that, for any finite set of observed points $\{t_{0}=0,t_{1},t_{2},...,t_{N}=T\}$, and any (possibly randomized) estimator function $\phi : \{X_{t_{0}},X_{t_{1}},...,X_{t_{N}}\}\to \mathbb{R}^{d\times d}$, there exists system parameter $A,B$ satisfying $\mathbb{P}\left[\Vert \phi(\{X_{t_i}\}_{i \le N}) - A\Vert\geq \frac{c_{1}}{\sqrt{T}}\right]\geq c_{2}$. Here the probability is with respect to noise.
\end{theorem}

In Theorem~\ref{thm:lower bound}, the mapping $\phi$ can refer to the output of any algorithm that exclusively relies on the finite set of states ${X_{t_{0}},X_{t_{1}},...,X_{t_{N}}}$. \emph{The interesting observation is that the lower bound does not decrease with a larger observation number $N$.}

We defer the proof of the theorem to the Appendix~\ref{sec:lower bound} and provide a proof sketch below. We consider two sets of dynamics, $(A,0)$ and $(\bar{A},0)$, where both $A$ and $\bar{A}$ are stable, and $|A-\bar{A}| = \frac{2c_{1}}{\sqrt{T}}$. Our key observation is that for the two distributions of observed states $S_{k}=\{X_{t_{0}},X_{t_{1}},...,X_{t_{k}}\}$ and $\bar{S}_{k}=\{\bar{X}_{t_{0}},X_{t_{1}},...,X_{t_{k}}\}$, where $X$ corresponds to the linear dynamic $A$ and $\bar{X}$ corresponds to $\bar{A}$, the KL divergence between $S_{k+1}$ and $\bar{S}_{k+1}$ increases by at most $\frac{c}{T}(t_{k+1}-t_{k})$. Here, $c$ is a universal constant independent of $t_{k}$ and $t_{k+1}$. Thus, regardless of how the observation times are selected, the KL divergence between the observed states remains bounded.

\begin{remark}[The Discussion of Lower Bound]
\normalfont

The construction in Theorem~\ref{thm:lower bound} involves matrices $A$ and $\bar{A}$ that depend on $T$, specifically with $\lVert A - \bar{A} \rVert = \tfrac{2c_{1}}{\sqrt{T}}$. 
One might be concerned that such a $T$-dependent construction lacks interpretability since the true system parameters are independent of $T$. 
However, as shown in Appendix~\ref{sec:lower bound}, the matrices are taken as $A = -I_{d}$ and $\bar{A} = -I_{d} - U$, where $U$ has only one nonzero entry at position $(1,1)$ equal to $\tfrac{1}{5\sqrt{T}}$. 
For these matrices, the key constants in the upper bound remain uniformly bounded: the inverse stability margin $\tfrac{1}{|\alpha(A)|}$ equals $1$ for $A=-I_{d}$ and is at most $\tfrac{1}{1+\tfrac{1}{5\sqrt{T}}} \leq 1$ for $\bar{A}$; the condition number $\kappa(A)$ equals $1$ for $A=-I_{d}$ and is at most $1+\tfrac{1}{5\sqrt{T}} \leq 2$ for $\bar{A}$. 
Thus both quantities are controlled by universal constants, independent of $T$, ensuring that the lower and upper bounds are comparable up to a constant factor.
\end{remark}

\subsection{Finding an Initial Stable Controller}

While previous work on continuous-time system identification~\cite{ref3, ref1} always assumes a known stable controller, our method extends to cases where a stabilizer is not known in advance. 
For general $(A,B)$, where a stabilizer is not predetermined, relying on a single trajectory is not feasible, as the state may diverge rapidly before obtaining a stable controller is obtained. Instead, we first find a stable controller $K$ using multiple short-interval trajectories and then employ it in Algorithm~\ref{alg:3} for online control. Below, we list the assumptions on system parameters.

\begin{assumption}[Assumptions for Algorithm~\ref{alg-oracle} and Theorem~\ref{thm:main-unstable}]\label{assumption:multiple} We assume
    \begin{enumerate}
        \item The constants $\kappa_{A},\kappa_{B},h$ follow the same assumptions as in \ref{assumption:stable}.
        \item The running time $T$ for each trajectory is small, say, $T=T_{0}h$ where $T_{0}\in\mathbb{N}$ and $T_{0}\leq 10$.
    \end{enumerate}
\end{assumption}
Then, we employ multiple short trajectories to identify $A$ and $B$ as outlined in Algorithm~\ref{alg-oracle}. Similar to what is demonstrated in \cite{ref5}, this procedure results in an $\mathcal{O}(H^{-1/2})$ estimation error on the trajectory number $H$.
\begin{theorem}\label{thm:main-unstable}
In Algorithm~\ref{alg-oracle}, there exists a constant $C\in poly(\kappa_{A},\kappa_{B})$ such that w.p. at least $1-\delta$, the estimation error of $(\hat{A},\hat{B})$ from $H$ trajectories satisfies:
	\begin{align*}
	\|\hat{A}-A\|,\|\hat{B}-B\|\leq C\sqrt{\frac{\log(1/\delta)}{H}}\,.
	\end{align*}

\end{theorem}
The proof details are shown in the Appendix~\ref{subsection: multiple trajectories}.

\begin{algorithm}[t]
	\caption{Multi-trajectory system identification algorithm} \label{alg-oracle}
	\begin{algorithmic}
		\STATE \textbf{Input:} $T$, $T_{0}$, $h$ as in Assumption~\ref{assumption:multiple}, number of trajectories $H$.
		\FOR{$l = 1, \ldots, H$}
		\FOR{$k = 0, \ldots, T_0-1$}
		\STATE Sample the action $u^{l}_{k}\stackrel{\text { i.i.d. }}{\sim}\mathcal{N}\left(0, I_p\right)$, use the action $U_t \equiv u^{l}_{k}$  during  $t \in [kh, (k+1)h]$.
		\STATE Observe the new state $x^{l}_{k+1}$ at time $(k+1)h$.
		\ENDFOR 
		\ENDFOR
		\STATE Compute $(\widetilde{A},\widetilde{B})$ by $(\widetilde{A},\widetilde{B})\in \arg
\min_{(A,B)}\frac{1}{2}\sum_{l=1}^{H}\left\Vert x^{l}_{T_{0}}-Ax^{l}_{T_{0}-1}-Bu^{l}_{T_{0}-1}\right\Vert^{2}_{2}$.
		\STATE Compute $\bar{A},\bar{B}$ as in~\eqref{eq:recover}, let $(\hat{A},\hat{B})=(\bar{A},\bar{B})$ be estimates for system dynamics $(A,B)$.
	\end{algorithmic}
\end{algorithm}

\section{A Continuous Online Control Algorithm with Improved Regret}\label{section: online control}
In this section, we apply our system identification method to a continuous LQR online control algorithm.  Recall the setup introduced in Section~\ref{sec:online} where we want to minimize the regret $R_{T}$ defined in \eqref{expression:regret}.
We will show in this section that with $\mathcal{O}(T)$ samples, our algorithm achieves $\mathcal{O}(\sqrt{T})$ expected regret on a single trajectory, thereby improving upon the previous $\mathcal{O}\left(\sqrt{T}\log(T)\right)$ result. We list the assumption for the online LQR problems  below.

\begin{assumption}[Assumptions for Algorithm~\ref{alg:3} and Theorem~\ref{thm:total regret}]\label{assumption:alg} We assume that:
\begin{enumerate}
    \item A stabilizer $K$ for $(A,B)$ (see Definition~\ref{def:stable}) with $\alpha(A+BK)<0$ is known in advance.
    \item Sample distance $h$ satisfies $h=\frac{1}{15\kappa}$, where $\kappa\geq\Vert A\Vert +\Vert B\Vert \Vert K\Vert\geq \Vert A+BK\Vert$ is known.
    \item Denote $P_{*}$ be the solution in \eqref{eq:stationary} and $K_{*}=-R^{-1}B^{\mathrm{T}}P_{*}$ be the baseline control dynamic.
    \item $Q,R$ are positive-definite symmetric matrices with bounded spectral norms $\Vert Q\Vert,\Vert R\Vert\leq M$ and for some $\mu > 0$, $\mu I \preceq Q, \mu I \preceq R $.
\end{enumerate}
\end{assumption}



\subsection{An $\mathcal{O}(\sqrt{T})$ Regret Algorithm for Continuous Online Control}

Our online continuous control algorithm is outlined in Algorithm~\ref{alg:3}, and we provide a detailed description below. Algorithm~\ref{alg:3} is divided into two phases, exploration and exploitation. For the first exploration phase, a previously known stabilizer $K$ is applied to prevent the state from diverging. During the $k$-th interval, by setting $U_{t}=KX_{t}+u_{k}$, the state $X_{t}$ transits according to 
$$dX_{t}=(A+BK)X_{t}dt+Bu_{k}dt+dW_{t}.$$
Since $A+BK$ is stable, through replacing $A$ in Theorem~\ref{thm:main-stable} by $A+BK$ in Algorithm~\ref{alg:3}, we can obtain a set of estimators $(\hat{A},\hat{B})$ for $(A,B)$ with small error. This further allows us to accurately estimate $(A,B)$, thereby a controller $\bar{K}=-R^{-1}(\hat{B})^{\mathrm{T}}P$ closed to $K_{*}$ is obtained.

During exploitation phase, the near-optimal controller is deployed to minimize the cost, resulting in a regret of $\mathcal{O}(\sqrt{T})$ (see Theorem~\ref{thm:total regret}). However, as we lack direct feedback on whether $\bar{K}$ is a stabilizer, we need to detect its stability. Our approach involves replacing it with the known stabilizer $K$ whenever the state deviates too far. Then we introduce the theorem of the regret analysis:

\begin{algorithm}[t]\caption{Continuous online control algorithm} \label{alg:3}
\begin{algorithmic}
    \STATE \textbf{Input:} $K,h$ which follows Assumption~\ref{assumption:alg}, running time $T$
    \FOR{$k = 0, \ldots, [\frac{\sqrt{T}}{h}]-1$}
    \STATE Sample the action $u_{k}\stackrel{\text { i.i.d. }}{\sim}\mathcal{N}\left(0, I_p\right)$.
    \STATE For $t\in[kh,(k+1)h]$, set $U_{t}=KX_{t}+u_{k}$.
    \STATE Observe the new state $x_{k+1}$ at time $(k+1)h$.
    \ENDFOR
    \STATE \textbf{Do system identification and estimate dynamics:}
    \STATE Compute $(\widetilde{A},\widetilde{B})$ according to ~\eqref{eq:alg1} by using $\{x_{k},u_{k}\}$.
    \STATE Compute $\bar{A},\bar{B}$ by \eqref{eq:recover} with $\widetilde{A} ,\widetilde{B}$, and  estimators $(\hat{A},\hat{B})$ by $\hat{A}=\bar{A}-\bar{B}K$,$\hat{B}=\bar{B}$.
    \STATE If $\hat{A}$ is stable, compute $P$ by \eqref{eq:stationary} with estimated $\hat{A},\hat{B}$, and set $\bar{K}=-R^{-1}(\hat{B})^{\mathrm{T}}P$.
    \STATE If $\hat{A}$ is not stable or $P$ computed above satisfies $\|P\|\geq T^{\frac{1}{5}}$, then set $\bar{K}=K$.
		
    \STATE \textbf{Perform exploitation:}
		
    \STATE For $t\in [\sqrt{T},T]$, set $U_{t}=\bar{K} X_{t}$.
		
    \STATE \textbf{Detect bad policy and prevent the trajectory from diverging:}
		
    \STATE If for some $t_{0}\geq \sqrt{T}$, $\|X_{t_{0}}\|\geq {T}^{\frac{1}{5}}$, then set $U_{t}=KX_{t}$ for $t\in[t_{0},T]$.
 	\end{algorithmic}
 \end{algorithm}


\begin{theorem}\label{thm:total regret}
Let $J_{T}$ be the expected LQR cost introduced in \eqref{eq:cost} that takes the action $U_{t}$ as in Algorithm 3. Then for some constant $C\in poly\left(\kappa,M, \mu^{-1}, \vert\alpha(A+BK)\vert^{-1},  \vert\alpha(A+BK_{*})\vert^{-1}\right)$, the regret satisfies:
 \begin{align*}
     R_{T}=J_{T}-J^{*}_{T}\leq C\sqrt{T}\,.
 \end{align*}
\end{theorem}

\paragraph{Proof Sketch of Theorem~\ref{thm:total regret}}
We analyze the two phases of our algorithm. During the exploration phase, the stabilizing controller $K$ effectively bounds the trajectory's radius, ensuring the average cost per unit time is within $\mathcal{O}(1)$, resulting in a total exploration cost of $C\sqrt{T}$. In the subsequent exploitation phase, we analyze two scenarios separately.  
The first scenario occurs when the estimators $(\hat{A},\hat{B})$ have large errors or when $\|X_{t}\|_{2} \geq T^{1/5}$ for some $t \in [\sqrt{T},T]$. This situation is rare and contributes a limited expected cost that can be bounded by a constant.  
The second scenario occurs when $(\hat{A},\hat{B})$ are accurately estimated, and the control $U_{t}=-R^{-1}(\hat{B})^{\mathrm{T}}PX_{t}$ is applied throughout the exploitation phase. In this case, the trajectory's performance is straightforward to analyze, and the expected cost is bounded by $\mathcal{O}(\sqrt{T})+J^{*}_{T}$.

By summing the expected costs, the total exploration cost is bounded by $\mathcal{O}(\sqrt{T})$, and the exploitation cost is bounded by $J^{*}_{T}+\mathcal{O}(\sqrt{T})$. By the definition of regret, $R_{T}=J_{T}-J^{*}_{T}$, the total regret is $\mathcal{O}(\sqrt{T})$, leading to the result of Theorem~\ref{thm:total regret}.

Our result is closely related to the result in \cite{ref1}, along with its similar version~\cite{faradonbeh2022regret}. They achieve $\mathcal{O}(\sqrt{T}\log(T))$ regret. However, they further assumes a known stabilization set for obtaining a stable controller, which is stronger compared with ours. Such difference exists because our approach detects divergence and avoids sticking to a controller which is not stable. Morever, in \cite{ref1}, the exploration and exploitation is simultaneous, where a random matrix is added to the near-optimal controller so that both $A$ and $B$ can be identified. This causes an extra $\log(T)$ factor to the regret. In contrast, our algorithm follows an explore-then-commit structure, which is enabled by the efficient system identification results presented previously. Finally, we additionally considered the setup of finite observation, which is not discussed in~\cite{ref1}.

\subsection{Experiments}

In this section, we conduct simulation experiments for the baseline algorithm and our proposed algorithm. The baseline algorithm follows the work of~\cite{ref1}. We set $d=p=3$ for simplicity. Each element of $A$ is sampled uniformly from $[-1,1]$, making $A$ unstable with high probability. The matrix $B$, $Q$, $R$ are set as the identity matrix $I_{3}$. The sampling interval is set to $h=\frac{1}{30}$.

First, we run Algorithm~\ref{alg-single} for system identification. We plot the expected Frobenius norms of the error matrices $\|\hat{A}-A\|^{2}_{F}$ and $\|\hat{B}-B\|^{2}_{F}$. The results demonstrate that our algorithm can identify $A$ and $B$ within sufficient running time or number of trajectories.

Next, we compare Algorithm~\ref{alg:3} with the baseline algorithm. We analyze the normalized regret $R(T)/T^{1/2}$ for different $t\in[600,10000]$ and plot the results in Figure~\ref{fig:experiment}. The results show that our online control algorithm with system identification achieves constant normalized regret (i.e., $O(\sqrt{T})$ regret) and outperforms the baseline algorithm when $T$ is sufficiently large.

\begin{figure*}[h]\centering
\setlength{\tabcolsep}{20pt}
\begin{tabular}{ccc}
\includegraphics[scale=0.35]{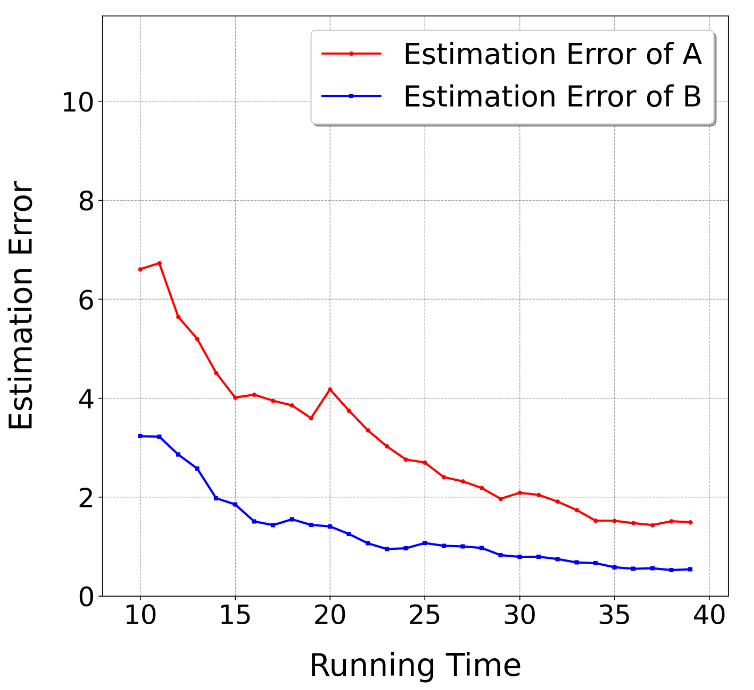} &
\includegraphics[scale=0.338]{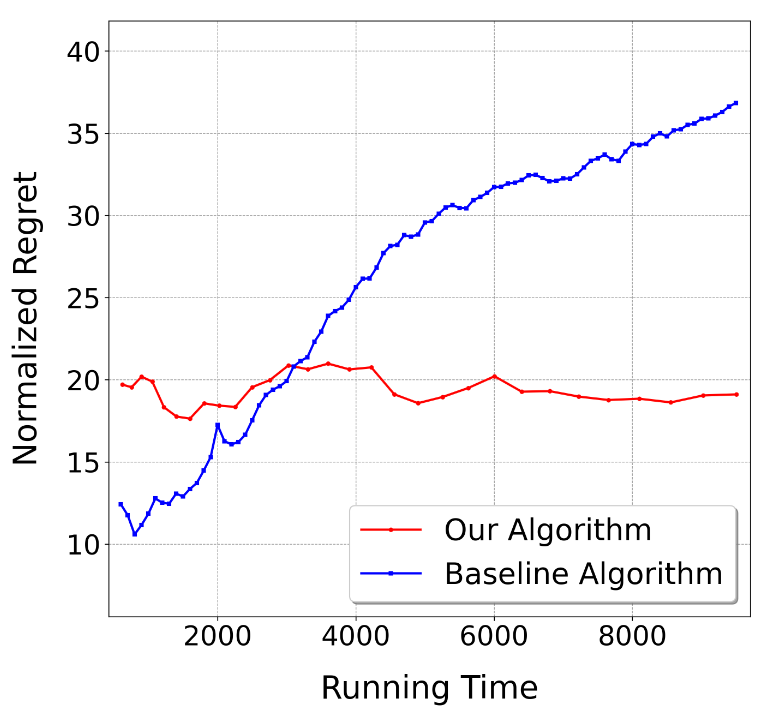} \\
\end{tabular}
\caption{\textbf{The empirical validation of our algorithm.} Left: Identification of system dynamics using a single trajectory. Right: The  normalized regret $R(T)/T^{1/2}$ of the baseline algorithm and our algorithm. The results show that our algorithm achieves small identification error and is more efficient than the baseline algorithm.} \label{fig:experiment} \vspace{-0.3cm}
\end{figure*}
\label{sec:limit}

\section{Conclusions, Limitations and Future Directions}
In this work, we establish a novel system identification method for continuous-time linear dynamical systems. This method only uses a finite number of observations and can be applied to an algorithm for online LQR continuous control which achieves $\mathcal{O}(\sqrt{T})$ regret on a single trajectory. Compared with existed works, our work not only eases the requirement for data collection and computation, but achieves fast convergence rate in identifying the unknown dynamics as well.

Although our method achieves near-optimal results in system identification and LQR online control for continuous systems with stochastic noise, many questions remain unsolved. First, it is unclear whether our system identification approach can be extended to more challenging setups, such as deterministic or adversarial noise. Additionally, many practical models are non-linear, raising the question of under what conditions discretization methods are effective. We believe these questions are crucial for real-world applications.

\bibliography{reference}

\newpage
\appendix

\section{System Identification for Continuous-time Linear System}
In this section, we analysis our system identification method in Algorithm~\ref{alg-single} and Algorithm~\ref{alg-oracle}. As a preparation, we establish some properties of matrix exponentials and their inverses.
\subsection{Matrix Exponential}
For a matrix exponential $e^{At}$, where the largest real component of $A$'s eigenvalues is denoted by $\alpha(A)$, the spectral norm of $e^{At}$ can be well-bounded~\citep{golub2013matrix}, as demonstrated in Lemma~\ref{lem:matrix}.

\begin{lemma}\label{lem:matrix}
	Suppose an $n\times n$ matrix $A$ satisfies that $0>\alpha(A)=\max\{\Re(\lambda_{i})\vert \lambda_{i}\in \lambda(A) \}$. Let $Q^{H}AQ=\mathrm{diag}(\lambda_{i})+N$ be the Schur decomposition of $A$, and let $M_{S}(t)=\sum_{k=0}^{n-1}\frac{\Vert Nt\Vert _{2}^{k}}{k!}$. Then for $t>0$, we have:
\begin{align}\label{eq:e^At}
\Vert e^{At}\Vert \leq e^{\alpha(A)tM_{s}(t)}\,,
\end{align}
\begin{align}\label{eq:e^(A+E)t}
\frac{\left\Vert e^{(A+E)t}-e^{At}\right\Vert}{\Vert e^{At}\Vert }\leq t\Vert E\Vert _{2}(M_{s}(t))^{2}e^{(tM_{S}(t)\Vert E\Vert _{2})}\,.
\end{align}	
\end{lemma}

In a special case where $\alpha(A)\leq 0$, since $M_{s}(t)\geq 1$ for all $t$, we obtain 
\begin{align*}
    \Vert e^{At}\Vert \leq e^{\alpha(A)t}\,.
\end{align*}

We also show some properties of matrix inverse in the following Lemma~\ref{lem:matrix inverse}.

\begin{lemma}[Matrix inverse]\label{lem:matrix inverse}
For any $A\in\mathbb{R}^{d\times d}$ and $t$ such that $0<\Vert At\Vert \leq \frac{1}{10}$, we have the following estimation of $e^{At}$:
\begin{align*}
    \Vert e^{At}-I_{d}\Vert \leq e^{\Vert At\Vert }-1\,,
\end{align*}
and if we denote $A_{1}=e^{At}$, then $A$ also satisfies that 
\begin{align*}
A=\frac{1}{t}\sum_{k\geq 1}\frac{(-1)^{k+1}}{k}(A_{1}-I_{d})^{k}\,.
\end{align*}

\end{lemma}
\begin{proof}
We expand $e^{At}$ by
\begin{align*}
    e^{At}=\sum_{k\geq 0}\frac{1}{k!}(At)^{k}\,,
\end{align*}
which follows that
\begin{align*}
    \Vert e^{At}-I_{d}\Vert =\left\Vert\sum_{k\geq 1}\frac{1}{k!}(At)^{k}\right\Vert\leq \sum_{k\geq 1}\frac{1}{k!}\Vert At\Vert ^{k}= e^{\Vert At\Vert }-1\leq \frac{1}{9}\,.
\end{align*}

Since $\Vert A_{1}-I_{d}\Vert <1$, the progression $A_{2}=\sum_{k\geq 1}\frac{(-1)^{k+1}}{kt}(A_{1}-I_{d})^{k}$ converges, and thus 
$e^{A_{2}t}=e^{At}$. Furthermore, it can be computed that
\begin{align*}
\Vert A_{2}t\Vert \leq \sum_{k\geq 1}\left\Vert \frac{1}{k}(A_{1}-I_{d})\right\Vert\leq \sum_{k\geq 1}\frac{1}{k}(\frac{1}{9})^{k}\leq \frac{1}{8}\,.
\end{align*}

Now we show that $A_{2}=A$. We have already known that $\Vert At\Vert $ and $\Vert A_{2}t\Vert $ are small. We also note that the function $f:X\to e^{X}$ $\left( \Vert X\Vert \leq \frac{1}{8}\right)$ constitutes a one-to-one mapping. This assertion is supported by the observation that for any  $X_{1},X_{2}$ such that $\Vert X_{1}\Vert ,\Vert X_{1}+X_{2}\Vert \leq\frac{1}{8}$, we have $\Vert X_{2}\Vert \leq \frac{1}{4}$, implying that 
\begin{align}\label{error of exponential}
\left\Vert e^{X_{1}+X_{2}}-e^{X_{1}}-X_{2}\right\Vert&=\left\Vert\sum_{k\geq 2}\frac{1}{k!}(X_{1}+X_{2})^{k}-X_{1}^{k}\right\Vert\\
&\leq\sum_{k\geq 2}\frac{1}{k!}\frac{2^{k}-1}{4^{k-1}}\Vert X_{2}\Vert \\
&\leq \frac{1}{2}\Vert X_{2}\Vert \,.
\end{align}
Then $\left\Vert e^{X_{1}+X_{2}}-e^{X_{1}}\right\Vert \geq \frac{1}{2}\Vert X_{2}\Vert $, which means $f$ is one-to-one, and thereby leading that $A_{2}=A$.

\end{proof}

\subsection{Proof of Lemma \ref{lem:tran}}
\label{supp:lem1}
\tran*

\begin{proof}
    Using Newton-Leibniz formula, we have
\begin{align*}
X_{t+h} = X_{t}+\int_{0}^{h}AX_{t+t_{1}}+BU_{t+t_{1}}+\frac{dW_{t+t_{1}}}{dt}d_{t_{1}} \,.
\end{align*}

Let $w_{t+t_1} = BU_{t+t_1}+\frac{dW_{t+t_1}}{dt}$, we have:

\begin{align*}
    X_{t+h} &= X_{t}+\int_{0}^{h}AX_{t+t_{1}}+w_{t+t_{1}}dt_{1} \\
&=(I+Ah)X_{t}+\int_{0}^{h}w_{t+t_{1}}dt_{1}+A\int_{0}^{h}(X_{t+t_{1}}-X_{t})dt_{1} \\
&=(I+Ah)X_{t}+\int_{0}^{h}w_{t+t_{1}}dt_{1}+A\int_{0}^{h}\int_{0}^{t_{1}}A{X_{t+t_{2}}}+w_{t+t_{2}}dt_{2}d_{t_{1}} \\
&=(I+Ah+\frac{1}{2}A^{2}h^{2})X_{t}+\int_{0}^{h}(I+A(h-t_{1}))w_{t+t_{1}}dt_{1}+A^{2}\int_{0}^{h}\int_{0}^{t_{1}}(X_{t+t_{2}}-X_{t})dt_{2}dt_{1} \,,
\end{align*}

where the last equality we use the Fubini theorem to change the integral order of $t_1$ and $t_2$ to calculate the second term.

Suppose we already have the following equality for integer $m$ (The case $m = 2$ has been checked above):

\begin{align*}
    X_{t+h} &= (I+\sum_{k=1}^{m}\frac{(hA)^{k}}{k!})X_{t}+\int_{0}^{h}[I+\sum_{k=1}^{m-1}\frac{{((h-t_{1})A)}^{k}}{k!}]w_{t+t_{1}}dt_{1} \\
    &+A^{m}\int_{0\le t_{m}\le ... \le h}(X_{t+t_{m}}-X_{t})dt_{1}dt_{2}...d_{t_{m}} \,.
\end{align*}

Then, replace $X_{t+t_{m}} - X_t$ by $\int_{0}^{t_{m}}[AX_{t}+A(X_{t+t_{m+1}}-X_{t})+w_{t_{m+1}}]d_{t_{m+1}}$,
we get:

\begin{align*}
    &A^{m}\int_{0\le t_{m}...\le h}(X_{t+t_{m}}-X_{t})dt_{1}dt_{2}...d_{t_{m}}\\
    =&A^{m+1}X_{t}\int_{0\le t_{m+1}\le ...\le h} dt_{1}dt_{2}...d_{t_{m+1}}
+A^{m}\int_{0\le t_{m+1}\le ...\le h}w_{t+t_{m+1}}dt_{1}dt_{2}...d_{t_{m}}\\
+&A^{m+1}\int_{0\le t_{m+1}\le ...\le h}(X_{t+t_{m+1}}-X_{t})dt_{1}dt_{2}...d_{t_{m+1}} \,.
\end{align*}

Using the property that

\begin{align*}
    \int_{0\le x_{1} \le x_{2} \le ... \le x_{m}\le h}dx_{1}dx_{2}...dx_{m}=\frac{h^{m}}{m!} \,.
\end{align*}
We finally get
\begin{align*}
    &A^{m}\int_{0\le t_{m} \le t_{m-1} \le ...\le t_1 \le h}(X_{t+t_{m}}-X_{t})dt_{1}dt_{2}...d_{t_{m}}\\
&=\frac{h^{m+1}}{(m+1)!}A^{m+1}X_{t}+A^{m}\int_{0}^{h}\frac{{(h-t_{m+1})}^{m}}{m!}w_{t+t_{m+1}}dt_{m+1}\\
&+A^{m+1}\int_{0\le t_{m+1}\le ... \le h}(X_{t+t_{m+1}}-X_{t})dt_{1}...d_{t_{m+1}} \,.
\end{align*}

In the calculation of the second term we use the Fubini theorem to change the integral order of $t_{m+1}$ and $t_1, t_2... t_m$.

So the induction hypothesis is true. For any positive integer m, we have the following equality: 

\begin{align*}
    X_{t+h} &= (I+\sum_{k=1}^{m}\frac{(hA)^{k}}{k!})X_{t}+\int_{t_{1}=0}^{h}[I+\sum_{k=1}^{m-1}\frac{{((h-t_{1})A)}^{k}}{k!}]w_{t+t_{1}}dt_{1}\\
&+A^{m}\int_{0\le t_{m}...\le h}(X_{t+t_{m}}-X_{t})dt_{1}dt_{2}...d_{t_{m}} \,.
\end{align*}

For the time interval $\Tilde{t} \in [t, t+h]$, by the continuity of $X_{\Tilde{t}}$ we know that 
$X_{\Tilde{t}}$ is uniformly bounded by some constant $C$. Therefore we have the convergence of the third term in the RHS:
\begin{align*}
    \lim_{m \rightarrow \infty} \|A^{m}\int_{0\le t_{m}...\le h}(X_{t+t_{m}}-X_{t})dt_{1}dt_{2}...d_{t_{m}} \| \le  \lim_{m \rightarrow \infty} \frac{2C (\kappa_A h)^m}{m!} = 0 \,.
\end{align*}

Therefore, we finally get:
\begin{equation}
X_{t+h} = e^{Ah}X_{t}+\int_{0}^{h}e^{A(h-s)}w_{t+s}ds
\end{equation}

Now, we use $w_{t+s} = BU_{t+s}+\frac{dW_{t+s}}{dt}$, we get:

\begin{align}
X_{t+h} &= e^{Ah}X_{t}+\int_{0}^{h}e^{A(h-s)}BU_{t+s}d{s}+\int_{0}^{h}e^{A(h-s)}dW_{t+s} \\
&= e^{Ah}X_{t}+\int_{0}^{h}e^{A(h-s)}BU_{t+s}d{s} + w_t \,,
\end{align}
where $w_{t}$ is Gaussian noise $\mathcal{N}(0, \Sigma)$ with covariance $\Sigma = \int_{0}^{h}e^{As}e^{A^{\mathrm{T}}s}ds$.

\end{proof}

\subsection{Proof of Lemma~\ref{lem:error-transform}}
We restate Lemma~\ref{lem:error-transform} and provide the proof here.
\paragraph{Lemma~\ref{lem:error-transform}}
In Algorithm~\ref{alg-single},~\ref{alg-oracle}, suppose we have obtained the relative error $\|\widetilde{A}-A^{'}\|,\|\widetilde{B}-B^{'}\|\leq \epsilon$ for some $\epsilon\leq \frac{1}{15}$ and  $\|Ah\|\leq \frac{1}{15}$, then we have the following relative error of the primal system: 
	\begin{align}
		\|\hat{A}-A\|,\|\hat{B}-B\|\leq \frac{C}{h}\epsilon\,,
	\end{align}
 where $C$ is a constant independent of $h$.
\begin{proof}
Firstly, according to Lemma~\ref{lem:matrix inverse}, the estimated $\widetilde{A}$ is not too far away from $I_{d}$, as we have:
\begin{align*}
\left\Vert \widetilde{A}-I_{d}\right\Vert \leq \left\Vert \widetilde{A}-e^{Ah}\right\Vert +\left\Vert e^{Ah}-I_{d}\right\Vert \leq \epsilon+e^{\Vert A\Vert h}-1\leq\frac{1}{7}\,,
\end{align*}

Then, from \eqref{eq:recover} we can bound the matrix norm $\left\Vert\hat{A}h\right\Vert$ by
\begin{align*}
\left\Vert \hat{A}h\right\Vert=\left\Vert\sum_{k\geq 1}\frac{(-1)^{k-1}}{k}(\widetilde{A}-I)^{k}\right\Vert\leq\sum_{k\geq 1}\frac{1}{k}(\frac{1}{7})^{k}\leq \frac{1}{6}\,.
\end{align*}
Now, let's denote $A_{1}=Ah$ and $A_{2}=\hat{A}h-A_{1}$, satisfying the relations $A^{'}=e^{A_{1}}$ and $\widetilde{A}=e^{A_{1}+A_{2}}$. It is given that $\Vert A_{1}\Vert \leq\frac{1}{15}$ and $\Vert A_{2}\Vert \leq \Vert A_{1}\Vert +\Vert \hat{A}h\Vert \leq \frac{1}{4}$, so by ~\eqref{error of exponential}, we obtain that $\Vert \hat{A}-A\Vert h=\Vert A_{2}\Vert \leq 2\Vert \widetilde{A}-A^{'}\Vert $, which follows that $\Vert \hat{A}-A\Vert \leq\frac{2}{h}\Vert \widetilde{A}-A^{'}\Vert \leq \frac{2}{h}\epsilon$.
\end{proof}

Next, we will upper bound the estimation error of $B$. Let $A_{h}=\int_{t=0}^{h}e^{At}dt$ and $\bar{A}_{h}=\int_{t=0}^{h}e^{\hat{A}t}dt$, satisfying
\begin{align*}
&\Vert A_{h}-hI\Vert =\left\Vert\int_{t=0}^{h}(e^{At}-I)dt\right\Vert\leq \int_{t=0}^{h}\left\Vert e^{At}-I\right\Vert dt\leq\int_{t=0}^{h}(e^{\Vert A\Vert t}-1)dt\leq \frac{1}{20}h\,,\\
&\Vert \bar{A}_{h}-A_{h}\Vert =\left\Vert \int_{t=0}^{h}e^{\hat{A}t}-e^{At}dt\right\Vert \leq \int_{t=0}^{h}\left\Vert e^{\hat{A}t}-e^{At}\right\Vert dt\leq \frac{3}{2}\int_{t=0}^{h}\Vert \hat{A}-A\Vert tdt\leq \frac{3}{4}h\epsilon\,.
\end{align*}
This follows that 
\begin{align*}
&\left\Vert A_{h}^{-1}\right\Vert=\frac{1}{h}\left\Vert\left[I+(\frac{A_{h}}{h}-I)\right]^{-1}\right\Vert\leq\frac{1}{h}\sum_{k\geq 0}\left\Vert \frac{A_{h}}{h}-I\right\Vert^{k}\leq \frac{20}{19h}\,,\\
&\Vert (\bar{A}_{h})^{-1}-A_{h}^{-1}\Vert \\
&=\left\Vert A_{h}^{-1}\right\Vert \left\Vert \left[I+(\bar{A}_{h}-A_{h})A_{h}^{-1}\right]^{-1}-I\right\Vert \leq \left\Vert A^{-1}_{h}\right\Vert \frac{1}{1-\left\Vert (\bar{A}_{h}-A_{h})A^{-1}_{h}\right\Vert}\leq  \frac{1}{h}\epsilon\,.
\end{align*}

Since $B$ and its estimator $\hat{B}$ satisfy that
\begin{align*}
B=(A_{h})^{-1}B^{'},\hat{B}=(\bar{A}_{h})^{-1}\widetilde{B}\,,
\end{align*}
we can upper bound the estimation error $\left\Vert \hat{B}-B\right\Vert$ by
\begin{align*}
\left\Vert\hat{B}-B\right\Vert\leq \left\Vert(\bar{A}_{h})^{-1}-A_{h}^{-1}\right\Vert \left\Vert B^{'}\right\Vert+\left\Vert(\bar{A}_{h})^{-1}\right\Vert \left\Vert \widetilde{B}-B^{'}\right\Vert \leq \frac{\Vert B^{'}\Vert }{h}\epsilon+\frac{2}{h}\epsilon\leq(2\Vert B\Vert +\frac{2}{h})\epsilon\,,
\end{align*}
where the last inequality is because $\Vert B^{'}\Vert \leq \Vert A_{h}\Vert \Vert B\Vert \leq 2h\Vert B\Vert $.

Since $2\Vert B\Vert\leq 2\kappa_{B}\leq \frac{1}{h}\cdot \frac{2\kappa_{B}}{15\kappa_{A}}\leq \frac{\kappa_{B}}{\kappa_{A}}$, we obtain Lemma~\ref{lem:error-transform}.

\subsection{Proof of Theorem~\ref{thm:main-stable}}
\label{sec:upper bound}

In this section, we derive the proof of Theorem~\ref{thm:main-stable}. We upper bound the estimation errors of intermediate dynamics $(A^{'},B^{'})$, obtained as in~\eqref{eq:alg1}.
We first prove Lemma~\ref{lem:A,B} below, providing system identification results on a single trajectory with a stable controller.

\begin{lemma}\label{lem:A,B}
Consider the trajectory $x_{k+1}=Ax_{k}+Bu_{k}+w_{k}$ with $A\in \mathbb{R}^{d\times d}$, $\|A\|<1$, $B\in \mathbb{R}^{d\times p}$; $u_{k}\sim\mathcal{N}(0,I_{p})$ and $w_{k}\sim\mathcal{N}(0,\Sigma)$ are i.i.d. random variables. Suppose we compute $(\hat{A},\hat{B})$ by	
\begin{align}\label{eq:lem:A,B}
(\hat{A})^{\mathrm{T}}=\left[\sum_{k=0}^{T_{0}-1}x_{k}x^{\mathrm{T}}_{k}\right]^{\dagger}\sum_{k=0}^{T_{0}-1}x_{k}x^{\mathrm{T}}_{k+1}\,,(\hat{B})^{\mathrm{T}}=\left[\sum_{k=0}^{T_{0}-1}u_{k}u^{\mathrm{T}}_{k}\right]^{\dagger}\sum_{k=0}^{T_{0}-1}u_{k}\left(x_{k+1}-\hat{A}x_{k}\right)^{\mathrm{T}}\,.
\end{align}	

Then there exists a constant $C$ (depending only on $A$, $B$, $d$, $p$ and $\Sigma$) such that for $T\geq C\left(\|X_{0}\|^{2}_{2}+\log^{2}(1/\delta)\right)$, w.p. at least $1-\delta$:
\begin{align}
\|\hat{A}-A\|,\|\hat{B}-B\|\leq C\sqrt{\frac{\log(1/\delta)}{T}}\,,
\end{align}
\end{lemma}

We first provide Lemma~\ref{lem:A}, which is used as the base of Lemma~\ref{lem:A,B}.
\begin{lemma}\label{lem:A}
Consider $A\in \mathbb{R}^{d\times d}$ such that $\rho(A)<1$ and the system $X_{k+1}=AX_{k}+w_{k}$ with $w_{k}\sim\mathcal{N}(0,\Sigma)$ be i.i.d. random variables. Suppose we estimate $A$ as in~\eqref{eq:alg1}. 
Then there exists a constant $C$ depending on $A$, $\Sigma$ and $d$ such that for $T\geq C(\Vert X_{0}\Vert ^{2}_{2}+\log(1/\delta))$, w.p. at least $1-\delta$, we have:
\begin{align*}
\Vert \hat{A}-A\Vert \leq C\sqrt{\frac{\log(1/\delta)}{T}}\,.
\end{align*}
\end{lemma}

The work of~\citep{ref4} has discussed such systems in their Theorem 2.4, and we list it below:

\begin{theorem}\label{thm:simchowitz}
Fix $\epsilon,\delta\in(0,1)$, $T\in\mathbb{N}$ and $0\prec \Gamma_{sb}\prec\bar{\Gamma}$. Then if $(X_{t},Y_{t})_{t\geq1}\in (\mathbb{R}^{d}\times\mathbb{R}^{n})^{T}$ is a random sequence such that (a) $Y_{t}=A_{*}X_{t}+\eta_{t}$, where $\eta_{t}\vert \mathcal{F}_{t}$ is $\sigma^{2}$-sub-Gaussian and mean zero, (b) $X_{1},...,X_{T}$ satisfies the $(k,\Gamma_{sb},p)$-small ball condition, and (c) such that $\mathbb{P}\left[\sum_{t=1}^{T}X_{t}X^{\mathrm{T}}_{t}\not\preceq T\bar{\Gamma}\right]\leq \delta$. Then if
\begin{align*}
T\geq\frac{10k}{p^{2}}\left(\log(1/\delta)+2d\log(10/p)+\log \det(\bar{\Gamma}\Gamma^{-1}_{sb})\right)\,,
\end{align*}
we have
\begin{align*}
\mathbb{P}\left[\Vert \hat{A}-A_{*}\Vert >\frac{90\sigma}{p}\sqrt{\frac{n+d\log\frac{10}{p}+\log \det(\bar{\Gamma}\Gamma^{-1}_{sb})+\log(\frac{1}{\delta})}{T\lambda_{\min}(\Gamma_{sb})}}\right]\leq 3\delta\,.
\end{align*}
Here, the $(k,\Gamma_{sb},p)$-small ball condition is defined as follows. Let $(Z_{t})_{t\geq1}$ be an ${\mathcal{F}_{t}}_{t\geq1}$-adapted random process taking values in $\mathbb{R}$. We say $(Z_{t})_{t\geq1}$ satisfies the $(k,\nu,p)$-block martingale small-ball (BMSB) condition if, for any $j\geq0$, one has $\frac{1}{k}\sum_{i=1}^{k}\mathbb{P}(\vert Z_{j+i}\vert \geq\nu \vert \mathcal{F}_{j})\geq p$ almost surely. Given a process $(X_{t})_{t\geq1}$ taking values in $\mathbb{R}^{d}$, we say that it satisfies the $(k,\Gamma_{sb},p)$-BMSB condition for $\Gamma_{sb}\succ 0$ if, for any fixed $w\in\mathcal{S}^{d-1}$, the process $Z_{t}:=\langle w,X_{t}\rangle$ satisfies $(k,\sqrt{w^{\mathrm{T}}\Gamma_{sb}w},p)$-BMSB.
\end{theorem}

In the work of~\citep{ref4}, they have discussed the case when $X_{0}=0$, and now we modify it to a general starting state $X_{0}$.
From~\eqref{eq:alg1}, we derive the estimation error of $A$ as
\begin{align*}
\hat{A}^{\mathrm{T}}-A^{\mathrm{T}}&=\left[\sum_{k=0}^{T-1}X_{k}X^{\mathrm{T}}_{k}\right]^{\dagger}\sum_{k=0}^{T-1}X_{k}X^{\mathrm{T}}_{k+1}-A^{\mathrm{T}}\\
&=\left[\sum_{k=0}^{T-1}X_{k}X^{\mathrm{T}}_{k}\right]^{\dagger}\sum_{k=0}^{T-1}X_{k}(AX_{k}+w_{k})^{\mathrm{T}}-A^{\mathrm{T}}\\
&=\left[\sum_{k=0}^{T-1}X_{k}X^{\mathrm{T}}_{k}\right]^{\dagger}\sum_{k=0}^{T-1}X_{k}w^{\mathrm{T}}_{k}\,.
\end{align*}
For the first term, consider any $v\in\mathcal{S}^{d-1}$, we lower bound $v^{\mathrm{T}}\left(\sum_{k=0}^{T-1}X_{k}X^{\mathrm{T}}_{k}\right)v$.
Let $a_{k}=v^{\mathrm{T}}X_{k}$, then  $a_{k}=v^{\mathrm{T}}AX_{k-1}+v^{\mathrm{T}}w_{k}$.
We claim that for any $k\geq 1$, $\mathbb{P}\left[\vert a_{k}\vert \geq \frac{1}{2}\vert X_{k-1}\right]\geq\frac{1}{2}$.
Let $b_{k}=v^{\mathrm{T}}w_{k}$, which is independent of $X_{k-1}$. It suffices to show that for any $c\in\mathbb{R}$, $\mathbb{P}\left[b_{k}\in [c,c+1]\right]\leq \frac{1}{2}$.

Since $\Vert v\Vert _{2}=1$ and $w_{k}\sim\mathcal{N}(0,I_{d})$, we have $b_{k}\sim\mathcal{N}(0,1)$, from which we estimate the probability as
\begin{align}\label{eq:bmsb}
\mathbb{P}\left[b_{k}\in [c,c+1]\right]=\int_{x=c}^{c+1}\frac{1}{\sqrt{2\pi}}e^{-\frac{1}{2}x^{2}}dx\leq\frac{1}{\sqrt{2\pi}}\leq\frac{1}{2}\,.
\end{align}

Based on \eqref{eq:bmsb}, we can simply choose $k=1$, $\Gamma_{sb}=\frac{1}{4}I_{d}$ and $p=\frac{1}{2}$, then the random sequence $(X_{i})_{i\geq0}$ satisfies the $(k,\Gamma_{sb},p)$-BMSB condition. It remains to choose a proper $\bar{\Gamma}$ that meets the condition (c) in Theorem~\ref{thm:simchowitz}.

Since $X_{k}=A^{k}X_{0}+\sum_{i=1}^{k}A^{k-i}w_{i}$, we have:
\begin{align*}
\mathbb{E}\left[\sum_{k=0}^{T-1}X_{k}X^{\mathrm{T}}_{k}\right]&=\mathbb{E}\left[\sum_{k=0}^{T-1}\left(A^{k}X_{0}+\sum_{i=1}^{k}A^{k-i}w_{i}\right)\left(A^{k}X_{0}+\sum_{i=1}^{k}A^{k-i}w_{i}\right)^{\mathrm{T}}\right]\\
&=\sum_{k=0}^{T-1}A^{k}X_{0}X_{0}^{\mathrm{T}}(A^{k})^{\mathrm{T}}+\mathbb{E}\left[\sum_{k=0}^{T-1}\sum_{i=0}^{k}A^{k}X_{0}w_{i}^{\mathrm{T}}(A^{k-i})^{\mathrm{T}}\right]\\
&+\mathbb{E}\left[\sum_{k=0}^{T-1}\sum_{i=0}^{k}A^{k-i}w_{i}X_{0}^{\mathrm{T}}(A^{k})^{\mathrm{T}}\right]+\mathbb{E}\left[\sum_{k=0}^{T-1}\sum_{i,j=0}^{k}A^{k-i}w_{i}w^{\mathrm{T}}_{j}(A^{k-j})^{\mathrm{T}}\right]\\
&=\sum_{k=0}^{T-1}A^{k}X_{0}X_{0}^{\mathrm{T}}(A^{k})^{\mathrm{T}}+\sum_{k=0}^{T-1}\sum_{i=0}^{k}A^{k-i}\Sigma(A^{k-i})^{\mathrm{T}}\,.
\end{align*}

Let $\Gamma_{\infty}=\sum_{k\geq0}A^{k}\Sigma(A^{k})^{\mathrm{T}}$ which is bounded and $C_{1}$ be a constant such that $C_{1}\geq\sum_{k\geq0}\Vert A^{k}\Vert ^{2}$. We then show that for $\bar{\Gamma}=\left(\frac{C_{1}\Vert X_{0}\Vert ^{2}_{2}}{T}dI_{d}+d\Vert \Gamma_{\infty}\Vert I_{d}\right)/\delta$, the condition (c) in Theorem~\ref{thm:simchowitz} is satisfied. This is because $\mathbb{E}\left[\mathrm{tr}\left(\sum_{k=0}^{T-1}X_{k}X_{k}^{\mathrm{T}}\right)\right]=\mathrm{tr}\left(\mathbb{E}\left[\sum_{k=0}^{T-1}X_{k}X_{k}^{\mathrm{T}}\right]\right)\leq \frac{T\delta}{d}\mathrm{tr}(\bar{\Gamma})$ so that 
\\$\mathbb{P}\left[\mathrm{tr}(\sum_{k=0}^{T-1}X_{k}X_{k}^{\mathrm{T}})\geq \frac{1}{d}T\mathrm{tr}(\bar{\Gamma})\right]\leq\delta$. Furthermore, a necessary condition for \\$\sum_{k=0}^{T-1}X_{k}X^{\mathrm{T}}\not\prec T\bar{\Gamma}$ is $\mathrm{tr}(\sum_{k=0}^{T-1}X_{k}X^{\mathrm{T}})\geq \frac{1}{d}T\mathrm{tr}(\bar{\Gamma})$.

Now, we apply such $\bar{\Gamma}$ to Theorem~\ref{thm:simchowitz}. It can be computed that \begin{align*}
\log\det(\bar{\Gamma}\Gamma^{-1}_{sb})=d\log\left(4d(C_{1}\Vert X_{0}\Vert ^{2}_{2}/T+\Vert \Gamma_{\infty}\Vert )\right)+d\log(1/\delta)\,.
\end{align*}
Then when $T\geq C_{1}\Vert X_{0}\Vert ^{2}$ as well as $T\geq 40\left(2d\log(20)+d\log(4d(1+\Vert \Gamma_{\infty}\Vert ))+2d\log(1/\delta)\right)$, we have:
\begin{align*}
\mathbb{P}\left[\Vert \hat{A}-A\Vert >360\sqrt{\frac{d+d\log(20)+d\log(4d(1+\Vert \Gamma_{\infty}\Vert ))+2d\log(\frac{1}{\delta})}{T}}\,\right]\leq 3\delta\,.
\end{align*}
This implies our Lemma~\ref{lem:A}.

\paragraph{Proof of Lemma~\ref{lem:A,B}}
As for the estimation error $\Vert \hat{A}-A\Vert $, let $w^{'}_{k}=Bu_{k}+w_{k}\sim\mathcal{N}(0,\Sigma+BB^{\mathrm{T}})$, which form a sequence of $i.i.d$ random variables. With the results in Lemma~\ref{lem:A}, there exist some constants $C_{1},C_{2}$ such that, as long as $T\geq C_{1}\left(\Vert X_{0}\Vert ^{2}_{2}+\log(1/\delta)\right)$ we have:
\begin{align*}
\Vert \hat{A}-A\Vert \leq C_{2}\sqrt{\frac{\log(1/\delta)}{T}}\,.
\end{align*}

Now we upper bound the estimation error $\Vert \hat{B}-B\Vert $. With the expression in~\eqref{eq:alg1}, we obtain:
\begin{align*}
\Vert \hat{B}-B\Vert &=\left\Vert\left[\sum_{k=0}^{T-1}u_{k}u^{\mathrm{T}}_{k}\right]^{\dagger}\sum_{k=0}^{T-1}\textcolor{blue}{u_{k}}\left[(A-\hat{A})X_{k}+w_{k}\right]^{\mathrm{T}}\right\Vert\\
&\leq\lambda^{-1}_{\min}(\sum_{k=0}^{T-1}u_{k}u^{\mathrm{T}}_{k})\left[\left\Vert\sum_{k=0}^{T-1}u_{k}X^{\mathrm{T}}_{k}\right\Vert \left\Vert\hat{A}-A\right\Vert+\left\Vert\sum_{k=0}^{T-1}u_{k}w^{\mathrm{T}}_{k}\right\Vert\right]\,.
\end{align*}
For the quantities $\lambda^{-1}_{\min}(\sum_{k=0}^{T-1}u_{k}u^{\mathrm{T}}_{k})$ and $\Vert \sum_{k=0}^{T-1}u_{k}w^{\mathrm{T}}_{k}\Vert $, we apply Lemma 2.1. and Lemma 2.2. in the work of~\citep{ref5}, where they present the following results.
\begin{lemma}
Let $N\geq 2\log(1/\delta)$. Suppose $f_{k}\in\mathbb{R}^{m}$, $g_{k}\in\mathbb{R}^{n}$ are independent vectors such that $f_{k}\sim\mathcal{N}(0,\Sigma_{f})$ and $g_{k}\sim\mathcal{N}(0,\Sigma_{g})$ for $1\leq k\leq N$. With probability at least $1-\delta$,
\begin{align*}
\left\Vert\sum_{k=1}^{N}f_{k}g^{\mathrm{T}}_{k}\right\Vert\leq 4\Vert \Sigma_{f}\Vert ^{1/2}_{2}\Vert \Sigma_{g}\Vert ^{1/2}_{2}\sqrt{N(m+n)\log(9/\delta)}\,.
\end{align*}
\end{lemma}
\begin{lemma}
Let $X\in\mathbb{R}^{N\times n}$ have $i.i.d.$ $\mathcal{N}(0,1)$ entries. With probability at least $1-\delta$,
\begin{align*}
\sqrt{\lambda_{\min}(X^{\mathrm{T}}X)}\geq \sqrt{N}-\sqrt{n}-\sqrt{2\log(1/\delta)}\,.
\end{align*}
\end{lemma}
With these two lemmas, we can conclude that if $T\geq 32(d+p)\log(4/\delta)$, then both
$\lambda_{\min}(u_{k}u^{\mathrm{T}}_{k})\geq\frac{1}{2}T$ and $\left\Vert\sum_{k=0}^{T-1}u_{k}w^{\mathrm{T}}_{k}\right\Vert \leq 4\left\Vert\Sigma\right\Vert^{1/2}_{2}\sqrt{T(d+p)\log(18/\delta)}$, w.p. at least $1-\delta$.

Now we concentrate on the term $\left\Vert \sum_{k=0}^{T-1}u_{k}X^{\mathrm{T}}_{k}\right\Vert$. Since $w^{'}_{i}=Bu_{i}+w_{i}\sim\mathcal{N}(0,\Sigma+BB^{\mathrm{T}})$, it can be directly computed that, w.p. at least $1-\delta/T$, $\left\Vert w^{'}_{i}\right\Vert_{2}\leq 2\left\Vert d(\Sigma+BB^{\mathrm{T}})\right\Vert ^{1/2}_{2}\sqrt{\log(T/\delta)}$. Then by union bound we get $\mathbb{P}\left[\sup_{0\leq i\leq T-1}\left\Vert w^{'}_{i}\right\Vert _{2}\leq 2\Vert \Sigma+BB^{\mathrm{T}}\Vert ^{1/2}_{2}\sqrt{d\log(T/\delta)} \right]\leq \delta$. Furthermore, when $\sup_{0\leq i\leq T-1}\left\Vert w^{'}_{i}\right\Vert_{2}\leq 2\left\Vert \Sigma+BB^{\mathrm{T}}\right\Vert ^{1/2}_{2}\sqrt{d\log(T/\delta)}$, we must have 
\begin{align}\label{eq:X}
\Vert X_{k}\Vert _{2}=\left\Vert A^{k}X_{0}+\sum_{i=0}^{k-1}A^{k-1-i}w_{i}\right\Vert \leq \Vert A\Vert ^{k}\Vert X_{0}\Vert _{2}+\frac{2}{1-\Vert A\Vert }\left\Vert\Sigma+BB^{\mathrm{T}}\right\Vert^{1/2}_{2}\sqrt{d\log(T/\delta)}\,.
\end{align}

For any $u\in\mathcal{S}^{p-1}$ and $v\in\mathcal{S}^{d-1}$, let $x_{i}=u^{\mathrm{T}}u_{i}(0\leq i\leq T-1)$. Then, $x_{i}$ follows a normal distribution $x_{i}\sim\mathcal{N}(0,1)$ and $\{x_{i}\}$ is a sequence of independent random variables. Furthermore, $x_{k}$ is also independent of $(X_{i})_{0\leq i\leq k}$. On the other hand, denote $y_{k}=X^{\mathrm{T}}_{k}v$, \eqref{eq:X} implies that w.p. at least $1-\delta$, for all $k$ we have
$\vert y_{k}\vert \leq \Vert X_{0}\Vert _{2}+\frac{2}{1-\Vert A\Vert }\left\Vert\Sigma+BB^{\mathrm{T}}\right\Vert^{1/2}_{2}\sqrt{d\log(T/\delta)}:=Y$. Let 
\begin{align*}
Z_{k}:=\sum_{i=0}^{k}u^{\mathrm{T}}\left(u_{k}X^{\mathrm{T}}_{k}\right)v\cdot 1_{\Vert X_{k}\Vert _{2}\leq Y}=\sum_{i=0}^{k}x_{k}y_{k}\cdot1_{\Vert X_{k}\Vert _{2}\leq Y}\,,  
\end{align*}
and let $\mathcal{F}_{0},\mathcal{F}_{1},...,\mathcal{F}_{T}$ be the filtration of $X_{0},X_{1},...,X_{T}$,  then for any $\alpha\geq 0$,
\begin{align*}
\mathbb{E}\left[e^{\frac{\alpha Z_{k+1}}{Y}}\vert \mathcal{F}_{k}\right]= e^{\frac{\alpha Z_{k}}{Y}}\mathbb{E}_{X_{k+1}}\left[\mathbb{E}_{x\sim\mathcal{N}(0,1)}\left[\exp\left(\frac{\alpha xy_{k+1}\cdot 1_{\Vert X_{k+1}\Vert _{2}\leq Y}}{Y}\right)\right]\right]\leq e^{\frac{1}{2}\alpha^{2}}e^{\frac{\alpha Z_{k}}{Y}}\,,
\end{align*}
implying that $\mathbb{E}\left[e^{\frac{\alpha Z_{k+1}}{Y}}\right]\leq e^{\frac{1}{2}\alpha^{2}}\mathbb{E}\left[e^{\frac{\alpha Z_{k+1}}{Y}}\right] $
So we have: $\mathbb{E}\left[e^{\frac{\alpha Z_{T-1}}{Y}}\right]\leq e^{\frac{1}{2}\alpha^{2}T}$.
By choosing $\alpha = \pm \sqrt{\frac{1}{T}}$, we obtain that
\begin{align*}
\mathbb{P}\left[\vert Z_{T-1}\vert\geq 2Y\sqrt{T\log(4/\delta)}\right]\leq \delta
\end{align*}
For $\mathcal{T}_{d}$ be a $\frac{1}{4}$-net of $\mathcal{S}^{d-1}$ and $\mathcal{T}_{p}$ be a $\frac{1}{4}$-net of $\mathcal{S}^{p-1}$, we use union bound on them and obtain that, w.p. at least $1-\delta$
\begin{align*}
\vert Z_{T-1}\vert \leq 2Y\sqrt{T\log(4\vert \mathcal{T}_{p}\vert \vert \mathcal{T}_{d}\vert/\delta )}\leq 2Y\sqrt{T[4(d+p)+\log(4/\delta)]}\,.
\end{align*}
Where the last inequality is because $\vert \mathcal{T}_{p}\vert \leq 9^{p}$ and $\vert \mathcal{T}_{d}\vert \leq 9^{d}$

Next we upper bound $\left\Vert\sum_{k=0}^{T-1}u_{k}X^{\mathrm{T}}_{k}\right\Vert$. For any $u_{*}\in\mathcal{S}^{p-1}$ and $v_{*}\in\mathcal{S}^{p-1}$, with some $u\in\mathcal{T}_{p}$ and $v\in\mathcal{T}_{d}$ $s.t.$ $\Vert u-u_{*}\Vert _{2},\Vert v-v_{*}\Vert _{2}\leq \frac{1}{2}$, we have:
\begin{align*}
&\left\vert u_{*}^{\mathrm{T}}(\sum_{k=0}^{T-1}u_{k}X^{\mathrm{T}}_{k})v_{*}\right\vert\\
&\leq\left\vert u^{\mathrm{T}}(\sum_{k=0}^{T-1}u_{k}X^{\mathrm{T}}_{k})v\right\vert+\left\vert(u_{*}-u)^{\mathrm{T}}(\sum_{k=0}^{T-1}u_{k}X^{\mathrm{T}}_{k})v_{*}\right\vert+\left\vert u^{\mathrm{T}}(\sum_{k=0}^{T-1}u_{k}X^{\mathrm{T}}_{k})(v-v_{*})\right\vert\\
&\leq \sup_{u\in \mathcal{T}_{p},v\in\mathcal{T}_{d}}\left\vert u^{\mathrm{T}}(\sum_{k=0}^{T-1}u_{k}X^{\mathrm{T}}_{k})v\right\vert +\frac{1}{2}\left\Vert\sum_{k=0}^{T-1}u_{k}X^{\mathrm{T}}_{k}\right\Vert\,.
\end{align*}
This leads $\left\Vert\sum_{k=0}^{T-1}u_{k}X^{\mathrm{T}}_{k}\right\Vert\leq 2\sup_{u\in \mathcal{T}_{p},v\in\mathcal{T}_{d}}\left\vert u^{\mathrm{T}}(\sum_{k=0}^{T-1}u_{k}X^{\mathrm{T}}_{k})v\right\vert$. Therefore, for any $\delta\in(0,\frac{1}{2})$, we have:
\begin{align*}
&\mathbb{P}\left[\left\Vert\sum_{k=0}^{T-1}u_{k}X^{\mathrm{T}}_{k}\right\Vert_{2}\geq 4Y\sqrt{T[4(d+p)+\log(4/\delta)]}\right]\\
&\leq \mathbb{P}\left[\sup_{u\in\mathcal{T}_{d},v\in\mathcal{T}_{p}}\left\vert u^{\mathrm{T}}\left(\sum_{k=0}^{T-1}u_{k}X^{\mathrm{T}}_{k}1_{\Vert X_{k}\Vert_{2}\leq Y}\right)v\right\vert\geq 2Y\sqrt{T[4(d+p)+\log(4/\delta)]}\right]\\
&+\mathbb{P}\left[\exists\,0\leq k\leq T-1, \Vert X_{k}\Vert _{2}\geq Y\right]\\
&\leq 2\delta\,.
\end{align*}

We choose constant $C$ depending on $A,B,d,p$ such that for all $T\geq C\left(\Vert X_{0}\Vert ^{2}_{2}+\log^{2}(1/\delta)\right)$, 
\begin{align*}
4Y\sqrt{T[4(d+p)+\log(4/\delta)]}\leq T\,, 
\end{align*}
and we further have: whenever $T\geq C\left(\Vert X_{0}\Vert ^{2}_{2}+\log^{2}(1/\delta)\right)$, w.p. at least $1-3\delta$,
\begin{align*}
\left\Vert\sum_{k=0}^{T-1}u_{k}X^{\mathrm{T}}_{k}\right\Vert\Vert \hat{A}-A\Vert \leq C_{2}\sqrt{\log(1/\delta)T}\,.
\end{align*}

Finally, when $T\geq \max\left(C\left(\Vert X_{0}\Vert ^{2}_{2}+\log^{2}(1/\delta)\right),32(d+p)\log(4/\delta)\right)$, we combine this upper bound with $\mathbb{P}\left(\lambda_{\min}(\sum_{k=0}^{T-1}u_{k}u^{\mathrm{T}}_{k})\leq\frac{1}{2}T\right)\leq \delta$,
and obtain Lemma~\ref{lem:A,B}.

\subsection{Proof of Theorem~\ref{thm:main-unstable}}\label{subsection: multiple trajectories}
Now, we aim to establish Theorem~\ref{thm:main-unstable}. The analysis of system identification for discrete-time linear dynamical systems with multiple trajectories has been thoroughly investigated by \citep{ref5}. We hereby cite their findings, denoting the relevant result as Lemma~\ref{lem:multi-trajectory}.

\begin{lemma}\label{lem:multi-trajectory}
Suppose we have $N$ i.i.d. trajectories $X^{i}_{k}$, each is defined by $X^{i}_{(k+1)h}=AX^{i}_{k}+Bu^{i}_{k}+w^{i}_{k}$, where $T_{0}$ is any integer, $u^{i}_{k}\sim\mathcal{N}(0,I_{p})$ and $w^{i}_{k}\sim\mathcal{N}(0,\Sigma)$ are two sets of i.i.d. random variables. Then, for the estimator $(\hat{A},\hat{B})$ of 
\begin{align}
(\hat{A},\hat{B})\in \arg\min_{(A,B)}\frac{1}{2}\sum_{i=1}^{N}\left\Vert X^{i}_{T_{0}}-AX^{i}_{T_{0}-1}-Bu^{i}_{T_{0}-1}\right\Vert^{2}_{2}
\end{align}
with probability at least $1-\delta$, we have:
	\begin{align*}
	\Vert \hat{A}-A\Vert ,\Vert \hat{B}-B\Vert \leq \mathcal{O}\left(\sqrt{\frac{\log(1/\delta)}{N}}\right)\,.
	\end{align*}
\end{lemma}

Combining Lemma~\ref{lem:multi-trajectory} with Lemma~\ref{lem:error-transform}, we directly obtain Theorem~\ref{thm:main-unstable}.

\subsection{Lower Bound of System Identification with Finite Observation}
\label{sec:lower bound}
We restate and provide the proof of Theorem~\ref{thm:lower bound}.

\paragraph{Theorem~\ref{thm:lower bound}}
Suppose $T\geq 1$ be the running time of a single trajectory of continuous-time linear differential system, represented as in~\eqref{equ:differential}. Then there exist constants $c_{1},c_{2}$ independent of $d$ such that, for any finite set of observed points $\{t_{0}=0,t_{1},t_{2},...,t_{N}=T\}$, and any (possibly randomized) estimator function $\phi : \{X_{t_{0}},X_{t_{1}},...,X_{t_{N}}\}\to \mathbb{R}^{d\times d}$, there exists bounded $A,B$ satisfying $\mathbb{P}\left[\Vert \phi(\{X_i\}_{i \le N}) - A\Vert\geq \frac{c_{1}}{\sqrt{T}}\right]\geq c_{2}$. Here the probability corresponds to the dynamical system dominated by $(A,B)$.

\begin{proof}
Firstly, we consider a special case where $d=1$, and let $A=[-1]$ and $\bar{A}=[-1-\delta]$. We show that when $\delta= \frac{1}{5\sqrt{T}}$, for the two dynamical systems $\psi_{\theta}:dX_{t}=AX_{t}dt+dW_{t}$ and $\psi_{\bar{\theta}}:dX_{t}=\bar{A}X_{t}dt+dW_{t}$, any algorithm $\mathcal{A}$ that outputs according only to $\{X_{t_{0}},X_{t_{1}},...,X_{t_{N}}\}
$ satisfies:

\begin{align*}
&\max\left\{\mathbb{P}\left[\|\mathcal{A}(X_{t_{0}},X_{t_{1}},...,X_{t_{N}})-A\|\geq\frac{1}{10\sqrt{T}}     \right], \mathbb{P}\left[\|\mathcal{A}(X_{t_{0}},X_{t_{1}},...,X_{t_{N}})-\bar{A}\|\geq\frac{1}{10\sqrt{T}}     \right]   \right\}\\&\geq \frac{1}{4e^{3}}\,.
\end{align*}

We note that this special case can be easily generalized to any dimension $d$, since we can consider $A=-I_{d}$ and $\bar{A}$ satisfies $\bar{A}_{1,1}=A_{1,1}-\delta$, and for any $(i,j)\neq (1,1)$, $\bar{A}_{i,j}=A_{i,j}$. In this case the last $d-1$ dimension is independent of the first dimension, so it is essentially the same as the simplest one-dimensional case.

Denote $X=\{X_{t_{0}},X_{t_{1}},...,X_{t_{N}}\}$ and $g(X)$, $\bar{g}(X)$ be the probability density of $\psi_{\theta}$ and $\psi_{\bar{\theta}}$, respectively. For these two probability densities we have:

\begin{align*}
g(X)=\prod_{i=1}^{N}\frac{1}{\sqrt{2\pi\Gamma(t_{i}-t_{i-1})}}exp\left(-\frac{1}{2\Gamma(t_{i}-t_{i-1})}(X_{t_{i}}-e^{-(t_{i}-t_{i-1})}X_{t_{i-1}})^{2}\right)\,,
\end{align*}
and 
\begin{align*}
\bar{g}(X)=\prod_{i=1}^{N}\frac{1}{\sqrt{2\pi\bar{\Gamma}(t_{i}-t_{i-1})}}exp\left(-\frac{1}{2\bar{\Gamma}(t_{i}-t_{i-1})}(X_{t_{i}}-e^{-(1+\delta)(t_{i}-t_{i-1})}X_{t_{i-1}})^{2}\right)\,.
\end{align*}
Where 
\begin{align*}
\Gamma(t)=\int_{s=0}^{t}e^{-2s}ds=\frac{1}{2}(1-e^{-2t})\quad \bar{\Gamma}(t)=\int_{s=0}^{t}e^{(-2-2\delta)s}ds=\frac{1}{2+2\delta}(1-e^{-(2+2\delta)t})\,.
\end{align*}

Denote $\alpha_{i}=\sqrt{\frac{1}{\Gamma(t_{i}-t_{i-1})}}(X_{t_{i}}-e^{-(t_{i}-t_{i-1})}X_{t_{i-1}})$, \\$\beta_{i}=\sqrt{\frac{1}{\Gamma(t_{i}-t_{i-1})}}(e^{-(t_{i}-t_{i-1})}-e^{-(1+\delta)(t_{i}-t_{i-1})})X_{t_{i-1}}$ and $\gamma_{i}=\sqrt{\frac{\Gamma(t_{i}-t_{i-1})}{\bar{\Gamma}(t_{i}-t_{i-1})}}$. Then

\begin{align*}
\ln\left(\frac{g(X)}{\bar{g}(X)}\right)=\sum_{i=1}^{N}-\ln(\gamma_{i})+\frac{1}{2}\gamma^{2}_{i}(\alpha_{i}+\beta_{i})^{2}-\frac{1}{2}\alpha_{i}^{2}\,.
\end{align*}

Next we show that $\left\vert \ln\left(\frac{g(X)}{\bar{g}(X)}\right)\right\vert$ is not large with high probability when $X$ follows the probability density of $g$. Consider the following subsets of $X$: $\mathcal{E}_{1}=\left\{X\big\vert \left\vert\sum_{i=1}^{N}-\ln(\gamma_{i})+\frac{1}{2}(\gamma^{2}_{i}-1)\alpha_{i}^{2}  \right\vert\leq 1   \right\}$.

$\mathcal{E}_{2}=\left\{ X\big\vert \vert \sum_{i=1}^{N} \gamma_{i}^{2}\alpha_{i}\beta_{i}\vert\leq 1 \right\}$
and 
$\mathcal{E}_{3}=\left\{X\big\vert \frac{1}{2}\sum_{i=1}^{N}\gamma^{2}_{i}\beta^{2}_{i}\leq 1 \right\}$.
When $X$ lies in the intersection of these three sets, $\left\vert \ln\left(\frac{g(X)}{\bar{g}(X)}\right)\right\vert$ is guaranteed to be not very large.

Let $\mathbb{P}$ be the probability with respect to density $g$. We will explicitly show that $\mathbb{P}[X\in\mathcal{E}_{k}]\geq \frac{5}{6}(k=1,2,3)$.

\paragraph{Lower bound $\mathbb{P}[X\in\mathcal{E}_{1}]$}

Firstly, we estimate $\sum_{i=1}^{N}\frac{1}{2}(\gamma^{2}_{i}-1)-\ln(\gamma_{i})$. We first prove the following inequality:

\begin{align}
0\leq \gamma^{2}_{i}-1\leq 2\delta\min\{1,t_{i}-t_{i-1}\}\,.
\end{align}

Let $t=t_{i}-t_{i-1}$. Then $\gamma^{2}_{i}=(1+\delta)\frac{1-e^{-2t}}{1-e^{-(2+2\delta)t}}$.

The left hand side of this inequality is because $\Gamma_{t}\geq \bar{\Gamma}_{t}$, due to the reason that $e^{-2s}\geq e^{-(2+2\delta)s}$ for all $s\geq 0$ and when $f(x)\geq g(x)$ for any $x\in I$ we have: $\int_{x\in I}f(x)dx\geq \int_{x\in I}g(x)dx$. Now we consider the right hand side of the inequality.

\textbf{Case 1:} When $t\geq 1$, we directly use the fact that $1-e^{-2t}\leq 1-e^{-(2+2\delta)t}$ and obtain $\gamma_{i}\leq 1+\delta$.

\textbf{Case 2:} When $t\in(0,1]$, it suffices to show that
\begin{align*}
(1+\delta)(1-e^{-2t})\leq (1+2\delta t)(1-e^{-(2+2\delta)t})\,.
\end{align*}

Let $h(t)=(1+\delta)(1-e^{-2t})-(1+2\delta t)(1-e^{-(2+2\delta)t})$, then
\begin{align*}
h(t)&=\delta(1-2t)-e^{-2t}[1+\delta -(1+2\delta t)e^{-2\delta t}]\\
&\leq \delta(1-2t-e^{-2t})\\
&\leq 0\,.
\end{align*}
Where for the first inequality we use the relation that $e^{-2\delta t}\leq \frac{1}{1+2\delta t}$. The second inequality is obtained by the relation that $e^{-2t}\geq 1-2t$.

Now we bound $\frac{1}{2}(\gamma^{2}_{i}-1)-\ln(\gamma_{i})$. We first show that

\begin{align*}
0\leq \frac{1}{2}(\gamma^{2}_{i}-1)-\ln(\gamma_{i})\leq \frac{1}{4}(\gamma^{2}_{i}-1)^{2}\,.
\end{align*}

Let $x=\gamma^{2}_{i}-1$ and we obtain $ \frac{1}{2}(\gamma^{2}_{i}-1)-\ln(\gamma_{i})=\frac{1}{2}[x-\ln(1+x)]$, and the inequality is obtained directly since we have $x\geq \ln(1+x)\geq x-\frac{1}{2}x^{2}(x\geq 0)$.

Then we can bound $\sum_{i=1}^{N}\frac{1}{2}(\gamma^{2}_{i}-1)-\ln(\gamma_{i})$ as

\begin{align*}
0\leq \sum_{i=1}^{N}\frac{1}{2}(\gamma^{2}_{i}-1)-\ln(\gamma_{i})&\leq \sum_{i=1}^{N}\frac{1}{4}(\gamma^{2}_{i}-1)^{2}\\
&\leq \sum_{i=1}^{N}\delta^{2}\min(1,(t_{i}-t_{i-1}))^{2}\\
&\leq \sum_{i=1}^{N}\delta^{2}(t_{i}-t_{i-1})\\
&\leq \delta^{2}T\\&\leq \frac{1}{25}\,.
\end{align*}

Now we bound $\sum_{i=1}^{N}\frac{1}{2}(\gamma^{2}_{i}-1)(\alpha^{2}_{i}-1)$. Notice that this variable has zero mean, so we can bound its variance and then apply Markov inequality to obtain a high probability bound.

At first, consider the variance of $\alpha_{i}^{2}-1$, denoted as $Var(\alpha_{i}^{2}-1)$. By noticing that $\alpha_{i}\sim\mathcal{N}(0,1)$, we can directly calculate that

\begin{align*}
Var(\alpha_{i}^{2}-1)&=\int_{x\in\mathbb{R}}\frac{1}{\sqrt{2\pi}}e^{-\frac{1}{2}x^{2}}(x^{2}-1)^{2}dx=2\,.
\end{align*}

Since all the $\alpha_{i}$'s are independent, we have:

\begin{align*}
Var\left(\sum_{i=1}^{N}\frac{1}{2}(\gamma^{2}_{i}-1)(\alpha^{2}_{i}-1)\right)
&=\sum_{i=1}^{N}\frac{1}{4}(\gamma^{2}_{i}-1)^{2}Var(\alpha^{2}_{i}-1)\\
&\leq \frac{1}{2}\sum_{i=1}^{N}(\gamma^{2}_{i}-1)^{2}\\
&\leq 2\delta^{2}\sum_{i=1}^{N}\min(1,t_{i}-t_{i-1})^{2}\\
&\leq 2\delta^{2}T\\
&\leq \frac{2}{25}\,.
\end{align*}

By Markov inequality, we have:

\begin{align*}
\mathbb{P}\left[\left\vert\sum_{i=1}^{N}\frac{1}{2}(\gamma^{2}_{i}-1)(\alpha^{2}_{i}-1) \right\vert \geq \frac{4}{5}\right]\leq Var\left(\sum_{i=1}^{N}\frac{1}{2}(\gamma^{2}_{i}-1)(\alpha^{2}_{i}-1)\right)/\left(\frac{4}{5}\right)^{2}\leq \frac{1}{8}\,.
\end{align*}

Finally, for the subset $\mathcal{E}_{1}=\left\{X\big\vert \left\vert\sum_{i=1}^{N}-\ln(\gamma_{i})+\frac{1}{2}(\gamma^{2}_{i}-1)\alpha_{i}^{2}  \right\vert\leq 1   \right\}$, we have:

\begin{align*}
\mathbb{P}\left[x\in\mathcal{E}_{1}\right]\geq 1-\mathbb{P}\left[\left\vert\sum_{i=1}^{N}\frac{1}{2}(\gamma^{2}_{i}-1)(\alpha^{2}_{i}-1) \right\vert \geq \frac{4}{5}\right]\geq \frac{7}{8}\,.
\end{align*}

\paragraph{Lower bound $\mathbb{P}[X\in\mathcal{E}_{2}]$}

Since all the $\alpha_{i}$'s are independent, and $\alpha_{i}$ is independent of $\{\beta_{1},...,\beta_{i}\}$ and $\{\gamma_{1},...,\gamma_{N}\}$, we obtain that 

\begin{align*}
\mathbb{E}\left[\left(\sum_{i=1}^{N} \gamma_{i}^{2}\alpha_{i}\beta_{i}\right)^{2}\right]
&=\mathbb{E}\left[\sum_{i=1}^{N}(\gamma_{i}^{2}\alpha_{i}\beta_{i})^{2}\right]\\
&=\mathbb{E}\left[\sum_{i=1}^{N}(\gamma_{i}^{2}\beta_{i})^{2}\right]\\
&=\sum_{i=1}^{N}\mathbb{E}\left[(\gamma_{i}^{2}\beta_{i})^{2}\right]\,.
\end{align*}

We have shown that $\gamma_{i}^{2}\leq 1+2\delta$. Then for $T\geq 1$ we have: $\gamma^{4}_{i}\leq (1+\frac{2}{5})^{2}\leq 2$. Therefore, we obtain:

\begin{align*}
\mathbb{E}\left[\left(\sum_{i=1}^{N} \gamma_{i}^{2}\alpha_{i}\beta_{i}\right)^{2}\right]\leq 2\sum_{i=1}^{N}\mathbb{E}\left[\beta^{2}_{i}\right]\,.
\end{align*}

Now we upper bound $\mathbb{E}\left[\beta_{i}^{2}\right]$, where $\beta_{i}=\sqrt{\frac{1}{\Gamma(t_{i}-t_{i-1})}}(e^{-(t_{i}-t_{i-1})}-e^{-(1+\delta)(t_{i}-t_{i-1})})X_{t_{i-1}}$

Firstly, we show that
\begin{align}
 \sqrt{\frac{1}{\Gamma(t_{i}-t_{i-1})}}(e^{-(t_{i}-t_{i-1})}-e^{-(1+\delta)(t_{i}-t_{i-1})})\leq \delta\sqrt{t_{i}-t_{i-1}}\,.
\end{align}

Again denote $t=t_{i}-t_{i-1}$. By using $\Gamma_{t}=\frac{1}{2}(1-e^{-2t})$, it suffices to show that

\begin{align*}
e^{-t}-e^{-(1+\delta)t}\leq \delta \sqrt{\frac{1}{2}t(1-e^{-2t})}\,.
\end{align*}

By multiplying $e^{t}$ on both sides, the inequality is equivalent to 
\begin{align*}
1-e^{-\delta t}\leq \delta \sqrt{\frac{1}{2}t(e^{2t}-1)}\,.
\end{align*}

This is true since $e^{-\delta t}\geq 1-\delta t$, and $e^{2t}\geq 1+2t$, implying that 

\begin{align*}
1-e^{-\delta t}\leq\delta t\leq \delta \sqrt{\frac{1}{2}t(e^{2t}-1)}\,.
\end{align*}

With this result, we can upper bound $2\sum_{i=1}^{N}\mathbb{E}\left[\beta^{2}_{i}\right]$ by

\begin{align*}
2\sum_{i=1}^{N}\mathbb{E}\left[\beta^{2}_{i}\right]\leq \sum_{i=1}^{N}2\delta^{2}(t_{i}-t_{i-1})\mathbb{E}\left[X^{2}_{t_{i-1}}\right]\,.
\end{align*}

Finally, since $X_{t}\sim\mathcal{N}(0,\Gamma(t))$, for all $t\geq 0$,
\begin{align*}
\mathbb{E}\left[X^{2}_{t}\right]=\Gamma_{t}=\frac{1}{2}(1-e^{-2t})\leq 1\,.
\end{align*}

Therefore, we obtain
\begin{align*}
\mathbb{E}\left[\left(\sum_{i=1}^{N} \gamma_{i}^{2}\alpha_{i}\beta_{i}\right)^{2}\right]\leq 2\sum_{i=1}^{N}\mathbb{E}\left[\beta^{2}_{i}\right]\leq \sum_{i=1}^{N}2\delta^{2}(t_{i}-t_{i-1})\mathbb{E}\left[X^{2}_{t_{i-1}}\right]\leq \sum_{i=1}^{N}2\delta^{2}(t_{i}-t_{i-1})=2T\delta^{2}=\frac{2}{25}
\end{align*}

Again by using Markov inequality, we obtain:

\begin{align*}
\mathbb{P}\left[\vert \sum_{i=1}^{N} \gamma_{i}^{2}\alpha_{i}\beta_{i}\vert > 1\right]\leq \frac{2}{25}\,.
\end{align*}
Which follows that 
\begin{align*}
\mathbb{P}\left[X\in\mathcal{E}_{2}\right]=1-\mathbb{P}\left[\vert \sum_{i=1}^{N} \gamma_{i}^{2}\alpha_{i}\beta_{i}\vert \geq 1\right]\geq \frac{23}{25}\,.
\end{align*}

\paragraph{Lower bound $\mathbb{P}[X\in\mathcal{E}_{3}]$}

We have shown that $\gamma^{2}_{i}\leq 2,\forall i$ and $\sum_{i=1}^{N}\mathbb{E}\left[\beta^{2}_{i}\right]\leq \delta^{2}T$. Therefore,

\begin{align*}
\mathbb{E}\left[\frac{1}{2}\sum_{i=1}^{N}\gamma^{2}_{i}\beta^{2}_{i}\right]\leq \delta^{2}T\leq \frac{2}{25}\,.
\end{align*}

And we also have 

\begin{align*}
\mathbb{P}\left[X\in\mathcal{E}_{3}\right]=1-\mathbb{P}\left[\frac{1}{2}\sum_{i=1}^{N}\gamma^{2}_{i}\beta^{2}_{i}>1\right]\geq \frac{23}{25}\,.
\end{align*}

Now we come back to prove the theorem. With lower bounds of $\mathbb{P}[X\in\mathcal{E}_{1}],\mathbb{P}[X\in\mathcal{E}_{2}],\mathbb{P}[X\in\mathcal{E}_{3}]$, we have

\begin{align*}
\mathbb{P}\left[X\in \mathcal{E}_{1}\cap \mathcal{E}_{2}\cap \mathcal{E}_{3}\right]\geq 1-(1-\mathbb{P}[X\in\mathcal{E}_{1}])-(1-\mathbb{P}[X\in\mathcal{E}_{2}])-(1-\mathbb{P}[X\in\mathcal{E}_{3}])\geq \frac{1}{2}\,.
\end{align*}

With this bound, we have:

\begin{align*}
&\mathbb{E}_{X\sim g}\left[1\left(\vert \phi(X)-A\vert\geq\frac{1}{10\sqrt{T}}\right)\right]+\mathbb{E}_{X\sim \bar{g}}\left[1\left(\vert \phi(X)-\bar{A}\vert\geq\frac{1}{10\sqrt{T}}\right)\right]\\
&\geq \int_{X\in\mathcal{E}_{1}\cap \mathcal{E}_{2}\cap \mathcal{E}_{3}} g(X)\mathbb{E}\left[1\left(\Vert \phi(X)-A\Vert\geq\frac{1}{10\sqrt{T}}\right)\big\vert X\right]+\bar{g}(X)\mathbb{E}\left[1\left(\Vert \phi(X)-\bar{A}\Vert\geq\frac{1}{10\sqrt{T}}\right)\big\vert X\right] dX\\
&\geq \int_{X\in\mathcal{E}_{1}\cap \mathcal{E}_{2}\cap \mathcal{E}_{3}}\min\{g(X),\bar{g}(X)\}dX\\
&\geq \int_{X\in\mathcal{E}_{1}\cap \mathcal{E}_{2}\cap \mathcal{E}_{3}}\frac{1}{e^{3}}g(X)dX\\
&\geq\frac{1}{2e^{3}}\,.
\end{align*}

Where the second inequality is because $\Vert \phi(X)-A\Vert+\Vert \phi(X)-\bar{A}\Vert\geq \Vert A-\bar{A}\Vert =\frac{1}{5\sqrt{T}}$ so we cannot have both $\Vert \phi(X)-A\Vert\leq \frac{1}{10\sqrt{T}}$ and $\Vert \phi(X)-\bar{A}\Vert\leq\frac{1}{10\sqrt{T}}$. The third inequality is because for any $X\in \mathcal{E}_{1}\cap\mathcal{E}_{2}\cap\mathcal{E}_{3}$, we have

\begin{align*}
\left\vert \ln\frac{g(X)}{\bar{g}(X)}\right\vert&=\left\vert\sum_{i=1}^{N}-\ln(\gamma_{i})+\frac{1}{2}\gamma^{2}_{i}(\alpha_{i}+\beta_{i})^{2}-\frac{1}{2}\alpha_{i}^{2}\right\vert\\
&\leq \left\vert \sum_{i=1}^{N}-\ln(\gamma_{i})+\frac{1}{2}(\gamma^{2}_{i}-1)\alpha_{i}^{2}\right\vert\\
&+\left\vert \sum_{i=1}^{N} \gamma_{i}^{2}\alpha_{i}\beta_{i}\right\vert\\
&+\frac{1}{2}\sum_{i=1}^{N}\gamma^{2}_{i}\beta^{2}_{i}\\
&\leq 3\,,
\end{align*}
implying that $\bar{g}(X)\geq\frac{1}{e^{3}}g(X)$.

Therefore, we have:

\begin{align*}
\max\left\{\mathbb{P}_{X\sim g}\left[\vert \phi(X)-A\vert\geq\frac{1}{10\sqrt{T}}\right],\mathbb{P}_{X\sim \bar{g}}\left[\vert \phi(X)-\bar{A}\vert\geq \frac{1}{10\sqrt{T}}\right]\right\}\geq \frac{1}{4e^{3}}\,.
\end{align*}

This means that for any algorithm, it cannot achieve $\frac{1}{10\sqrt{T}}$ estimation error with success probability $1-\frac{1}{4e^{3}}$ for at least one of the systems controlled by $(A,0)$ and $(\bar{A},0)$.

\end{proof}

\section{Regret Analysis}
\label{app:B}
Having demonstrated the results of system identification for continuous-time linear systems, we leverage these findings to establish upper bounds on the regret for Algorithm~\ref{alg:3}. Elaborations on the details will be presented in the subsequent sections.

\subsection{Convergence of $P$ and the Estimation Error of $K$}
In this section we provide the following Lemma~\ref{lem:convergence of P}, along with its proof, which shows that $\Vert P-P_{*}\Vert$ converges at the same speed as $\Vert \hat{A}-A\Vert+\Vert \hat{B}-B\Vert$.

\begin{lemma}\label{lem:convergence of P}
	There exist constants $\epsilon_{0}>0$ and $C_{2}>0$ such that as long as $\|\hat{A}-A\|,\|\hat{B}-B\|\leq \epsilon$ for some $0<\epsilon<\epsilon_{0}$, with $P$ obtained from \eqref{eq:stationary} we have:
	\begin{align}
	\|P-P_{*}\|\leq C_{2}\epsilon\,.
	\end{align}	
\end{lemma}

Recall that the optimal dynamic is $K_{*}=-R^{-1}B^{\mathrm{T}}P_{*}$ with $P_{*}$ obtained from equation~\eqref{eq:stationary}. Now we consider the distance between it and the sub-optimal dynamic $\bar{K}=-R^{-1}B^{\mathrm{T}}P$ with $P$ obtained from~\eqref{eq:stationary} with $(\hat{A},\hat{B})$. Denote $\Delta A=\hat{A}-A$ and $\Delta B=\hat{B}-B$, along with $\Vert \Delta A\Vert ,\Vert \Delta B\Vert \leq \epsilon$ where $\epsilon\in[0,\epsilon_{0}]$ with some $\epsilon_{0}$ determined later.
We establish the proof by constructing a sequence of matrices $(P_{k})_{k\geq 0}$, and we will prove that such sequence converges to the unique symmetric solution $P$ satisfying 
\begin{align*}
P\hat{B}R^{-1}\hat{B}^{\mathrm{T}}P-\hat{A}^{\mathrm{T}}P-P\hat{A}-Q=0\,. 
\end{align*}

At first we introduce a solution of a particular kind of matrix equation~\citep{kleinman1968iterative}.

\begin{lemma}\label{lem:solution}
	Suppose $A$ satisfies $\alpha(A)=\max\{\Re(\lambda_{i})\vert \lambda_{i}\in \lambda(A)\}<0$. $Q$ is a symmetric matrix. Consider such a function
	
	\begin{align}\label{eq:equation}
	A^{\mathrm{T}}X+XA+Q=0\,.
	\end{align}
	
	Then, the unique symmetric solution $X$ of this equation can be expressed as:
	
	\begin{align}\label{eq:solution}
	X=\int_{t\geq0}e^{A^{\mathrm{T}}t}Qe^{At}dt\,.
	\end{align}
\end{lemma}

Now we consider the relation between $P$ and $P_{*}$. The core is iteratively constructing a sequence of matrices $P_{k}$ such that $P_{0}=P_{*}$ and $\lim_{k\to+\infty}P_{k}=P$. Such matrices follows the relation $P_{k+1}=P_{k}+\Delta P_{k}$ where $\Delta P_{k}$ converges rapidly. As for the starting case, consider the expansion
\begin{align*}
&(P_{*}+\Delta P)(B+\Delta B)R^{-1}(B+\Delta B)^{\mathrm{T}}(P_{*}+\Delta P)\\
&-(A+\Delta A)^{\mathrm{T}}(P_{*}+\Delta P)-(P_{*}+\Delta P)(A+\Delta A)-Q\\
&=\left[(B+\Delta B)R^{-1}(B+\Delta B)^{\mathrm{T}}P_{*}-A-\Delta A\right]^{\mathrm{T}}\Delta P\\
&+\Delta P\left[(B+\Delta B)R^{-1}(B+\Delta B)^{\mathrm{T}}P_{*}-A-\Delta A\right]\\
&+\left[P_{*}BR^{-1}B^{\mathrm{T}}P_{*}-A^{\mathrm{T}}P_{*}-P_{*}A-Q\right]+P_{*}\left[\Delta B\left(R^{-1}(B+\Delta B)^{\mathrm{T}}\right)+BR^{-1}\Delta B\right]P_{*}\\
&+\Delta P(B+\Delta B)R^{-1}(B+\Delta B)^{\mathrm{T}}\Delta P\,.
\end{align*}

Define
\begin{align*}
&A_{0}=A+\Delta A-(B+\Delta B)R^{-1}(B+\Delta B)^{\mathrm{T}}P_{*}\,,\\& F_{0}=-P_{*}\left[\Delta B\left(R^{-1}(B+\Delta B)^{\mathrm{T}}\right)+BR^{-1}\Delta B\right]P_{*}\,. 
\end{align*}
We set $\Delta P_{0}$ be a solution of

\begin{align*}
A_{0}^{\mathrm{T}}\Delta P_{0}+\Delta P_{0} A_{0}+F_{0}=0\,.
\end{align*}
which satisfies that (see Lemma~\ref{lem:solution})
\begin{align*}
&\Delta P_{0}=\int_{t\geq 0}e^{A_{0}^{\mathrm{T}}t}F_{0}e^{A_{0}t}dt\,,\\ & \Vert \Delta P_{0}\Vert \leq \int_{t\geq 0}e^{2\alpha(A_{0})t}\Vert F_{0}\Vert dt=\frac{1}{-2\alpha(A_{0})}\Vert F_{0}\Vert \leq\frac{1}{-\alpha(A_{0})}\Vert P_{*}\Vert ^{2}(\Vert BR^{-1}\Vert \epsilon+\Vert R^{-1}\Vert \epsilon^{2})\,.
\end{align*}
This $\Delta P_{0}$ also satisfies
\begin{align*}
&(P_{*}+\Delta P_{0})(B+\Delta B)R^{-1}(B+\Delta B)^{\mathrm{T}}(P_{*}+\Delta P_{0})\\&-(A+\Delta A)^{\mathrm{T}}(P_{*}+\Delta P)-(P_{*}+\Delta P)(A+\Delta A)-Q\\&=\Delta P_{0}(B+\Delta B)R^{-1}(B+\Delta B)^{\mathrm{T}}\Delta P_{0}\,.
\end{align*}

An important thing is to guarantee that $A_{0}$ is stable, and $\vert\alpha(A_{0})\vert$ can not be too closed to zero.  
For any $\epsilon_{1}\in(0,1)$ and $C_{1}= \Vert R^{-1}\Vert \Vert P_{*}\Vert +1+2\Vert BR^{-1}\Vert \Vert P_{*}\Vert $, as long as $\epsilon\leq \epsilon_{1}$, $\Vert A_{0}-(A-BR^{-1}B^{\mathrm{T}}P_{*})\Vert \leq C_{1}\epsilon$. Furthermore, there exists $\epsilon_{2}>0$ such that if $\Vert X-(A-R^{-1}B^{\mathrm{T}}P_{*})\Vert \leq \epsilon_{2}$, then $\alpha(X)\leq\frac{1}{2}\alpha(A-R^{-1}B^{\mathrm{T}}P_{*})$(the work of \citep{ref1} shows this result). We can further let this $\epsilon_{2}$ satisfies that, as long as $\Vert \Delta A\Vert,\Vert \Delta B\Vert,\Vert \Delta P\Vert \leq \epsilon_{2}$, we always have:
\begin{align}\label{stable margin}
\alpha\left(A+\Delta(A)-(B+\Delta B)R^{-1}(B+\Delta B)^{\mathrm{T}}(P_{*}+\Delta P)\right)\leq \frac{1}{2}\alpha(A-BR^{-1}B^{\mathrm{T}}P_{*})\,.
\end{align}
Now we additionally set $\epsilon_{1}$ satisfying $\epsilon_{1}\leq \frac{1}{2C_{1}}\epsilon_{2}$ and $\Vert R^{-1}\Vert \epsilon_{1}\leq 1$, then for all $\epsilon\leq \epsilon_{1}$,
\begin{align*}
\left\Vert \Delta P_{0}\right\Vert \leq \frac{2}{-\alpha(A-BR^{-1}B^{\mathrm{T}}P_{*})}\Vert P_{*}\Vert ^{2}(1+\Vert BR^{-1}\Vert )\epsilon\,.
\end{align*}

Denote $P_{1}=P_{0}+\Delta P_{0}$,  $C_{2}=\frac{2}{-\alpha(A-BR^{-1}B^{\mathrm{T}}P_{*})}\Vert P_{*}\Vert ^{2}(1+\Vert BR^{-1}\Vert )$, and set some constant $C_{3}$ satisfying
$C_{3}\geq \Vert BR^{-1}B^{\mathrm{T}}\Vert +2\Vert BR^{-1}\Vert +\Vert R^{-1}\Vert $. We then inductively define $P_{k+1}$ and $\Delta P_{k}$ $(k\geq 1)$. For defined $\Delta P_{k-1}$, we set $P_{k}=P_{k-1}+\Delta P_{k-1}$, which satisfies
\begin{align*}
& P_{k}(B+\Delta B)R^{-1}(B+\Delta B)^{\mathrm{T}}P_{k}-(A+\Delta A)^{\mathrm{T}}P_{k}-P_{k}(A+\Delta(A))-Q\\&=\Delta P_{k-1}(B+\Delta B)R^{-1}(B+\Delta B)^{\mathrm{T}}\Delta P_{k-1}\,.
\end{align*} 
Then we denote $A_{k}=A+\Delta A-(B+\Delta B)R^{-1}(B+\Delta B)^{\mathrm{T}}P_{k}$, and set $\Delta P_{k}$ satisfying:
\begin{align*}
A^{\mathrm{T}}_{k}\Delta P_{k}+\Delta P_{k}A_{k} = \Delta P_{k-1}(B+\Delta B)R^{-1}(B+\Delta B)^{\mathrm{T}}\Delta P_{k-1}\,.
\end{align*}

By the hypothesis of $\epsilon_{2}$, as long as $\Vert P_{k}-P_{*}\Vert\leq \epsilon_{2}$, we have $\alpha(A_{k})\geq\frac{1}{2}\alpha(A-BR^{-1}B^{\mathrm{T}}P_{*})$. By using ~\eqref{eq:solution} we obtain that $\Vert \Delta P_{k}\Vert \leq C_{4}\Vert \Delta P_{k-1}\Vert ^{2}$, where $C_{4}=\frac{2}{-\alpha(A-BR^{-1}B^{\mathrm{T}}P_{*})}C_{3}$. Now if we define $P_{k+1}=P_{k}+\Delta P_{k}$, $P_{k+1}$ also satisfies:
\begin{align*}
& P_{k+1}(B+\Delta B)R^{-1}(B+\Delta B)^{\mathrm{T}}P_{k+1}-(A+\Delta A)^{\mathrm{T}}P_{k+1}-P_{k+1}(A+\Delta(A))-Q\\&=\Delta P_{k}(B+\Delta B)R^{-1}(B+\Delta B)^{\mathrm{T}}\Delta P_{k}\,,
\end{align*} 
Then these sequences $\Delta P_{k}$ and $P_{k}$ are well defined, along with the relation that $P_{k+1}=P_{k}+\Delta P_{k}$. Furthermore, when $\Vert P_{k}-P_{*}\Vert\leq \epsilon_{2}$, we have$\Vert\Delta P_{k+1}\Vert\leq C_{4}\Vert\Delta P_{k}\Vert^{2}$. Note that for the base case we have $\Vert \Delta P_{0}\Vert \leq C_{2}\epsilon$.

Finally, it remains to constrain $\Vert P_{k}-P_{*}\Vert$. By choosing $\epsilon\leq \min( \frac{1}{2C_{2}C_{4}}, \frac{1}{2C_{2}}\epsilon_{2},1)$, we obtain $\Vert \Delta P_{0}\Vert \leq C_{2}\epsilon$. We can also see that if for all $0\leq k\leq m$, $\Vert \Delta P_{k} \Vert \leq2^{-k}C_{2}\epsilon$, then $\Vert P_{m}-P_{*}\Vert \leq 2(1-2^{-m+1})C_{2}\epsilon\leq \epsilon_{2}$ so that $\Vert \Delta P_{m+1}\Vert \leq C_{4}\Vert \Delta P_{m}\Vert^{2}\leq 2^{-m-1}C_{2}\epsilon$. So by induction we see that $\Vert \Delta P_{k}\Vert \leq 2^{-k}C_{2}\epsilon$ for any $k$.

On the other hand, since $\Vert \Delta P_{k}\Vert \leq 2^{-k}\Vert \Delta P_{0}\Vert $, $\lim_{k\to+\infty}P_{k}=P_{\infty}$ exists, and such $P_{\infty}$ is the unique symmetric solution of 
\begin{align*}
P(B+\Delta B)R^{-1}(B+\Delta B)^{\mathrm{T}}P-(A+\Delta A)^{\mathrm{T}}P-P(A+\Delta(A))-Q=0\,,
\end{align*}
such that $(A+\Delta A)-(B+\Delta B)R^{-1}(B+\Delta B)^{\mathrm{T}}P$ is stable (recall the stable margin in~\eqref{stable margin}, which implies that $(A+\Delta A)-(B+\Delta B)R^{-1}(B+\Delta B)^{\mathrm{T}}P_{\infty}$ is stable). 

So $P_{\infty}$ is exactly $P$, satisfying $\Vert P-P_{*}\Vert \leq 2C_{2}\epsilon$.

Therefore, we conclude that there exists some $\epsilon_{0}>0$ and constant $C$, both depending on $A,B,K,d,p$ such that for any $\epsilon\in[0,\epsilon_{0}]$, $\Vert P-P_{*}\Vert \leq C\epsilon$ as long as $\Vert \hat{A}-A\Vert ,\Vert \hat{B}-B\Vert \leq \epsilon$.

Then we apply our results for system identification to establish an upper bound for $\Vert \bar{K}-K_{*}\Vert$.

Based on Lemma~\ref{lem:convergence of P}, fix constant $\epsilon_{1}>0$ and constant $C_{1}\geq 0$ so that we have $\Vert P-P_{*}\Vert\leq C_{1}\left(\Vert\hat{A}-A\Vert+\Vert\hat{B}-B\Vert\right)$ whenever $\Vert\hat{A}-A\Vert+\Vert\hat{B}-B\Vert\leq \epsilon_{1}$

We set $C_{2}\geq 1$ be two times the constant $C$ in Lemma~\ref{lem:A,B}, and obtain that, when $\log^{2}(1/\delta)\leq \frac{T^{1/2}}{C_{2}}$ and $T^{1/2}\geq C_{2}\Vert X_{0}\Vert^{2}_{2}$, we have:
\begin{align*}
    \mathbb{P}\left[\Vert\hat{A}-A\Vert+\Vert \hat{B}-B\Vert \leq 2C_{2}\sqrt{\frac{\log(1/\delta)}{T^{1/2}}}\right]\geq 1-\delta\,.
\end{align*}
Then, for $\log(1/\delta)\leq \min\left\{\frac{T\epsilon_{1}^{2}}{4C^{2}_{2}},\frac{T^{1/4}}{C^{1/2}_{2}}\right\}\leq \frac{T^{1/4}\epsilon^{2}_{1}}{4C_{2}^{2}}$, we have:
\begin{align}\label{eq: upper bound P}
\mathbb{P}\left[\Vert P-P_{*}\Vert \leq 2C_{1}C_{2}\sqrt{\frac{\log(1/\delta)}{T^{1/2}}}\right]\geq 1-\delta\,.    
\end{align}

Finally, since $\bar{K}=-R^{-1}(\hat{B})^{\mathrm{T}}P$, $K_{*}=-R^{-1}B^{\mathrm{T}}P_{*}$, we have:
\begin{align*}
\left\Vert \bar{K}-K_{*}\right\Vert&\leq \Vert R^{-1}\Vert \left[\Vert \hat{B}-B\Vert \Vert P\Vert +\Vert B\Vert \Vert P-P_{*}\Vert   \right]\,.  
\end{align*}

 We can reset $C_{1}$ such that
$\Vert \bar{K}-K_{*}\Vert\leq C_{1}\left(\Vert\hat{A}-A\Vert+\Vert\hat{B}-B\Vert\right)$ whenever $\Vert\hat{A}-A\Vert+\Vert\hat{B}-B\Vert\leq \epsilon_{1}$, and combine this with ~\eqref{eq: upper bound P}, we have: for any $\log(1/\delta)\leq \frac{T^{1/4}\epsilon^{2}_{1}}{4C_{2}^{2}}$
\begin{align}\label{eq: upper bound K}
\mathbb{P}\left[\Vert \bar{K}-K_{*}\Vert \leq 2C_{1}C_{2}\sqrt{\frac{\log(1/\delta)}{T^{1/2}}}\right]\geq 1-\delta\,.    
\end{align}

With this probability bound on $\Vert \bar{K}-K_{*}\Vert$, we can further upper bound the regret, shown in the following part.
\subsection{Key Lemmas}

We first upper bound the radius of a single trajectory with stable controller, for which we introduce and provide a proof for the following lemma:

\begin{restatable}{lemma}{continuous}
\label{lem:bound continuous trajectory}
Consider the continuous system $dX_{t}=AX_{t}dt+dW_{t}$ such that $\alpha(A)<0$ where $\alpha(A)$ is the largest real component of $A$ and $W$ is a standard Brownian noise. Then, w.p. at least $1-\delta$:
	\begin{align*}
	\sup_{0\leq t\leq T}\left(\|X_{t}\|_{2}-e^{\alpha(A)t}\|X_{0}\|_{2}\right)\leq C\sqrt{d\log((1+T)/\delta)}\,.
	\end{align*}  
\end{restatable}

Then we concentrate on how the error $\Vert P-P_{*}\Vert$ will influence the regret during the exploitation phase.
For a dynamic $U$ with $\alpha(A+BU)<0$, we define a cost function:
\begin{align*}\label{eq:cost function}
cost(U)=\mathrm{tr}\left(\int_{t\geq 0}(e^{(A+BU)t})^{\mathrm{T}}(Q+U^{\mathrm{T}}RU)e^{(A+BU)t}dt\right)\,.   
\end{align*}
The convergence rate of this cost function is stated in the following lemma:

\begin{restatable}{lemma}{convergence}
\label{lem:convergence of cost}
Let $U_{*}$ minimize $cost(U)$. Then, there exists $\epsilon_{0}\geq 0$ such that for any $\Vert \Delta U\Vert=1$ and $\epsilon\in[0,\epsilon_0]$, we have:
\begin{align*}
    cost(U_{*}+\epsilon\Delta U)-cost(U_{*})\leq C_{1}\epsilon^{2}\,.
\end{align*}
\end{restatable}

The above result shows the average cost per unit time when applying fixed controller for infinite time.

Then we further consider the case when the running time is finite. We derive the following lemma:

\begin{restatable}{lemma}{approximated}
\label{lem:approximated optimal cost}
Let $U_{*}$ follows the same definition as in Lemma~\ref{lem:convergence of cost}. Then, for some $\epsilon>0$, there exist constants $C_{2}$ and $C_{3}$ (independent of $U$) such that for all $T>0$ and any $U$ such that $\|U-U_{*}\|\leq \epsilon$,
\begin{align*}
\left\vert J_{T}-cost(U)T\right\vert\leq C_{2}\|x\|^{2}_{2}+C_{3}\,.
\end{align*}
Here $J_{T}$ is the expected cost of the policy that takes action by $U_{t}=UX_{t}$ $(t\in[0,T])$, with initial state $X_{0}=x$.
\end{restatable}

With this lemma, by definition of $U_{*}$, we actually have $U_{*}=K_{*}$, where $K_{*}=-R^{-1}B^{\mathrm{T}}P_{*}$ and $P_{*}$ is the solution of ~\eqref{eq:ricatti}.
Since such $C_{2},C_{3}$ also satisfy:
\begin{align*}
\left\vert J^{*}_{T}-cost(U_{*})T\right\vert \leq C_{2}\|x\|^{2}_{2}+C_{3}\,,
\end{align*}
so it follows that
\begin{align}
R_{T}=J_{T}-J^{*}_{T}\leq 2C_{2}\|x\|^{2}_{2}+2C_{3}\,.
\end{align}
\subsection{Proof of Lemma~\ref{lem:bound continuous trajectory}}\label{subsection: proof of state}
We first upper bound the radius of a single trajectory with stable controller, for which we introduce and provide a proof for the following lemma:
\continuous*
\begin{proof}
The trajectory $X_{t}$ with differential equation $dX_{t}=AX_{t}+dW_{t}$ can be derived as 
\begin{align*}
X_{t}=e^{At}X_{0}+\int_{s=0}^{t}e^{A(t-s)}dW_{t}\,.
\end{align*}
Lemma~\ref{lem:matrix} tells that when $A$ is stable, $\left\Vert e^{At}X_{0}\right\Vert_{2}\leq e^{\alpha(A)t}\Vert X_{0}\Vert _{2}$. So it suffices to show that 
\begin{align*}
\mathbb{P}\left[\sup_{0\leq t\leq T}\left\Vert \int_{s=0}^{t}e^{A(t-s)}dW_{t}\right\Vert _{2}\geq C\sqrt{d\log(1+T)/\delta}\right]\leq \delta\,.
\end{align*}
Let $T=T_{0}h$ with $T_{0}$ be an integer. We first consider the set of points $\{X_{kh}\}$. Denote $w_{k}:=\int_{t=0}^{kh}e^{A(kh-t)}dW_{t}$, then $w_{k}\sim\mathcal{N}(0,\Sigma_{k})$ with $\Sigma_{k}=\int_{t=0}^{kh}e^{At}e^{A^{\mathrm{T}}t}dt$. This $\Sigma_{h}$ also satisfies
\begin{align*}
\left\Vert \Sigma_{k}\right\Vert \leq\int_{t=0}^{h}\left\Vert e^{At}\right\Vert ^{2}dt\leq \int_{t=0}^{kh}e^{2\alpha(A)t}dt\leq \frac{1}{2\vert \alpha(A)\vert}\,.
\end{align*}
Which follows that $\sup_{0\leq k\leq T_{0}}\Vert w_{k}\Vert _{2}\leq 2\sqrt{\frac{d}{\vert \alpha(A)\vert} \log((1+T_{0})/\delta)}$, w.p. at least $1-\delta$.

Next we consider any $X_{kh+t}$ with $t\in[0,h]$. Bounding such terms requires the Doob's martingale inequality~\citep{durrettprobability}, stated as in Lemma~\ref{lem:doob}. We denote $x^{k}_{t}=\int_{s=0}^{t}e^{A(t-s)}dW_{kh+s}ds$ with corresponding filtration $\mathcal{F}_{t}$. We also 
define $Z^{k}_{t}:=e^{\lambda\left\Vert e^{-At}x^{k}_{t}\right\Vert ^{2}_{2}}$ with $\lambda\geq 0$. Then $Z^{k}_{t}$ is a submartingale under the filtration $\mathcal{F}_{t}$, since for any $t\geq s$,
\begin{align*}
\mathbb{E}\left[Z^{k}_{t}\vert \mathcal{F}_{s}\right]=\mathbb{E}\left[\exp\left(\lambda\left\Vert e^{-As}x^{k}_{s}+\int_{t_{1}=s}^{t}e^{-At_{1}}dW_{kh+t_{1}}\right\Vert ^{2}_{2}\right)\big\vert x^{k}_{s}\right]
\geq e^{\lambda\left\Vert e^{-As}x^{k}_{s}\right\Vert_{2} ^{2}}=Z^{k}_{s}\,.
\end{align*}
Where we notice that $\mathbb{E}\left[\left\Vert e^{-As}x^{k}_{s}+\int_{t_{1}=s}^{t}e^{-At_{1}}dW_{kh+t_{1}}\right\Vert ^{2}_{2}\big\vert x^{k}_{s} \right]\geq\left\Vert e^{-As}x^{k}_{s}\right\Vert_{2} ^{2}$, and apply Jensen's inequality on the non-decreasing convex function $f(x)=e^{\lambda x}$ to obtain the above inequality.

Now we apply Lemma~\ref{lem:doob} and get
\begin{align}\label{eq:doob}
\mathbb{P}\left[\sup_{t\in[0,h]}\left\Vert e^{-At}x^{k}_{t}\right\Vert _{2}\geq C\right]\leq e^{-\lambda C^{2}}\mathbb{E}[Z^{k}_{h}]\,.
\end{align}
We next estimate $\mathbb{E}(Z^{k}_{h})$. Since $e^{-Ah}x^{k}_{h}=\int_{t=0}^{h}e^{-At}dW_{kh+t}$, we obtain that $e^{-Ah}x_{h}^{k}\sim\mathcal{N}(0,\bar{\Sigma})$, where
\begin{align*}
\bar{\Sigma}=\int_{t=0}^{h}e^{-At}e^{-A^{\mathrm{T}}t}dt\,.
\end{align*}
By setting $\lambda=\frac{1}{4\Vert \bar{\Sigma}\Vert }$, it can be computed that
\begin{align*}
\mathbb{E}\left[e^{\lambda\left\Vert e^{-Ah}x^{k}_{h}\right\Vert ^{2}_{2}}\right]
&=\int_{x\in\mathbb{R}^{d}}\frac{1}{(2\pi)^{d/2}\sqrt{\det(\bar{\Sigma})}}e^{-\frac{1}{2}x^{\mathrm{T}}\Sigma^{-1}_{1}x}e^{\lambda x^{\mathrm{T}}I_{d}x}dx\\
&=\sqrt{\frac{1}{\det(\bar{\Sigma})\det(\Sigma^{-1}_{1}-2\lambda I_{d})}}\\
&=\sqrt{\frac{1}{\det(I_{d}-2\lambda\bar{\Sigma})}}\\
&\leq 2^{d/2}\,,
\end{align*}
where the last inequality is because $I_{d}-2\lambda\bar{\Sigma}\succeq \frac{1}{2}I_{d}$.

We combine this result with \eqref{eq:doob} and obtain:
\begin{align*}
&\mathbb{P}\left[\sup_{0\leq k\leq T_{0}-1, 0\leq t\leq h}\left\Vert x^{k}_{t}\right\Vert _{2}\geq 2e^{\Vert A\Vert h}\left\Vert \bar{\Sigma}\right\Vert ^{1/2}\sqrt{\log(2^{d/2}T_{0}/\delta)}\right]\\
&\leq \sum_{k=0}^{T_{0}-1}\mathbb{P}\left[\sup_{t\in[0,h]}Z^{k}_{t}\geq 2^{d/2}\frac{T_{0}}{\delta}\right]\\
&\leq \sum_{k=0}^{T_{0}-1}\mathbb{P}\left[\sup_{t\in[0,h]}Z^{k}_{t}\geq \frac{T_{0}}{\delta}\mathbb{E}(Z^{k}_{h})\right]\\
&\leq \delta\,.
\end{align*}
Finally, since $X_{kh+t}=e^{A(kh+t)}X_{0}+e^{At}w_{k}+x^{k}_{t}$, it follows that
\begin{align*}
\left\Vert X_{kh+t}\right\Vert _{2}&\leq \left\Vert e^{A(kh+t)}X_{0}\right\Vert _{2}+\left\Vert e^{At}w_{k}\right\Vert _{2}+\left\Vert x^{k}_{t}\right\Vert _{2}\\
&\leq e^{\alpha(A)(kh+t)}\left\Vert X_{0}\right\Vert _{2}+\left\Vert w_{k}\right\Vert _{2}+\left\Vert x^{k}_{t}\right\Vert _{2}\,.
\end{align*}
By applying union bound on $\left\Vert w_{k}\right\Vert _{2}$ and $\left\Vert x^{k}_{t}\right\Vert _{2}$ we finally obtain Lemma~\ref{lem:bound continuous trajectory}.
\end{proof}

\begin{lemma}[Doob's martingale inequality]\label{lem:doob}
	Let $X_1, \ldots, X_n$ be a discrete-time submartingale relative to a filtration $\mathcal{F}_1, \ldots, \mathcal{F}_n$ of the underlying probability space, which is to say:
	$$
	X_i \leq \mathbb{E}\left[X_{i+1} \mid \mathcal{F}_i\right] .
	$$
	The submartingale inequality says that
	$$
	\mathbb{P}\left[\max _{1 \leq i \leq n} X_i \geq C\right] \leq \frac{\mathbb{E}\left[\max \left(X_n, 0\right)\right]}{C}
	$$
	for any positive number $C$. 
	
	Moreover, let $X_t$ be a submartingale indexed by an interval $[0, \mathrm{~T}]$ of real numbers, relative to a filtration $F_t$ of the underlying probability space, which is to say:
	$$
	X_s \leq \mathrm{E}\left[X_t \mid \mathcal{F}_s\right]
	$$
	for all $s<t$. The submartingale inequality says that if the sample paths of the martingale are almost-surely right-continuous, then
	$$
	\mathbb{P}\left[\sup _{0 \leq t \leq T} X_t \geq C\right] \leq \frac{\mathbb{E}\left[\max \left(X_T, 0\right)\right]}{C}
	$$
	for any positive number $C$. 
\end{lemma}

\subsection{Proof of Lemma~\ref{lem:convergence of cost}}
In this section, we proof Lemma~\ref{lem:convergence of cost} which refers to the convergence rate of the cost function:

\convergence*
\begin{proof}
For any $\Vert \Delta U\Vert =1$ and $\epsilon>0$, consider $U=U_{*}+\epsilon\Delta U$, we show that as $\epsilon\to 0$, there exists $V\in\mathbb{R}^{d}$ such that $\mathrm{tr}(V)=0$, and 
\begin{align*}
&\int_{t\geq 0}e^{(A+BU)^{\mathrm{T}}t}(Q+U^{\mathrm{T}}RU)e^{(A+BU)t}dt-\int_{t\geq 0}e^{(A+BU_{*})^{\mathrm{T}}t}(Q+U_{*}^{\mathrm{T}}RU_{*})e^{(A+BU_{*})t}dt\\
&=\epsilon V+\mathcal{O}(\epsilon^{2})\,.
\end{align*}

Let $D(\epsilon,t)=e^{(A+B(U_{*}+\epsilon \Delta U))t}-e^{(A+BU_{*})t}$. The most important intuition is that $D(\epsilon,t)$ can be represented by the form of $D(\epsilon,t)=\epsilon D_{1}(t)+\epsilon^{2}D_{2}(\epsilon,t)$, where $D_{1}(t)$ does not depend on $\epsilon$, and the residual $D_{2}(\epsilon,t)$ can be well bounded. Now we find such $D_{1}(t)$ and upper bound $\left\Vert D_{2}(\epsilon,t)\right\Vert$. For $t\leq t_{0}=\frac{1}{\max\{\Vert A+BU_{*}\Vert ,\Vert B\Vert \}}$ and $\epsilon<1$, the Taylor expansion of $e^{(A+B(U_{*}+\epsilon \Delta U))t}$ can be represented as follows:
\begin{align*}
D(\epsilon,t)&=\sum_{k\geq 1}\frac{1}{k!}\left[(A+BU_{*}+\epsilon B\Delta U)^{k}t^{k}-(A+BU_{*})^{k}t^{k}\right]\\
&=\sum_{k\geq 1}\frac{1}{k!}\left[\left(\sum_{i=0}^{k-1}(A+BU_{*})^{i}(B\Delta U_{*})(A+BU_{*})^{k-1-i}\right)\epsilon+D_{1}(\epsilon,k)\epsilon^{2}\right]t^{k}\,,
\end{align*}

where $D_{1}(\epsilon,k)$ is the residual of $(A+BU+\epsilon B\Delta U)^{k}-(A+BU)^{k}$ with order at least $\epsilon^{2}$. This sequence of matrices are expressed and bounded as follows.
\begin{align*}
&D_{1}(k,\epsilon)=\sum_{i=2}^{k}\epsilon^{i}\sum_{ j_{1}+...+j_{i+1}=k-i}(A+BU_{*})^{j_{1}}(B\Delta U)(A+BU_{*})^{j_{2}}(B\Delta U)...(A+BU_{*})^{j_{i+1}}\,,\\
&\Vert D_{1}(k,\epsilon)\Vert \leq \sum_{i=2}^{k} \frac{k!}{i!(k-i)!}\Vert A+BU_{*}\Vert ^{k-i}\Vert B\Vert ^{i}\epsilon^{i-2}\,.
\end{align*}

Thus we have:
\begin{align*}
\left\Vert \sum_{k\geq 1}\frac{t^{k}}{k!}D_{1}(k,\epsilon)\right\Vert \leq \sum_{k\geq 2}\sum_{i\geq 2}\frac{1}{i!(k-i)!}\leq 4\,.
\end{align*}

Define $E(t)$ and $E_{1}(\epsilon,t)$ as follows: for $0\leq t\leq t_{0}$, let
\begin{align*}
E(t)=\sum_{k\geq 1}\frac{t^{k}}{k!}\sum_{i=0}^{k-1}(A+BU_{*})^{i}(B\Delta U_{*})(A+BU_{*})^{k-1-i}, E_{1}(\epsilon,t)=\sum_{k\geq 1}\frac{t^{k}}{k!}D_{1}(k,\epsilon)\,,
\end{align*}
and for $t\in[\frac{1}{2}t_{0},t_{0}]$, $l\geq1$, we inductively define $E(2^{l}t)$ and $E_{1}(2^{l}t)$ as follows:
\begin{align*}
E(2^{l}t)=e^{(A+BU_{*})2^{l-1}t}E(2^{l-1}t)+E(2^{l-1}t)e^{(A+BU_{*})2^{l-1}t}\,,
\end{align*}
\begin{align*}
E_{1}(\epsilon,2^{l}t)=&e^{(A+BU_{*})2^{l-1}t}E_{1}(\epsilon, 2^{l-1}t)+E_{1}(\epsilon, 2^{l-1}t)e^{(A+BU_{*})2^{l-1}t}\\
&+\left(E(2^{l-1}t)+\epsilon E_{1}(\epsilon,2^{l-1}t)\right)^{2}\,.
\end{align*}
Then we have the relation that $e^{(A+BU_{*}+B\Delta U)t}-e^{(A+BU_{*})t}=\epsilon E(t)+\epsilon^{2}E_{1}(\epsilon,t)$.

Now we upper bound $\Vert E(t)\Vert $ and $\Vert E_{1}(\epsilon,t)\Vert $. When $t\leq t_{0}$:

\begin{align*}
\Vert E(t)\Vert \leq\sum_{k\geq 1}\frac{t^{k}}{k!}\sum_{i=0}^{k-1}\left\Vert (A+BU_{*})^{i}(B\Delta U_{*})(A+BU_{*})^{k-1-i}\right\Vert \leq \sum_{k\geq 1}\frac{1}{(k-1)!}=e
\end{align*}
For $t\geq t_{0}$, let $t=2^{l_{1}}t_{1}$, with $l_{1}$ be an integer and $t_{1}\in(\frac{1}{2}t_{0},t_{0}]$, then
\begin{align*}
\Vert E(2^{l_{1}}t_{1})\Vert &=\left\Vert e^{(A+BU_{*})2^{l_1-1}t_{1}}E(t)+E(t)e^{(A+BU_{*})2^{l_1-1}t_{1}}\right\Vert \\
&\leq 2e^{\alpha(A+BU_{*})2^{l_1-1}t_{1}}\left\Vert E(2^{l_1-1}t_{1})\right\Vert \\
&\leq 2^{l_1}e^{1+\alpha(A+BU_{*})2^{l_1-2}t_{0}}\\
&\leq \frac{4}{-\alpha(A+BU_{*})t_{0}}\,,
\end{align*}
where the last inequality is because for any $x,a>0$, $xe^{-ax}\leq \frac{1}{ae}$, and thus for any $t\geq 0$, $\Vert E(t)\Vert \leq C=\frac{4}{-\alpha(A+BU_{*})t_{0}}$.

When $t\geq \frac{2}{-\alpha(A+BU_{*})}$, we additionally have
\begin{align*}
\Vert E(t)\Vert \leq 2e^{\frac{1}{2}\alpha(A+BU_{*})t}\left\Vert E(\frac{t}{2})\right\Vert\leq \frac{4t}{t_{0}}e^{\frac{1}{2}\alpha(A+BU_{*})t}\leq\frac{8}{-\alpha(A+BU_{*})t_{0}}e^{\frac{1}{4}\alpha(A+BU_{*})t}\,.
\end{align*}

Now we consider $E_{1}(\epsilon,t)$. When $t\leq t_{0}$,
\begin{align*}
\Vert E_{1}(\epsilon,t)\Vert \leq \sum_{k\geq 1}\left\Vert \frac{t^{k}}{k!}D_{1}(k,\epsilon)\right\Vert \leq 4\,.
\end{align*}

When $t>t_{0}$, with $t=2^{l}t_{1}$ and $t_{1}\in(\frac{1}{2}t_{0},t_{0}]$, we obtain:
\begin{align*}
&\left\Vert E_{1}(\epsilon,2^{l}t_{1})\right\Vert= \\
&\left\Vert e^{(A+BU_{*})2^{l-1}t_{1}}E_{1}(\epsilon, 2^{l-1}t_{1})+E_{1}(\epsilon, 2^{l-1}t_{1})e^{(A+BU_{*})2^{l-1}t_{1}}+\left(E(2^{l-1}t_{1})+\epsilon E_{1}(\epsilon,2^{l-1}t_{1})\right)^{2}\right\Vert \\
&\leq 2e^{\alpha(A+BU_{*})2^{l-1}t_{1}}\left\Vert E_{1}(\epsilon,2^{l-1}t_{1})\right\Vert +\left\Vert E(2^{l-1}t_{1})+\epsilon E_{1}(\epsilon,2^{l-1}t_{1})\right\Vert ^{2}\\
&\leq 2e^{\alpha(A+BU_{*})2^{l-1}t_{1}}\left\Vert E_{1}(\epsilon,2^{l-1}t_{1})\right\Vert +2\left\Vert E(2^{l-1}t_{1})\right\Vert ^{2}+2\epsilon^{2}\left\Vert E_{1}(\epsilon,2^{l-1}t_{1})\right\Vert ^{2}\,.
\end{align*}

Now, we show that $\left\Vert E_{1}(\epsilon,2^{l}t_{1})\right\Vert$ converges exponentially eventually. The proof consists of two parts: first, for $t$ which is not too large, $\left\Vert E_{1}(\epsilon,t)\right\Vert$ can be bounded uniformly over all possible $\Delta U$ and any constrained $\epsilon$. Then, for larger $t$ we can utilize the construction of $\left\Vert E_{1}(\epsilon,t)\right\Vert$ to estimate its convergence speed.  

Let $\epsilon\leq\frac{-\alpha(A+BU_{*})t_{0}}{(64C)^{2}}$, $l_{0}=1+\lfloor\log_{2}\frac{4}{-\alpha(A+BU_{*})t_{0}}\rfloor$. 
We first inductively show that for any $l\leq l_{0}$, $\left\Vert E_{1}(\epsilon,2^{l}t_{1})\right\Vert \leq (2^{l+3}-4)C^{2}$. The base case where $l=0$ is certainly true. Suppose we already have $\left\Vert E_{1}(\epsilon,2^{l-1}t_{1})\right\Vert \leq (2^{l+2}-4)C^{2}$. Then for the case of $l$, we obtain:

\begin{align*}
\left\Vert E_{1}(\epsilon,2^{l}t_{1})\right\Vert &\leq 2\left\Vert E_{1}(\epsilon,2^{l-1}t_{1})\right\Vert +4C^{2}\leq (2^{l+3}-4)C^{2}\,,
\end{align*}
where for the first inequality we use the inductive hypothesis that 
\begin{align*}
\epsilon\left\Vert E_{1}(\epsilon,2^{l-1}t_{1})\right\Vert\leq 2^{l_{0}+3}C^{2}\epsilon \leq \frac{64}{-\alpha(A+BU_{*})t_{0}}C^{2}\epsilon\leq C\,,
\end{align*}
along with facts that $\left\Vert E(2^{l-1}t_{1})\right\Vert\leq C$ and $2e^{\alpha(A+BU_{*})2^{l-1}t_{1}}\leq 2$.
Specifically, we have \\$\left\Vert E_{1}(\epsilon, 2^{l_{0}}t_{1})\right\Vert \leq \frac{64C^{2}}{-\alpha(A+BU_{*})t_{0}}$.

Now, we consider $l>l_{0}$. We first show that for all such $l$, $\left\Vert E_{1}(\epsilon, 2^{l}t_{1})\right\Vert \leq \frac{64C^{2}}{-\alpha(A+BU_{*})t_{0}}$. Since $2^{l-1}t_{1}\geq 2^{l_{0}-1}t_{0}\geq \frac{2}{-\alpha(A+BU_{*})}$, we have $2e^{\alpha(A+BU_{*})2^{l}t_{1}}\leq 2e^{-2}$, and thus
\begin{align*}
\left\Vert E_{1}(\epsilon,2^{l}t_{1})\right\Vert  &\leq 2e^{\alpha(A+BU_{*})2^{l-1}t_{1}}\left\Vert E_{1}(\epsilon,2^{l-1}t_{1})\right\Vert +2\left\Vert E(2^{l-1}t_{1})\right\Vert ^{2}+2\epsilon^{2}\left\Vert E_{1}(\epsilon,2^{l-1}t_{1})\right\Vert ^{2} \\
&\leq 2e^{-2}\left\Vert E_{1}(\epsilon,2^{l-1}t)\right\Vert  +4C^{2}\\&\leq \frac{64C^{2}}{-\alpha(A+BU_{*})t_{0}}\,,
\end{align*}
which holds for all $l\geq l_{0}$ with induction on $l$. Now we reuse the above expression and obtain that
\begin{align*}
&\left\Vert E_{1}(\epsilon,2^{l}t_{1})\right\Vert  
\\&\leq 2e^{\alpha(A+BU_{*})2^{l-1}t_{1}}\left\Vert E_{1}(\epsilon,2^{l-1}t_{1})\right\Vert +2\left\Vert E(2^{l-1}t_{1})\right\Vert ^{2}+2\epsilon^{2}\left\Vert E_{1}(\epsilon,2^{l-1}t_{1})\right\Vert ^{2} \\
&\leq 2e^{-2^{l-l_{0}}}\frac{64C^{2}}{-\alpha(A+BU_{*})t_{0}}+\frac{128}{\alpha^{2}(A+BU_{*})t^{2}_{0}}e^{-2^{l-l_{0}-1}}+2\epsilon^{2}\left\Vert E_{1}(\epsilon,2^{l-1}t_{1})\right\Vert ^{2}\,.
\end{align*}
Let $l_{*}$ be the smaller integer greater than $l_{0}+1$ which satisfies: 
\begin{align*}
2e^{-2^{l_{*}-l_{0}}}\frac{64C^{2}}{-\alpha(A+BU_{*})t_{0}}+\frac{128}{\alpha^{2}(A+BU_{*})t^{2}_{0}}e^{-2^{l_{*}-l_{0}-1}}\leq \frac{1}{4}\,.
\end{align*}
Then by using the relation that $2\epsilon^{2}\left\Vert E_{1}(\epsilon,2^{l-1}t_{1})\right\Vert ^{2}\leq 2\epsilon^{2}\left(\frac{64C^{2}}{-\alpha(A+BU_{*})t_{0}}\right)^{2}\leq\frac{1}{4}$, we have:
\begin{align*}
    \left\Vert E_{1}(\epsilon,2^{l_{*}}t_{1})\right\Vert\leq \frac{1}{2}\,.
\end{align*}

Now we inductively show that for all $k\geq 0$, 
\begin{align*}
\left\Vert E_{1}(\epsilon,2^{l_{*}+k}t_{1})\right\Vert\leq2^{-2^{k}}\,.
\end{align*}
By using the hypothesis for $k$ and $2\epsilon^{2}\leq\frac{1}{4}$, we obtain:

\begin{align*}
\left\Vert E_{1}(\epsilon,2^{l_{*}+k+1}t_{1})\right\Vert&\leq 2\epsilon^{2}\left\Vert E_{1}(\epsilon,2^{l_{*}+k}t_{1})\right\Vert^{2}+\frac{1}{4}e^{-2^{k+l_{*}-l_{0}}+2^{l_{*}-l_{0}}}\\
&\leq \frac{1}{4}2^{-2^{k+1}}+\frac{1}{4}e^{-2^{k+2}+2^{2}}\\
&\leq 2^{-2^{k+1}}\,,
\end{align*}
leading to the claim. This means there exist some constants $C_{1},c_{1}>0$ depending on $\alpha(A+BU_{*})$ such that for all $t\geq 0$,
$\left\Vert E_{1}(\epsilon,t)\right\Vert\leq C_{1}e^{-c_{1}t}$\,.

Finally, we consider $\int_{t\geq 0}e^{(A+BU)^{\mathrm{T}}t}(Q+U^{\mathrm{T}}RU)e^{(A+BU)t}dt $. Since \\$e^{(A+BU_{*}+\epsilon \Delta U)t}=e^{(A+BU_{*})t}+\epsilon E(t)+\epsilon^{2}E_{1}(\epsilon,t)$, with $\Vert E(t)\Vert \leq \frac{8}{-\alpha(A+BU_{*})t_{0}}e^{\frac{1}{4}\alpha(A+BU_{*})t}$ and bounded $E_{1}(\epsilon,t)$, we obtain:
\begin{align*}
&\int_{t\geq 0}e^{(A+BU)^{\mathrm{T}}t}(Q+U^{\mathrm{T}}RU)e^{(A+BU)t}dt\\
&=\int_{t\geq 0}(e^{(A+BU_{*})^{\mathrm{T}}t}+\epsilon E^{\mathrm{T}}(t)+\epsilon^{2}E_{1}^{\mathrm{T}}(\epsilon,t))(Q+U^{\mathrm{T}}RU)(e^{(A+BU_{*})t}+\epsilon E(t)+\epsilon^{2}E_{1}(\epsilon,t))dt\\
&=\int_{t\geq 0}e^{(A+BU_{*})^{\mathrm{T}}t}(Q+U_{*}^{\mathrm{T}}RU_{*})e^{(A+BU_{*})t}dt\\
&+\epsilon \int_{t\geq 0} E^{\mathrm{T}}(t)(Q+U_{*}^{\mathrm{T}}RU_{*})e^{(A+BU_{*})t}+e^{(A+BU_{*})^{\mathrm{T}}t}(Q+U_{*}^{\mathrm{T}}RU_{*})E(t)dt\\
&+\epsilon \int_{t\geq 0}e^{(A+BU_{*})^{\mathrm{T}}t}\left(\Delta U^{\mathrm{T}}RU_{*}+U^{\mathrm{T}}_{*}R\Delta U\right)e^{(A+BU_{*})t}dt\\
&+\mathcal{O}(\epsilon^{2})\,.
\end{align*}

Where the last term $\mathcal{O}(\epsilon^{2})$ contains any terms with order at least $\epsilon^{2}$, whose norm is at most $C_{2}\epsilon^{2}$ for any $\epsilon\in[0,\epsilon_{0})$ and $\Vert \Delta U\Vert =1$, where the constant $C_{2}$ depends on $A,B,\alpha(A+BU_{*})$ and $\epsilon_{0}$ is some small constant.

For any $\Vert \Delta U\Vert =1$, define $V$ by
\begin{align*}
V&=\int_{t\geq 0} E^{\mathrm{T}}(t)(Q+U_{*}^{\mathrm{T}}RU_{*})e^{(A+BU_{*})t}+e^{(A+BU_{*})^{\mathrm{T}}t}(Q+U_{*}^{\mathrm{T}}RU)E(t)dt\\&+\int_{t\geq 0}e^{(A+BU_{*})^{\mathrm{T}}t}\left(\Delta U^{\mathrm{T}}RU_{*}+U^{\mathrm{T}}R\Delta U\right)e^{(A+BU_{*})t}dt\,,
\end{align*}
then $cost(U)=cost(U_{*})+\epsilon \mathrm{tr}(V)+O(\epsilon^{2})$. 

Since $U_{*}$ minimizes $cost(U)$, $\mathrm{tr}(V)=\lim_{\epsilon\to 0} \epsilon^{-1}(cost(U_{*}+\epsilon\Delta U)-cost(U_{*}))= 0$. Therefore, we obtain that $cost(U)=cost(U_{*})+\mathcal{O}(\epsilon^{2})$.

\end{proof}

\subsection{Proof of Lemma~\ref{lem:approximated optimal cost}}
In this section, we proof Lemma~\ref{lem:approximated optimal cost}.
\approximated*

\begin{proof}
    By definition of $J_{T}$, we have:
\begin{align*}
J_{T}=\mathbb{E}\left[\int_{t=0}^{T}\left(X^{\mathrm{T}}_{t}QX_{t}+U_{t}^{\mathrm{T}}RU_{t}\right)dt\right]=\mathbb{E}\left[\int_{t=0}^{T}X^{\mathrm{T}}_{t}(Q+U^{\mathrm{T}}RU)X_{t}dt\right]\,.
\end{align*}

Since the state transits according to $dX_{t}=AX_{t}dt+BUX_{t}dt+dW_{t}$, we can derive the expression of $X_{t}$ by    $X_{t}=e^{(A+BU)t}X_{0}+\int_{s=0}^{t}e^{(A+BU)(t-s)}dW_{s}$. Then by utilizing this expression we obtain:
\begin{align*}
&\mathbb{E}\left[X^{\mathrm{T}}_{t}(Q+U^{\mathrm{T}}RU)X_{t}\right]\\
&=(e^{(A+BU)t}X_{0})^{\mathrm{T}}(Q+U^{\mathrm{T}}RU)e^{(A+BU)t}X_{0}\\
&+2\mathbb{E}\left[(e^{(A+BU)t}X_{0})^{\mathrm{T}}(Q+U^{\mathrm{T}}RU)\left(\int_{s=0}^{t}e^{(A+BU)(t-s)}dW_{s}\right)\right]\\
&+\mathbb{E}\left[\left(\int_{s=0}^{t}e^{(A+BU)(t-s)}dW_{s}\right)^{\mathrm{T}}(Q+U^{\mathrm{T}}RU)\left(\int_{s=0}^{t}e^{(A+BU)(t-s)}dW_{s}\right)\right]\\
&=X_{0}^{\mathrm{T}}e^{(A+BU)^{\mathrm{T}}t}(Q+U^{\mathrm{T}}RU)e^{(A+BU)t}X_{0}\\
&+tr\left(\int_{s=0}^{t}e^{(A+BU)^{\mathrm{T}}s}(Q+U^{\mathrm{T}}RU)e^{(A+BU)s}ds\right)\\
&=X_{0}^{\mathrm{T}}e^{(A+BU)^{\mathrm{T}}t}(Q+U^{\mathrm{T}}RU)e^{(A+BU)t}X_{0}\\
&+\int_{s=0}^{t}tr\left(e^{(A+BU)^{\mathrm{T}}s}(Q+U^{\mathrm{T}}RU)e^{(A+BU)s}\right)ds\,.
\end{align*}
Then, the expected cost on a trajectory lasting for time $T$ can be computed as:
\begin{align*}
&\mathbb{E}\left[\int_{t=0}^{T}X_t^{\mathrm{T}}(Q+U^{\mathrm{T}}RU)X_{t}dt\right]\\
&=\int_{t=0}^{T}\mathbb{E}\left[X_t^{\mathrm{T}}(Q+U^{\mathrm{T}}RU)X_{t}\right]dt\\
&=\int_{t=0}^{T}X_{0}^{\mathrm{T}}e^{(A+BU)^{\mathrm{T}}t}(Q+U^{\mathrm{T}}RU)e^{(A+BU)t}X_{0}dt\\
&+\int_{t=0}^{T}(T-t)tr\left(e^{(A+BU)^{\mathrm{T}}t}(Q+U^{\mathrm{T}}RU)e^{(A+BU)t}\right)dt\\
&=\int_{t=0}^{T}X_{0}^{\mathrm{T}}e^{(A+BU)^{\mathrm{T}}t}(Q+U^{\mathrm{T}}RU)e^{(A+BU)t}X_{0}dt+cost(U)T\\
&-\int_{t=0}^{T}tr\left(e^{(A+BU)^{\mathrm{T}}t}(Q+U^{\mathrm{T}}RU)e^{(A+BU)t}\right)tdt\\
&-T\int_{t=T}^{+\infty}tr\left(e^{(A+BU)^{\mathrm{T}}t}(Q+U^{\mathrm{T}}RU)e^{(A+BU)t}\right)dt\,.
\end{align*}
Here the first term satisfies
\begin{align*}
\left\vert\int_{t=0}^{T}X_{0}^{\mathrm{T}}e^{(A+BU)^{\mathrm{T}}t}(Q+U^{\mathrm{T}}RU)e^{(A+BU)t}X_{0}dt\right\vert&\leq \int_{t\geq 0}e^{2\alpha(A+BU)t}\left\Vert X_{0}\right\Vert ^{2}_{2}dt\\
&\leq\frac{1}{-2\alpha(A+BU)}\left\Vert X_{0}\right\Vert ^{2}_{2}\,,
\end{align*}
and the latter two integral terms can be bounded as follows.
\begin{align*}
&\left\vert\int_{t=0}^{T}tr\left(e^{(A+BU)^{\mathrm{T}}t}(Q+U^{\mathrm{T}}RU)e^{(A+BU)t}\right)tdt\right\vert\\
&\leq \int_{t\geq 0}d\cdot e^{2\alpha(A+BU)t}\left\Vert Q+U^{\mathrm{T}}RU\right\Vert tdt\\
&\leq \frac{d\left\Vert Q+U^{\mathrm{T}}RU\right\Vert }{4\alpha^{2}(A+BU)}\,,
\end{align*}
\begin{align*}
&\left\vert T\int_{t=T}^{+\infty}tr\left(e^{(A+BU)^{\mathrm{T}}t}(Q+U^{\mathrm{T}}RU)e^{(A+BU)t}\right)dt\right\vert\\
&\leq T\int_{t\geq T}d\cdot e^{2\alpha(A+BU)t}\left\Vert Q+U^{\mathrm{T}}RU\right\Vert dt\\
&\leq \frac{Td\left\Vert Q+U^{\mathrm{T}}RU\right\Vert }{-2\alpha(A+BU)}e^{2\alpha(A+BU)T}\\
&\leq \frac{d\left\Vert Q+U^{\mathrm{T}}RU\right\Vert }{4\alpha^{2}(A+BU)}\,.
\end{align*}
Therefore, for $C_{2}\geq -\frac{1}{2\alpha(A+BU)}$ and $C_{3}\geq \frac{d\left\Vert Q+U^{\mathrm{T}}RU\right\Vert }{2\alpha^{2}(A+BU)}$, we have
\begin{align*}
\left\vert J_{T}-cost(U)T\right\vert \leq C_{2}\Vert x\Vert ^{2}_{2}+C_{3}\,.
\end{align*}
\end{proof}

\subsection{Proof of Lemma~\ref{lem:estimating expected cost}}
\label{app:final}
Finally, we prove Lemma~\ref{lem:estimating expected cost}. In this part we suppose $T\geq T_{0}$, where $T_{0}\geq 1$ is a constant depending on some hidden constants and $\Vert X_{0}\Vert^{2}_{2}$. 

\begin{lemma}{regret}
\label{lem:estimating expected cost}
Let $U_{t}$ be the action applied as in Algorithm~\ref{alg:3}. Then there exists a constant $C\in poly(\kappa,M, \mu^{-1}, \vert\alpha(A+BK)\vert^{-1},  \vert\alpha(A+BK_{*})\vert^{-1})$ such that for sufficiently large $T$:
\begin{align*}
&\mathbb{E}\left[\int_{t=0}^{\sqrt{T}}\left(X^{\mathrm{T}}_{t}QX_{t}+U^{\mathrm{T}}_{t}RU_{t}\right)dt \right]\leq C\cdot \sqrt{T}\,,\\   
&\mathbb{E}\left[\int_{t=\sqrt{T}}^{T}\left(X^{\mathrm{T}}_{t}QX_{t}+U^{\mathrm{T}}_{t}RU_{t}\right)dt \right]\leq C\cdot \sqrt{T}+J^{*}_{T}\,.   
\end{align*}
\end{lemma}

Define the following events where the stabilizing controller $K$ might ever be applied during the exploitation phase. Let $\mathcal{E}_{1}=\left\{\Vert X_{\sqrt{T}}\Vert_{2}\geq \frac{1}{2}T^{1/5}    \right\}$, $\mathcal{E}_{2}=\left\{ \Vert X_{t}\Vert_{2}\geq T^{1/5}\, for\, some\, t\in[\sqrt{T},T]     \right\}$, and
$\mathcal{E}_{3}=\left\{ \Vert \bar{K}-K_{*}\Vert\leq\epsilon_{3}   \right\}$, where $\epsilon_{3}>0$ depends on the constant $\epsilon_{0}$ in Lemma~\ref{lem:convergence of cost}, which will be determined later. In this part, we again let $C_{1},C_{2}$ be the same as in ~\eqref{eq: upper bound K}, and denote $C_{3}$ be the constant $C_{1}$ in Lemma~\ref{lem:convergence of cost}. We firstly analyze these three events.

\paragraph{Upper Bound of $\mathbb{P}[\mathcal{E}_{1}]$}
By Lemma~\ref{lem:bound continuous trajectory}, we can find some constant $C_{0}$ depending on $\|A\|,\|B\|,\|K\|,d,p,h$ such that 

\begin{align*}
\mathbb{P}\left[\Vert X_{\sqrt{T}}\Vert_{2}\geq C_{0}\sqrt{\log(2T/\delta)}     \right]\leq \delta\,.
\end{align*}

This is because we have the recursive function of $\{X_{kh}\}$ that
\begin{align*}
X_{(k+1)h}=e^{(A+BK)h}X_{kh}+\int_{t=0}^{h}e^{(A+BK)(h-t)}dW_{kh+t}+\int_{t=0}^{h}e^{(A+BK)(h-t)}u_{k}dt\,,
\end{align*}
from which we can derive that
\begin{align*}
&X_{kh}\\&=e^{(A+BK)kh}X_{0}+\int_{t=0}^{kh}e^{(A+BK)(kh-t)}dW_{t}+\sum_{i=0}^{k-1}e^{(A+BK)(k-i-1)h}\left(\int_{t=0}^{h}e^{(A+BK)t}dt\right)u_{i}\,.
\end{align*}

Then, for sufficiently large $T$, $\left\Vert e^{(A+BK)\sqrt{T}} X_{0}\right\Vert_{2}$ can be bounded by $1$, and from the proof in Lemma~\ref{lem:bound continuous trajectory} we can apply similar idea to upper bound the norm of the last two terms. So we can obtain the probability bound on $\Vert X_{\sqrt{T}}\Vert_{2}$.

By setting $\delta=2T\cdot e^{-\frac{T^{1/5}}{4C_{0}^{2}}}$, we obtain that $\mathbb{P}[\mathcal{E}_{1}]\leq 2T\cdot e^{-\frac{T^{1/5}}{4C_{0}^{2}}}$.

\paragraph{Upper Bound of $\mathbb{P}[\mathcal{E}^{C}_{3}]$}
By~\eqref{eq: upper bound P}, we obtain that, for $\epsilon_{3}\leq \frac{C_{1}\epsilon_{1}}{T^{1/8}4C^{2}_{2}}$, we have:

\begin{align*}
\mathbb{P}\left[\Vert \bar{K}-K_{*}\Vert\geq x\right]\leq e^{-\frac{T^{1/2}x^{2}}{4C_{1}^{2}C_{2}^{2}}}\,\forall x\leq\epsilon_{3}\,,
\end{align*}
and we also have: $\mathbb{P}[\mathcal{E}^{C}_{3}]\leq e^{-\frac{T^{1/2}\epsilon_{3}^{2}}{4C_{1}^{2}C_{2}^{2}}}$.

By setting $\epsilon_{3}= \frac{C_{1}\epsilon_{1}}{T^{1/8}4C^{2}_{2}}$, we have: $\mathbb{P}[\mathcal{E}^{C}_{3}]\leq e^{-\frac{T^{1/4}\epsilon_{1}^{2}}{64C_{2}^{2}}}$.

\paragraph{Upper Bound of $\mathbb{P}[\mathcal{E}_{2}]$}
Consider any $\Vert X_{\sqrt{T}}\Vert_{2}\leq \frac{1}{2}T^{1/5}$ and any $\Vert \bar{K}-K_{*}\Vert\leq \epsilon_{3}$, we claim that $\mathbb{P}\left[\mathcal{E}_{2}\big\vert X_{\sqrt{T}},\bar{K}\right]\leq e^{-\Omega(T^{1/5})}$.

As what have discussed in Lemma~\ref{lem:convergence of P} (see the discussion about stable margin near ~\eqref{stable margin}), such $\bar{K}$ satisfies $\alpha(A+B\bar{K})\leq \frac{1}{2}\alpha(A+BK_{*})$.

Then by Lemma~\ref{lem:bound continuous trajectory} we can derive that, for some constant $C$, 
\begin{align*}
\mathbb{P}\left[\sup_{t\in[\sqrt{T},T]}\Vert X_{t}\Vert_{2}-\Vert X_{\sqrt{T}}\Vert_{2}\leq \frac{1}{2}T^{1/5}\right]\leq CTe^{-\frac{T^{1/5}}{C}}\leq e^{-\Omega(T^{1/5})}\,.
\end{align*}
Therefore, 
\begin{align*}
\mathbb{P}\left[\mathcal{E}_{2}\right]&\leq 1-\mathbb{P}\left[\mathcal{E}^{C}_{1}\cap \mathcal{E}_{3}\right]+e^{-\Omega(T^{1/5})}\mathbb{P}\left[\mathcal{E}^{C}_{1}\cap \mathcal{E}_{3}\right]\\
&\leq \mathbb{P}\left[\mathcal{E}_{1}\right]+\mathbb{P}\left[\mathcal{E}^{C}_{3}\right]+e^{-\Omega(T^{1/5})}\\
&\leq e^{-\Omega(T^{1/5})}\,.
\end{align*}

Now we come to estimate the expected cost of Algorithm~\ref{alg:3}, as well as bound the regret. We separately calculate the cost during the two phases.

\paragraph{Cost During Exploration Phase} 
For $(k+1)h\leq \sqrt{T}$ and $t\in[0,h]$, we have:
\begin{align*}
&X_{kh+t}=e^{(A+BK)t}X_{kh}+\int_{s=kh}^{kh+t}e^{(A+BK)(kh+t-s)}dW_{s}+\left(\int_{s=0}^{t}e^{(A+BK)s}ds\right)u_{k}\,.
\end{align*}
Then
\begin{align*}
\mathbb{E}\left[X^{\mathrm{T}}_{kh+t}QX_{kh+t}+U^{\mathrm{T}}_{kh+t}RU_{kh+t}\right]
&=\mathbb{E}\left[X^{\mathrm{T}}_{kh+t}(Q+K^{\mathrm{T}}RK)X_{kh+t}+u^{\mathrm{T}}_{k}Ru_{k}\right]+2\mathbb{E}\left[X^{\mathrm{T}}_{kh+t}K^{\mathrm{T}}Ru_{k}\right]\\
&\leq \mathbb{E}\left[X^{\mathrm{T}}_{kh+t}(Q+K^{\mathrm{T}}RK)X_{kh+t}+u^{\mathrm{T}}_{k}Ru_{k}\right]\\
&+2\mathbb{E}\left[u^{\mathrm{T}}_{k}\left(\int_{s=0}^{t}e^{(A+BK)s}ds\right)^{\mathrm{T}}K^{\mathrm{T}}Ru_{k}   \right]\,,
\end{align*}
where the inequality is because $u_{k}$ is independent of $X_{kh}$ and $W_{s}(s\in[kh,kh+t])$.

For the first term, we first upper bound $\mathbb{E}\left[\Vert X_{kh+t}\Vert^{2}_{2} \right]$.

Denote $w_{k,t}=\int_{s=kh}^{kh+t}e^{(A+BK)(kh+t-s)}dW_{s}+\left(\int_{s=0}^{t}e^{(A+BK)s}ds\right)u_{k}$, which is a Gaussian variable with zero mean and is independent of $X_{kh}$. Then 
\begin{align*}
\mathbb{E}\left[\Vert X_{kh+t}\Vert^{2}_{2} \right]&=\mathbb{E}\left[\left\Vert e^{(A+BK)t}X_{kh}+w_{k,t}\right\Vert^{2}_{2} \right]\\
&=\mathbb{E}\left[\left\Vert e^{(A+BK)t} X_{kh}\right\Vert^{2}_{2} \right]+\mathbb{E}\left[\Vert w_{k,t}\Vert^{2}_{2} \right]\\
&\leq \mathbb{E}\left[\left\Vert X_{kh}\right\Vert^{2}_{2} \right]+\mathbb{E}\left[\Vert w_{k,t}\Vert^{2}_{2} \right]\,.
\end{align*}

For $\mathbb{E}\left[\Vert X_{kh}\Vert_{2}^{2}\right]$, since 
\begin{align*}
X_{kh}=e^{(A+BK)h}X_{0}+\int_{t=0}^{kh}e^{(A+BK)(kh-t)}dW_{t}+\sum_{i=0}^{k-1}e^{(A+BK)(k-i-1)h}\left(\int_{t=0}^{h}e^{(A+BK)t}dt\right)u_{i}\,.
\end{align*}
We have
\begin{align*}
\mathbb{E}\left[\Vert X_{kh}\Vert^{2}_{2}  \right]&=\left\Vert e^{(A+BK)kh}X_{0}\right\Vert^{2}_{2}+\mathbb{E}\left[\left\Vert \int_{t=0}^{kh}e^{(A+BK)(kh-t)}dW_{t}  \right\Vert^{2}_{2}\right] \\ &+\sum_{i=0}^{k-1}\mathbb{E}\left[\left\Vert e^{(A+BK)(k-i-1)h}\left(\int_{t=0}^{h}e^{(A+BK)t}dt\right)u_{i} \right\Vert^{2}_{2}\right] \\
&\leq e^{2\alpha(A+BK)\cdot kh}\Vert X_{0}\Vert^{2}_{2}+tr\left(\int_{t=0}^{kh}e^{(A+BK)t}e^{(A+BK)^{\mathrm{T}}t}dt\right)\\
&+\sum_{i=0}^{k-1}tr\left(\left[e^{(A+BK)ih}\left(\int_{t=0}^{h}e^{(A+BK)t}dt\right)\right]\left[e^{(A+BK)ih}\left(\int_{t=0}^{h}e^{(A+BK)t}dt\right)\right]^{\mathrm{T}}\right)\\
\end{align*}

Therefore, we have
\begin{align*}
\mathbb{E}\left[\Vert X_{kh}\Vert^{2}_{2}  \right]
&\leq e^{2\alpha(A+BK)\cdot kh}\Vert X_{0}\Vert^{2}_{2}+\int_{t=0}^{kh}d\cdot e^{2\alpha(A+BK)t}dt+\sum_{i=0}^{k-1}d\cdot e^{2\alpha(A+BK)ih}\cdot h^{2}\\
&\leq e^{2\alpha(A+BK)\cdot kh}\Vert X_{0}\Vert^{2}_{2}+\frac{d}{-2\alpha(A+BK)}+\frac{dh^{2}}{1-e^{2\alpha(A+BK)h}}\\
&\leq C_{3}+e^{2\alpha(A+BK)\cdot kh}\Vert X_{0}\Vert^{2}_{2}\,.
\end{align*}
Where $C_{3}$ is a constant depending on $\alpha(A+BK)$ and $d$.

For the second term $\mathbb{E}\left[\Vert w_{k,t}\Vert^{2}_{2}\right]$, can follow the same process of the above bound and obtain $\mathbb{E}\left[\Vert w_{k,t}\Vert^{2}_{2}\right]\leq C_{3}$.
Therefore, $\mathbb{E}\left[\Vert X_{kh+t}\Vert^{2}_{2}\right]\leq 2C_{3}$.

Now we can upper bound $\mathbb{E}\left[X^{\mathrm{T}}_{kh+t}QX_{kh+t}+U^{\mathrm{T}}_{kh+t}RU_{kh+t}\right]$. 
We have
\begin{align*}
\mathbb{E}\left[X^{\mathrm{T}}_{kh+t}(Q+K^{\mathrm{T}}RK)X_{kh+t}\right]\leq \mathbb{E}\left[\left\Vert Q+K^{\mathrm{T}}RK\right\Vert \Vert X_{kh+t}\Vert^{2}_{2}\right]\leq \left\Vert Q+K^{\mathrm{T}}RK\right\Vert\mathbb{E}\left[\Vert X_{kh+t}\Vert^{2}_{2}\right]\,,
\end{align*}
We also have $\mathbb{E}\left[u^{\mathrm{T}}_{k}Ru_{k}\right]=tr(R)$ and the following inequality:
\begin{align*}
\mathbb{E}\left[u^{\mathrm{T}}_{k}\left(\int_{s=0}^{t}e^{(A+BK)s}ds\right)^{\mathrm{T}}K^{\mathrm{T}}Ru_{k}   \right]\leq (d+p)\cdot \left\Vert \left(\int_{s=0}^{t}e^{(A+BK)s}ds\right)^{\mathrm{T}}K^{\mathrm{T}}R \right\Vert\leq (d+p)h\Vert KR\Vert\,.
\end{align*}
We can conclude that there exists constant $C_{4}$ depending on $A,B,K,Q,R,d,p,h$ such that 
\begin{align*}
\mathbb{E}\left[X^{\mathrm{T}}_{kh+t}QX_{kh+t}+U^{\mathrm{T}}_{kh+t}RU_{kh+t}\right]\leq C_{4}\left(1+e^{2\alpha(A+BK)\cdot (kh+t)}\Vert X_{0}\Vert^{2}_{2}\right)\,, \forall k,t    \,.
\end{align*}
Therefore, the cost during exploration phase can be bounded as
\begin{align}\label{cost during exploration}
\mathbb{E}\left[\int_{t=0}^{\sqrt{T}}\left(X^{\mathrm{T}}_{kh+t}QX_{kh+t}+U^{\mathrm{T}}_{kh+t}RU_{kh+t}\right)dt\right]\leq C_{4}\left(\sqrt{T}+\frac{\Vert X_{0}\Vert^{2}_{2}}{-2\alpha(A+BK)}\right)\,.
\end{align}

\paragraph{Cost During Exploitation Phase}
We first concentrate on $\mathcal{E}_{2}$, which is the hardest event for the analysis of the cost. Consider the following two cases:

\textbf{Case 1:} $\Vert X_{\sqrt{T}}\Vert_{2}\geq T^{1/5}$. In this case, the action is applied by $U_{t}=KX_{t},t\in[\sqrt{T},T]$.

\textbf{Case 2:} $\Vert X_{\sqrt{T}}\Vert_{2}< T^{1/5}$. In this case, the trajectory is unfortunately controlled by a bad controller, and suffers from large risk of diverging.

We first consider \textbf{Case 1}. By~\eqref{transformed form} we can derive that
\begin{align*}
X_{t}=e^{(A+BK)(t-\sqrt{T})}X_{\sqrt{T}}+\int_{s=\sqrt{T}}^{t}e^{(A+BK)(t-s)}dW_{s}\,.
\end{align*}

Then, we have:
\begin{align*}
&\mathbb{E}\left[X^{\mathrm{T}}_{t}QX_{t}+U^{\mathrm{T}}_{t}RU_{t}\right]\\&=\mathbb{E}\left[X^{\mathrm{T}}_{t}(Q+K^{\mathrm{T}}RK)X_{t}\right]\\
&\leq \left\Vert Q+K^{\mathrm{T}}RK\right\Vert \mathbb{E}\left[\Vert X_{t}\Vert^{2}_{2}\right]\\
&\leq \left\Vert Q+K^{\mathrm{T}}RK\right\Vert\left[\Vert X_{\sqrt{T}}\Vert^{2}_{2}+\int_{s=\sqrt{T}}^{t}tr\left(e^{(A+BK)(t-s)}e^{(A+BK)^{\mathrm{T}}(t-s)}\right)dt \right]\\
&\leq \left\Vert Q+K^{\mathrm{T}}RK\right\Vert\left[\Vert X_{\sqrt{T}}\Vert^{2}_{2}+\int_{s=\sqrt{T}}^{t}d\cdot e^{2\alpha(A+BK)(t-s)}dt \right]\,.
\end{align*}

Therefore, for some constants $C_{5}, C_{6}$, we have:
\begin{align*}
\mathbb{E}\left[X^{\mathrm{T}}_{t}QX_{t}+U^{\mathrm{T}}_{t}RU_{t}\right]\leq C_{5}\Vert X_{\sqrt{T}}\Vert^{2}_{2}+C_{6}\,.
\end{align*}

Now we consider \textbf{Case 2}. Let $t_{0}=\inf_{t}\{\|X_{t}\|_{2}\geq T^{1/5}, t\geq \sqrt{T}\}$, then $\Vert X_{t_{0}}\Vert_{2}=T^{1/5}$ almost surely. 

For $t\in[\sqrt{T},t_{0}]$, since we always have
\begin{align*}
\Vert U_{t}\Vert_{2} \leq \max\left\{\Vert K\Vert, \left\Vert R^{-1}B^{\mathrm{T}}P\right\Vert\right\}\Vert X_{t}\Vert_{2}\leq \left(\Vert K\Vert+\left\Vert R^{-1}B^{\mathrm{T}}\right\Vert T^{1/5}\right)T^{1/5}\,,   
\end{align*}
the cost satisfies:
\begin{align*}
X^{\mathrm{T}}_{t}QX_{t}+U^{\mathrm{T}}_{t}RU_{t}\leq C_{7}T^{4/5}\,.
\end{align*}
Where $C_{7}$ is a constant depending on $B,R,K,P$.

For $t\in[t_{0},T]$, the trajectory $X_{t}$ satisfies
\begin{align*}
    X_{t}=e^{(A+BK)(t-t_{0})}X_{t_{0}}+\int_{s=t_{0}}^{t}e^{(A+BK)(t-s)}dW_{s}\,.
\end{align*}
Similar to the analysis for \textbf{Case 1}, we have:
\begin{align*}
\mathbb{E}\left[X^{\mathrm{T}}_{t}QX_{t}+U^{\mathrm{T}}_{t}RU_{t}\right]\leq C_{5}T^{2/5}+C_{6}
\end{align*}

Combining them, we can conclude that for some constant $C_{8}$, no matter whether $\mathcal{E}_{2}$ happens, we always have:
\begin{align*}
\mathbb{E}\left[X^{\mathrm{T}}_{t}QX_{t}+U^{\mathrm{T}}_{t}RU_{t}\right]\leq C_{8}\left[T^{4/5}+\Vert X_{\sqrt{T}}\Vert^{2}_{2}\right]\, \forall t\in[\sqrt{T},T]\,.
\end{align*}

Now we establish the upper bound for the regret. Since
\begin{align*}
1=1_{\mathcal{E}^{C}_{1}\cap \mathcal{E}_{3}}+1_{\mathcal{E}_{1}}+1_{\mathcal{E}^{C}_{1}\cap\mathcal{E}^{C}_{3}}
\end{align*}
Then we can rewrite $\mathbb{E}\left[\int_{t=\sqrt{T}}^{T}\left(X^{\mathrm{T}}_{t}QX_{t}+U^{\mathrm{T}}_{t}RU_{t}\right)dt\right]$ as
\begin{align*}
&\mathbb{E}\left[\int_{t=\sqrt{T}}^{T}\left(X^{\mathrm{T}}_{t}QX_{t}+U^{\mathrm{T}}_{t}RU_{t}\right)dt\right]\\
&=\mathbb{E}\left[\int_{t=\sqrt{T}}^{T}\left(X^{\mathrm{T}}_{t}QX_{t}+U^{\mathrm{T}}_{t}RU_{t}\right)dt\cdot 1_{\mathcal{E}^{C}_{1}\cap \mathcal{E}_{3}}\right]\\
&+\mathbb{E}\left[\int_{t=\sqrt{T}}^{T}\left(X^{\mathrm{T}}_{t}QX_{t}+U^{\mathrm{T}}_{t}RU_{t}\right)dt\cdot 1_{\mathcal{E}_{1}} \right]\\
&+\mathbb{E}\left[\int_{t=\sqrt{T}}^{T}\left(X^{\mathrm{T}}_{t}QX_{t}+U^{\mathrm{T}}_{t}RU_{t}\right)dt\cdot1_{\mathcal{E}^{C}_{1}\cap\mathcal{E}^{C}_{3}} \right]\,.
\end{align*}
For the first term, we can upper bound it by
\begin{align*}
&\mathbb{E}\left[\int_{t=\sqrt{T}}^{T}\left(X^{\mathrm{T}}_{t}QX_{t}+U^{\mathrm{T}}_{t}RU_{t}\right)dt\cdot 1_{\mathcal{E}^{C}_{1}\cap \mathcal{E}_{3}}\right]\\
&\leq \mathbb{E}\left[\int_{t=\sqrt{T}}^{T}\left(X^{\mathrm{T}}_{t}QX_{t}+U^{\mathrm{T}}_{t}RU_{t}\right)dt\cdot 1_{\mathcal{E}^{C}_{1}\cap \mathcal{E}^{C}_{2}\cap\mathcal{E}_{3}}\right]\\
&+\mathbb{E}\left[\int_{t=\sqrt{T}}^{T}\left(X^{\mathrm{T}}_{t}QX_{t}+U^{\mathrm{T}}_{t}RU_{t}\right)dt\cdot 1_{\mathcal{E}^{C}_{1}\cap \mathcal{E}_{2}}\right]\\
&\leq \mathbb{E}\left[\left(cost\left(R^{-1}B^{\mathrm{T}}P\right)T+C_{9}\Vert X_{\sqrt{T}}\Vert^{2}_{2}\right) \cdot 1_{\mathcal{E}^{C}_{1}\cap\mathcal{E}_{3}}\right]
+\mathbb{E}\left[C_{8}\left(T^{4/5}+\Vert X_{\sqrt{T}}\Vert^{2}_{2}\right)\cdot 1_{\mathcal{E}^{C}_{1}\cap \mathcal{E}_{2}}\right]\\
&\leq C_{9}T^{2/5}+cost(R^{-1}B^{\mathrm{T}}P_{*})T+C_{10}T\mathbb{E}\left[\Vert \bar{K}-K_{*}\Vert^{2}\cdot 1_{\mathcal{E}_{3}}\right] +2C_{8}T^{4/5}\cdot \mathbb{E}\left[1_{\mathcal{E}^{C}_{1}\cap \mathcal{E}_{2}}\right]\,.
\end{align*}
Here the first inequality is because $1_{\mathcal{E}^{C}_{1}\cap \mathcal{E}_{3}}=1_{\mathcal{E}^{C}_{1}\cap \mathcal{E}^{C}_{2}\cap\mathcal{E}_{3}}+1_{\mathcal{E}^{C}_{1}\cap \mathcal{E}_{2}\cap\mathcal{E}_{3}}$ and
$1_{\mathcal{E}^{C}_{1}\cap \mathcal{E}_{2}\cap\mathcal{E}_{3}}\leq 1_{\mathcal{E}^{C}_{1}\cap\mathcal{E}_{3}}$. For the second inequality, the first term is because we can assume a situation that we do not change the dynamic when $\mathcal{E}_{2}$ happens, and that will not make the expectation smaller. By applying the results of Lemma~\ref{lem:convergence of cost} and Lemma~\ref{lem:approximated optimal cost} we can get this term, where the constant $C_{9}$ is related to constants in these two lemmas. The last inequality is obtained from these two lemmas and the definitions of $\mathcal{E}_{1},\mathcal{E}_{2},\mathcal{E}_{3}$.

As for $\mathbb{E}\left[\Vert \bar{K}-K_{*}\Vert^{2}\cdot 1_{\mathcal{E}_{3}}\right]$, we use the bound that
\begin{align*}
\mathbb{P}\left[\Vert \bar{K}-K_{*}\Vert\geq x\right]\leq e^{-\frac{T^{1/2}x^{2}}{4C_{1}^{2}C_{2}^{2}}}\,\forall x\leq\epsilon_{3}\,,
\end{align*}
and compute that
\begin{align*}
&\mathbb{E}\left[\Vert \bar{K}-K_{*}\Vert^{2}\cdot 1_{\mathcal{E}_{3}}\right]\\
&\leq \int_{x=0}^{\epsilon_{3}^{2}}\mathbb{P}\left[\Vert \bar{K}-K_{*}\Vert^{2}\geq x\right]\cdot dx\\
&\leq \int_{x\geq 0}e^{-\frac{T^{1/2}x}{4C_{1}^{2}C_{2}^{2}}}dx\\
&=\frac{4C_{1}^{2}C^{2}_{2}}{T^{1/2}}\,.
\end{align*}
For $\mathbb{E}\left[1_{\mathcal{E}^{C}_{1}\cap \mathcal{E}_{2}}\right]$, we directly have $\mathbb{E}\left[1_{\mathcal{E}^{C}_{1}\cap \mathcal{E}_{2}}\right]\leq \mathbb{P}\left[\mathcal{E}_{2}\right]\leq e^{-\Omega(T^{1/5})}$.
Combining these results and Lemma~\ref{lem:approximated optimal cost} we obtain that for some constant $C$,
\begin{align*}
\mathbb{E}\left[\int_{t=\sqrt{T}}^{T}\left(X^{\mathrm{T}}_{t}QX_{t}+U^{\mathrm{T}}_{t}RU_{t}\right)dt\cdot 1_{\mathcal{E}^{C}_{1}\cap \mathcal{E}_{3}}\right]\leq J_{\theta_{*},T}+C\sqrt{T}\,.
\end{align*}

For the second term $\mathbb{E}\left[\int_{t=\sqrt{T}}^{T}\left(X^{\mathrm{T}}_{t}QX_{t}+U^{\mathrm{T}}_{t}RU_{t}\right)dt\cdot 1_{\mathcal{E}_{1}} \right]$, given any $X_{\sqrt{T}}$, we always have
\begin{align*}
\mathbb{E}\left[X^{\mathrm{T}}_{t}QX_{t}+U^{\mathrm{T}}_{t}RU_{t}\right]\leq C_{8}\left[T^{4/5}+\Vert X_{\sqrt{T}}\Vert^{2}_{2}\right]\, \forall t\in[\sqrt{T},T]\,.
\end{align*}
So we can upper bound $\mathbb{E}\left[\int_{t=\sqrt{T}}^{T}\left(X^{\mathrm{T}}_{t}QX_{t}+U^{\mathrm{T}}_{t}RU_{t}\right)dt\cdot 1_{\mathcal{E}_{1}} \right]$ by
\begin{align*}
&\mathbb{E}\left[\int_{t=\sqrt{T}}^{T}\left(X^{\mathrm{T}}_{t}QX_{t}+U^{\mathrm{T}}_{t}RU_{t}\right)dt\cdot 1_{\mathcal{E}_{1}} \right]\\
&\leq C_{8}T^{9/5}\mathbb{P}[\mathcal{E}_{1}]+C_{8}T\mathbb{E}\left[\Vert X_{\sqrt{T}}\Vert^{2}_{2}\cdot 1_{\mathcal{E}_{1}}\right]\\
&\leq O(1)+C_{8}T\mathbb{E}\left[\Vert X_{\sqrt{T}}\Vert^{2}_{2}\cdot 1_{\mathcal{E}_{1}}\right]\,,
\end{align*}
where for the last inequality we apply the upper bound of $\mathbb{P}[\mathcal{E}_{1}]$ shown before.

For $\mathbb{E}\left[\Vert X_{\sqrt{T}}\Vert^{2}_{2}\cdot 1_{\mathcal{E}_{1}}\right]$, we can apply Lemma~\ref{lem:bound continuous trajectory} and obtain that for some constant $c>0$, for any $x\geq\frac{1}{2}T^{1/5}$, we have
\begin{align*}
\mathbb{P}\left[\Vert X_{\sqrt{T}}\Vert_{2}\geq x\right]\leq e^{-cx^{2}}\,.
\end{align*}
Thus we have:
\begin{align*}
&T\mathbb{E}\left[\Vert X_{\sqrt{T}}\Vert^{2}_{2}\cdot 1_{\mathcal{E}_{1}}\right]\\
&\leq\frac{1}{4}T^{7/5}\mathbb{P}\left[\Vert X_{\sqrt{T}}\Vert_{2}\geq\frac{1}{2}T^{1/5}\right] +T\int_{x\geq \frac{1}{4}T^{2/5}}\mathbb{P}\left[\Vert X_{\sqrt{T}}\Vert^{2}_{2}\geq x\right]dx\\
&\leq \mathcal{O}(1)\,.
\end{align*}
Therefore, we have $\mathbb{E}\left[\int_{t=\sqrt{T}}^{T}\left(X^{\mathrm{T}}_{t}QX_{t}+U^{\mathrm{T}}_{t}RU_{t}\right)dt\cdot 1_{\mathcal{E}_{1}} \right]\leq \mathcal{O}(1)$

Finally, for the last term $\mathbb{E}\left[\int_{t=\sqrt{T}}^{T}\left(X^{\mathrm{T}}_{t}QX_{t}+U^{\mathrm{T}}_{t}RU_{t}\right)dt\cdot1_{\mathcal{E}^{C}_{1}\cap\mathcal{E}^{C}_{3}} \right]$, when condition on any $\Vert X_{\sqrt{T}}\Vert_{2}\leq \frac{1}{2}T^{1/5}$, estimator $(\hat{A},\hat{B})$ and $X_{t_{0}}$, where $t_{0}=\inf_{t\geq \sqrt{T}}(\Vert X_{t}\Vert_{2}\geq T^{1/5})$, we still have:
\begin{align*}
\mathbb{E}\left[X^{\mathrm{T}}_{t}QX_{t}+U^{\mathrm{T}}_{t}RU_{t}\right]\leq C_{8}\left[T^{4/5}+\Vert X_{\sqrt{T}}\Vert^{2}_{2}\right]\leq 2C_{8}T^{4/5}\,, \forall t\in[\sqrt{T},T]\,.
\end{align*}
So we can upper bound it by
\begin{align*}
&\mathbb{E}\left[\int_{t=\sqrt{T}}^{T}\left(X^{\mathrm{T}}_{t}QX_{t}+U^{\mathrm{T}}_{t}RU_{t}\right)dt\cdot1_{\mathcal{E}^{C}_{1}\cap\mathcal{E}^{C}_{3}}\right]\\
&\leq 2C_{8}T^{9/5}\mathbb{P}\left[\mathcal{E}^{C}_{1}\cap\mathcal{E}^{C}_{3}\right]\\
&\leq 2C_{8}T^{9/5}\mathbb{P}[\mathcal{E}^{C}_{3}]\\
&\leq \mathcal{O}(1)\,.
\end{align*}

Combining them we finally obtain Lemma~\ref{lem:estimating expected cost}.

\newpage
\section*{NeurIPS Paper Checklist}

\begin{enumerate}

\item {\bf Claims}
    \item[] Question: Do the main claims made in the abstract and introduction accurately reflect the paper's contributions and scope?
    \item[] Answer: \answerYes{} 
    \item[] Justification: We clarify our contributions and basic problem setups in both abstract and introduction.

\item {\bf Limitations}
    \item[] Question: Does the paper discuss the limitations of the work performed by the authors?
    \item[] Answer: \answerYes{} 
    \item[] Justification: We discuss the limitations of the work in Section 6.

\item {\bf Theory assumptions and proofs}
    \item[] Question: For each theoretical result, does the paper provide the full set of assumptions and a complete (and correct) proof?
    \item[] Answer: \answerYes{} 
    \item[] Justification: Assumptions can be found just near the main theorems. The complete proof is contained in our Appendix.

    \item {\bf Experimental result reproducibility}
    \item[] Question: Does the paper fully disclose all the information needed to reproduce the main experimental results of the paper to the extent that it affects the main claims and/or conclusions of the paper (regardless of whether the code and data are provided or not)?
    \item[] Answer: \answerYes{} 
    \item[] Justification: We disclose the experiment details in Section 5.3.

\item {\bf Open access to data and code}
    \item[] Question: Does the paper provide open access to the data and code, with sufficient instructions to faithfully reproduce the main experimental results, as described in supplemental material?
    \item[] Answer: \answerNo{} 
    \item[] Justification: Our code is very simple, just use a simulation experiment of 3*3 matrix. Our main contribution is the theoretical analysis.

\item {\bf Experimental setting/details}
    \item[] Question: Does the paper specify all the training and test details (e.g., data splits, hyperparameters, how they were chosen, type of optimizer, etc.) necessary to understand the results?
    \item[] Answer: \answerYes{} 
    \item[] Justification: We specify all the training and test details in Section 5.3.

\item {\bf Experiment statistical significance}
    \item[] Question: Does the paper report error bars suitably and correctly defined or other appropriate information about the statistical significance of the experiments?
    \item[] Answer: \answerYes{} 
    \item[] Justification: We compute the average regret in our experiment.

\item {\bf Experiments compute resources}
    \item[] Question: For each experiment, does the paper provide sufficient information on the computer resources (type of compute workers, memory, time of execution) needed to reproduce the experiments?
    \item[] Answer: \answerNo{} 
    \item[] Justification: Our code is very simple, which can run on CPU of a computer.
    
\item {\bf Code of ethics}
    \item[] Question: Does the research conducted in the paper conform, in every respect, with the NeurIPS Code of Ethics \url{https://neurips.cc/public/EthicsGuidelines}?
    \item[] Answer: \answerYes{} 
    \item[] Justification: We conform with the NeurIPS Code of Ethics.

\item {\bf Broader impacts}
    \item[] Question: Does the paper discuss both potential positive societal impacts and negative societal impacts of the work performed?
    \item[] Answer: \answerNA{} 
    \item[] Justification: : Our work is about the theory on online control and system identification, which does not seem to have
evident societal impacts.
    
\item {\bf Safeguards}
    \item[] Question: Does the paper describe safeguards that have been put in place for responsible release of data or models that have a high risk for misuse (e.g., pretrained language models, image generators, or scraped datasets)?
    \item[] Answer: \answerNA{} 
    \item[] Justification: The paper poses no such risks.

\item {\bf Licenses for existing assets}
    \item[] Question: Are the creators or original owners of assets (e.g., code, data, models), used in the paper, properly credited and are the license and terms of use explicitly mentioned and properly respected?
    \item[] Answer: \answerNA{} 
    \item[] Justification: The paper does not use existing assets.

\item {\bf New assets}
    \item[] Question: Are new assets introduced in the paper well documented and is the documentation provided alongside the assets?
    \item[] Answer: \answerNA{} 
    \item[] Justification: The paper does not release new assets.

\item {\bf Crowdsourcing and research with human subjects}
    \item[] Question: For crowdsourcing experiments and research with human subjects, does the paper include the full text of instructions given to participants and screenshots, if applicable, as well as details about compensation (if any)? 
    \item[] Answer: \answerNA{} 
    \item[] Justification: The paper does not involve crowdsourcing nor research with human subjects.

\item {\bf Institutional review board (IRB) approvals or equivalent for research with human subjects}
    \item[] Question: Does the paper describe potential risks incurred by study participants, whether such risks were disclosed to the subjects, and whether Institutional Review Board (IRB) approvals (or an equivalent approval/review based on the requirements of your country or institution) were obtained?
    \item[] Answer: \answerNA{} 
    \item[] Justification: The paper does not involve crowdsourcing nor research with human subjects.

\item {\bf Declaration of LLM usage}
    \item[] Question: Does the paper describe the usage of LLMs if it is an important, original, or non-standard component of the core methods in this research? Note that if the LLM is used only for writing, editing, or formatting purposes and does not impact the core methodology, scientific rigorousness, or originality of the research, declaration is not required.
    \item[] Answer: \answerNo{} 
    \item[] Justification: The LLM is used only for editing the paper of grammar mistake.

\end{enumerate}

\end{document}